\documentclass[12pt,a4paper,reqno]{article}

\usepackage{amssymb,amsmath,amsfonts,geometry,graphicx,caption,color,setspace,natbib,array,amsthm,tikz}

\usetikzlibrary{decorations.pathreplacing}
\usepackage[T1]{fontenc}
\usepackage{newpxmath}
\usepackage{newpxtext}
\usetikzlibrary{patterns}
\usepackage{hyperref}

\usepackage[titletoc]{appendix}

\usepackage[immediate]{silence}

\allowdisplaybreaks

\newtheorem{theorem}{Theorem}[]

\theoremstyle{definition}

\newtheorem{assumption}{Assumption}
\newtheorem{proposition}{Proposition}
\newtheorem{lemma}{Lemma}
\newtheorem{remark}{Remark}

\newcommand{\mR}{\mathbb{R}}
\newcommand{\mE}{\mathbb{E}}

\usetikzlibrary{calc}
\usetikzlibrary{positioning}
\usetikzlibrary{matrix,fit}

\makeatletter
\let\@makefntextOring\@makefntext
\def\@makefntext#1{\@makefntextOring{\baselineskip=15pt#1}}
\makeatother

\definecolor{navyblue}{rgb}{0.0, 0.0, 0.5}
\definecolor{green}{rgb}{0.2, 0.5, 0.2}

\hypersetup{
    colorlinks=true,
    citecolor=navyblue,
    linkcolor=navyblue,
    urlcolor=navyblue,
}

\def\citeapos#1{\citeauthor{#1}\textcolor{navyblue}{'s} (\citeyear{#1})}

\usepackage{etoolbox}
\patchcmd{\thanks}{#1}{\protect\doublespacing}{}{}

\DeclareFontShape{OT1}{cmr}{m}{n}{<->cmr10}{}
\usepackage{fix-cm}

\usepackage{titlesec}

\begin{document}

\begin{titlepage}
\title{Feasible Search Behavior\thanks{\protect%\doublespacing 
Initial Draft: October 2024. A preliminary result for this paper appears in the subsumed paper \citet{sato2023information}. 
We thank 
Ian Ball, Michael Choi, Laura Doval, Makoto Hanazono, Nima Haghpanah, WenYu Hsieh, Ilwoo Hwang, Ryota Iijima, Shinsuke Kambe, Michihiro Kandori, Satoshi Kasamatsu, Kohei Kawamura, Daiki Kishishita, Masanori Kobayashi, Andrew Koh, Fuhito Kojima, Noritaka Kudoh, Yichuan Lou, Daniel Luo, Eric Maskin, Akihiko Matsui, Hitoshi Matsushima, Akira Matsushita, Masaki Miyashita, Stephen Morris, Satoshi Nakada, Shunya Noda, Daisuke Oyama, Hitoshi Sadakane, Kota Saito, Kazuki Sekiya, Takashi Shimizu, Wataru Tamura, Alex Wolitzky, Jidong Zhou, seminar participants at University of Tokyo, Tokyo University of Science, the 2023 Japanese Economic Association Autumn Meeting, MIT Theory Lunch, Nagoya University, 10th Annual Communication Theory Workshop, and Econometric Society World Congress 2025. 
Sato acknowledges the financial support from the JSPS KAKENHI Grant 24KJ0100.
Shirakawa acknowledges the financial support from the Funai Foundation for Information Technology.
All remaining errors are our own.
}}
\author{
\Large Hiroto Sato\thanks{Nagoya University (Email: \url{sato.hiroto.s9@f.mail.nagoya-u.ac.jp}).}
\and 
\Large Ryo Shirakawa\thanks{Massachusetts Institute of Technology (Email: \url{shira723@mit.edu}).}
} 
\date{\today}
\maketitle
\begin{abstract}
% Consider a situation wherein an agent sequentially searches heterogeneous options in an arbitrary order. The agent partially learns the value of an option upon inspection. Information structures jointly determine the ex-ante and ex-post value of investigating each option, thereby shaping the entire learning path. We characterize all search behaviors compatible with some information structure, which form a polytope. A single information structure rationalizes all these search behaviors and minimizes the agent's welfare. Under a certain symmetry assumption, we also examine feasible pairs of search and choice behaviors, with and without price competition among option sellers, and prove the same results. 

Consider a situation wherein a decision maker sequentially searches for the best alternative among heterogeneous options with an arbitrary search order. The agent partially learns the value of an option when inspecting it. The information structure jointly determines the ex-ante and ex-post value of investigating each option, thereby shaping the entire learning path. We characterize the set of all search behaviors compatible with some information structure, which forms a polytope. A single information structure rationalizes all these search behaviors, which minimizes the agent's welfare among all information structures. Under certain symmetry assumption over primitives, we also examine the set of all feasible pairs of search and choice behaviors, with and without option sellers' price competition, and prove the same results. Our study not only provides a simple framework for analyzing how information shapes ordered learning paths, but also offers sharp insights into information design and identification problems in such environments. 

\textbf{Keywords:} Search, information acquisition. 
%agent search, sequential search, Bayesian persuasion, online intermediary, sponsored advertising
\vspace{0in}
%\noindent
%\textbf{JEL Codes:} \\
%\bigskip
\end{abstract}
\setcounter{page}{0}
\thispagestyle{empty}
\end{titlepage}

%%%%%%%%%%%%%%%%%%%%%%%%%%%%%%%%%%%%%%%%%%%%%%%%%%%%%%%%%%%%%%%%%%%%%%%%%%%%%%%%%%%%%%%%%%%%%%%%%%%%%%%%
% We can remove table of contents by eliminating these command.
%\pagenumbering{roman}
%\tableofcontents
%%%%%%%%%%%%%%%%%%%%%%%%%%%%%%%%%%%%%%%%%%%%%%%%%%%%%%%%%%%%%%%%%%%%%%%%%%%%%%%%%%%%%%%%%%%%%%%%%%%%%%%%
%\pagebreak \newpage
%\setcounter{page}{0}
%\pagenumbering{arabic} % numbering is moving back to arabic

\newcolumntype{C}[1]{>{\centering\arraybackslash}p{#1}}

% Now, we admit double spacing in the entire of this paper.
% \doublespacing
\onehalfspacing

% We can adjust margins of math-environment here 
%===================================================
\setlength{\abovedisplayskip}{5pt}
\setlength{\belowdisplayskip}{5pt}
%==================================================

%%%%%%%%%%%%%%%%%%%%%%%%%%%%%%%%%%%%%%%%%%%%%%%%%%%%%%%%%%%%%%%%%%%%%%%%%%%%%%%%%%%%%%%%%%%%%%%%%%%%%%%%
\section{Introduction} \label{sec:introduction}
%%%%%%%%%%%%%%%%%%%%%%%%%%%%%%%%%%%%%%%%%%%%%%%%%%%%%%%%%%%%%%%%%%%%%%%%%%%%%%%%%%%%%%%%%%%%%%%%%%%%%%%%

% Information economics
Recent work in information economics has studied how information structures affect agents’ decisions in various environments. Information structures are often modeled as directly shaping agents’ beliefs, and the literature analyzes how the resulting beliefs affect their eventual decisions. However, it is also natural to think that information structures first shape the \textit{learning path}: how much information decision makers acquire about multiple options, and in what order, before affecting their eventual beliefs and decisions. This paper provides a simple framework for analyzing this dual role of information structures and establishes benchmark results in this setting. 

% Model
For analyzing this interaction, our model consists of a decision maker who ultimately chooses one option among multiple heterogeneous options with unknown values. The decision maker sequentially searches these options in an arbitrary order, paying a cost for each inspection. The information structure in this market determines the precision of the information generated when the decision maker inspects each option. Theoretically, our model can be interpreted as a minimal extension of the ordered-search model of \citet{weitzman1979optimal}, known as \textit{Pandora's problem}, in which costly inspection does not fully reveal the value of an option. 

% Dual role of information
The precision of the information acquired through the search process plays two roles. First, as emphasized in the Bayesian persuasion literature, it naturally affects the posterior beliefs that the decision maker forms about each option. Second, it changes the ex-ante value of inspecting each option, and therefore affects the learning path itself: which options the decision maker inspects and in what order. Thus, information structures jointly determine the ex-ante and ex-post values of inspecting each option, and thereby shape the decision maker's overall search behavior. To capture this latter role, we allow information structures to be arbitrarily asymmetric across options. Consequently, even when the values of the options are symmetric in terms of their prior distributions, the decision maker may exhibit asymmetric search behavior depending on the precision of the information obtained through inspection. 

% Theorem 1 & Proposition 1: Characterization
Theorem \ref{thm: feasible search behavior} summarizes our main findings. First, we characterize the entire set of feasible search behaviors, that is, the probabilities with which each option is explored, rationalized by some information structure. Since an information structure may change the order in which the decision maker learns about the options in a discrete way, it is not a priori obvious whether this set has a tractable structure; nevertheless, we show that it is a polytope and provide a closed-form representation of each of its vertices. Second, the strongest version of the reduction principle holds: a single information structure rationalizes all feasible search behaviors. Under this information structure, infinitely many optimal search strategies exist, and every search order is optimal. Finally, this information structure minimizes the agent's welfare among all information structures. This result suggests that our model not only captures the interaction between information structures and the ordered learning path in perhaps the simplest possible way, but also serves as a tractable tool for analyzing it. 

% Implication 1: Information design 
As an application, we consider an information design problem \`a la \citet{kamenica2011bayesian}. 
Here, we have in mind an online platform, such as Amazon or Yelp, that earns advertising revenue when consumers visit product or business pages.\footnote{For example, Amazon reports advertising services through sponsored ads, display ads, and video advertising, and states that revenue is recognized as ads are delivered based on clicks or impressions. Yelp reports that its revenue consists primarily of advertising placements through performance-based cost-per-click advertising.} 
Such platforms can control, to some extent, the information disclosed on these pages by designing product-page layouts, restricting sellers' descriptions, and choosing which customer reviews to highlight. 
For example, Amazon limits product titles to certain number of characters, allows sellers to describe key product features in up to five bullet points, and recently, uses generative AI to summarize customer reviews in a short paragraph. 
These designs affect consumers' ex-post decisions; however, as emphasized above, they also affect the ex-ante decisions of how consumers visit product pages, and hence the revenue generated from pay-per-click fees.\footnote{One may also consider platforms that earn revenue not only from user visits but also from sales commissions. In Section \ref{sec: general state space}, we consider the information design problem of such platforms and obtain the same implication.} Therefore, it is natural to study the extent to which these designs may influence advertising revenue. 

% Implication 1 cont'd
Theorem \ref{thm: feasible search behavior} has a direct implication for the information design problems of such online platforms: for any objective function defined over search behaviors, there always exists an optimal information structure that minimizes consumers' welfare to the same level as if the designer provided no information. 
In particular, the optimal information structure does not depend on the shape of the objective function; moreover, in the case of a linear objective function on search behavior, the information designer can robustly implement the optimal search behavior, that is, by suitably perturbing the optimal information structure, she can make the designer-optimal search behavior uniquely optimal for the consumer. 

% Implication 2: Identifiaction problem
Another immediate application is to partial identification problems. 
Although many empirical works build on search models and estimate unobservable elements such as search costs and consumer preferences, these studies assume complete information, namely consumers fully observe the realizations of the underlying distributions when inspecting options. 
As we can infer from our results however, both these estimates and their implications are highly sensitive to that assumption. 
Because our closed-form characterization of feasible search behaviors maps directly to the set of primitive parameters consistent with any observed search behavior, it offers a tractable way to derive the identified set under minimal informational assumptions. 

% Extension 1. Choice
In Section \ref{sec: general state space}, we extend our main theorem in two directions.
First, note that we define a search behavior as the vector of probabilities with which a decision maker inspects each option. 
This notion, however, omits her eventual consumption choice.
Theorem \ref{thm: joint distribution} therefore analyzes feasible choice behavior, namely the probability that each option is ultimately selected.
The implications of Theorem \ref{thm: feasible search behavior} persist: the set of all feasible pairs is still a polytope, and a single welfare-minimizing information structure rationalizes every feasible pair. 
In Proposition \ref{prop: joint distribution vertex}, we derive all vertexes in closed form.
For example, although the preceding paragraph suggests that an online platform maximizing pay-per-click revenue can reduce consumer welfare, Theorem \ref{thm: joint distribution} confirms that this conclusion persists even when the platform also earns sales commissions for each time a consumer purchases an item displayed on the platform. 

% Extension 2. Price competition
Second, since each option in the model often corresponds to its seller, it is natural to incorporate price competition among these sellers into the model. 
In ordered search markets, prices affect consumers' search order, making sellers' payoffs discontinuous and equilibrium price dispersion, and hence equilibrium search strategies, remain an open research agenda. 
Our approach bypasses this issue: Theorem \ref{thm: price competition} proves that the set of equilibrium search-choice pairs does not change at all whether sellers compete on price. 
The same information structure supports all equilibrium pairs. 
Welfare implications, however, diverge once prices become endogenous to information. Under our information structure, sellers optimally post zero prices and earn no surplus, whereas consumers may benefit because lower prices may offset the loss in information quality. 

% Extention: Final remark
These two extensions are discussed under a weak symmetry assumption. 
The assumption is imposed over primitive parameters, and we still allow information structures to be heterogeneous. 
Although we assume in the model section that the reward of each box is binary, we also find that Theorems \ref{thm: joint distribution} and \ref{thm: price competition} do not rely on this assumption. 
In particular, the statements do not change at all for any general priors as long as the symmetry assumption is satisfied. 

%%%%%%%%%%%%%%%%%%%%%%%%%%%%%%%%%%%%%%%%%%%%%%%%%%%%%%%%%%%%%%%%%%%%%%%%%%%%%%%%%%%%%%%%%%%%%%%%%%%%%%%%
%%%%%%%%%%%%%%%%%%%%%%%%%%%%%%%%%%%%%%%%%%%%%%%%%%%%%%%%%%%%%%%%%%%%%%%%%%%%%%%%%%%%%%%%%%%%%%%%%%%%%%%% 
\subsection{Related Literature}

% Sequential experimentation
In the sense that the decision maker chooses her own learning path, our study is conceptually related to the literature on dynamic information acquisition; see, e.g., \citet{wald1947foundations}, \citet{arrow1949bayes}, \citet{moscarini2001optimal}, and \citet{morris2019wald}. More closely related, recent studies in this literature, such as \citet{steiner2017rational}, \citet{fudenberg2018speed}, \citet{che2019optimal}, \citet{liang2022dynamically}, \citet{zhong2022optimal}, \citet{mayskaya2024following}, and \citet{adusumilli2026sample}, analyze the optimal dynamic allocation of limited attention across multiple information sources. In contrast to these models, where the same information source can be sampled repeatedly, each source in our model can be sampled only once, so the learning path is an irreversible sequence of distinct experiments. This feature makes it relatively tractable to analyze how information structures jointly shape the ordered learning path and the eventual decision, while exploring the same question in these alternative settings would also be interesting. 

% Search and information
More directly, this paper is best positioned in the literature on search theory and information economics, which we discuss below.

% Robust approach
Broadly speaking, this paper is related to the literature on informationally robust predictions. 
Since the seminal works by \citet{bergemann2013robust} and \citet{bergemann2016bayes}, several studies have provided information-free set predictions for the outcomes of various Bayesian models. 
Among others, \citet{bergemann2021search} study a random search setting as in \citet{stahl1989oligopolistic} wherein firms have incomplete information over consumers' search intensities and develop a belief-free upper-bound for the equilibrium price distributions. 
We study the incompleteness of information on the consumer side, hence complementing their work. 
To the best of our knowledge, neither our theorem nor its proof has an analogue in existing studies. 

% Information economics and search markets
More specifically, this paper belongs to the recently growing literature at the intersection of information economics and search models. 
\citet{zhou2020improved} is one of the studies related to our work. 
For example, in a sequential random search setup as in \citet{wolinsky1986true} and \citet{anderson1999pricing}, his study finds a certain law on the effect of an information improvement on the equilibrium search duration. 
By contrast, when search is directed, we find that a single information structure is compatible with any search behavior, indicating that generally we have no robust connection between information structures and search lengths. 

% Random search 
Particularly active in the above literature are the applications of information design towards random search models. 
\citet{dogan2022consumer} and \citet{hu2021industry} derive consumer-optimal and producer-optimal information structures. 
\citet{board2018competitive}, \citet{whitmeyer2020persuasion}, \citet{he2023competitive}, analyze firms competitively disclosing information and observe that firms may seize monopoly profit; \citet{whitmeyer2020persuasion} terms this phenomenon the informational \citeapos{diamond1971model} paradox.\footnote{See also \citet{choi2019optimal}, \citet{lyu2023information}, and \citet{mekonnen2025efficient}, which study the designs of pre-search information/market segmentation in search markets.} 
Recent studies by \citet{hwang2025competitive} and \citet{boleslavsky2025limits} extend these settings and consider heterogeneous search frictions. 
By contrast, \citet{mekonnen2025persuaded} consider a principal who designs and sells information at each period and show that the principal sells a socially efficient level of information in stationary equilibrium.\footnote{For example, he/she sells fully-informative signals under binary priors. That is, he/she chooses to sell informative signals with high price rather than selling uninformative signals with low price and prolonging search. A recent study by \citet{mekonnen2025competition} also analyzes competition among multiple principals.} 
We contribute to this body of work by solving a general class of information design problems in the ordered search setting.

% Ordered search: Heterogeneity
Few papers investigate the role of information in \citeapos{weitzman1979optimal} ordered search. 
\citet{au2023attraction} and \citet{au2024attraction} study competitive information disclosure of symmetric firms in the ordered search setting and characterize the symmetric equilibria of these games, under which firms may randomize over information structures. 
Their studies are different from our work particularly in that they look at the eventual choices of consumers while our Theorem \ref{thm: feasible search behavior} focuses on their search processes. 
Accordingly, we cannot rely on their methodology of applying \citet{armstrong2017ordered} and \citet{choi2019optimal}, which give a clean characterization of each product's demand. 
Nevertheless, in proving Theorems \ref{thm: joint distribution} and \ref{thm: price competition}, which characterize the feasible pairs of search and choice behaviors, we also employ this technique to derive bounds on the set of feasible choice behaviors. 
The technical challenge of our study also stems from allowing for general asymmetry in the distributions across options. This asymmetry is not only a key feature of \citeapos{weitzman1979optimal} model, but also essential for generally exploring the dual role of information structures discussed in the Introduction. 

% Identification
Finally, our study also contributes to the literature on identification problems in search markets. 
Prior empirical work has developed methods to identify unobservables in search models, assuming that decision makers perfectly observe the realizations of their draws from value distributions.
See, for example, \citet{jolivet2019consumer} and \citet{moraga2023consumer} for recent articles and \citet{ursu2025sequential} for an overview. 
For other economic models, however, a growing literature studies partial identification under minimal assumptions about the underlying information structure; see, for example, \citet{gualdani2024identification}, \citet{dickstein2018exporters}, \citet{dickstein2024patient}, and \citet{porcher2024measuring} for some empirical strategies and their applications, and \citet{bergemann2022counterfactuals}, \citet{doval2024revealed} and \citet{de2025rationalizing} for theory. 
Our results allow us to take this approach and yield a tractable method to find an identified set.  

%%%%%%%%%%%%%%%%%%%%%%%%%%%%%%%%%%%%%%%%%%%%%%%%%%%%%%%%%%%%%%%%%%%%%%%%%%%%%%
The remainder of this paper is organized as follows.
Section \ref{sec:model} describes the model. 
We present our main theorem and discuss the implications in Section \ref{sec:results}. 
In Section \ref{sec: general state space}, we extend our analysis. 
Section \ref{sec:conclusion} concludes this study.
All proofs are provided in Appendices \ref{sec:appendixa}, \ref{sec:appendixb}, \ref{sec:appendixc}, and \ref{sec:appendixd}. 
%%%%%%%%%%%%%%%%%%%%%%%%%%%%%%%%%%%%%%%%%%%%%%%%%%%%%%%%%%%%%%%%%%%%%%%%%%%%%%

%%%%%%%%%%%%%%%%%%%%%%%%%%%%%%%%%%%%%%%%%%%%%%%%%%%%%%%%%%%%%%%%%%%%%%%%%%%%%%
%%%%%%%%%%%%%%%%%%%%%%%%%%%%%%%%%%%%%%%%%%%%%%%%%%%%%%%%%%%%%%%%%%%%%%%%%%%%%%
\section{Model} \label{sec:model}
%%%%%%%%%%%%%%%%%%%%%%%%%%%%%%%%%%%%%%%%%%%%%%%%%%%%%%%%%%%%%%%%%%%%%%%%%%%%%%
%%%%%%%%%%%%%%%%%%%%%%%%%%%%%%%%%%%%%%%%%%%%%%%%%%%%%%%%%%%%%%%%%%%%%%%%%%%%%%

% Model Elements: Weitzman 
% boxes, Pandora, Payoff
The model is based on \citeapos{weitzman1979optimal} ordered search, which is referred to as \textit{Pandora's problem}.
There are $n\geq 2$ \textit{boxes} and a risk-neutral agent, \textit{Pandora}, with unit demand. 
Each box $i$ contains a random \textit{reward} $x_{i}$, which is either $\underline{x}$ or $\overline{x}$. 
Normalize that $0\leq \underline{x} < \overline{x}\leq 1$.
Let $\lambda_{i}\in(0,1)$ be the prior probability that the reward is $\overline{x}$. 
Denote by $\mu_{i}=\lambda_{i}\overline{x}+(1-\lambda_{i})\underline{x}$ the prior mean.
Pandora obtains $x_{i}$ when choosing box $i$.
We discuss how our result can be extended beyond binary state space in Section \ref{sec: general state space}.

% Information
If Pandora opens box $i$, she observes a signal $s_{i}$ according to an information structure, which is a random variable whose realization depends on $x_{i}$.\footnote{Formally, an information structure for box $i$ is a mapping from the state space $\{\underline{x},\overline{x}\}$ into the space of distributions $\Delta(S_{i})$ with some signal space $S_{i}$. 
We assume that signals are independent across boxes; otherwise, we cannot even derive Pandora's optimal search strategy. 
See, for example, \citet{ke2020informational} and \citet{bao2022search} for a study on Pandora's problem with correlation.} 
Upon observing $s_{i}$, she forms expectation $y_{i}=\mE[x_{i}|s_{i}]$.
It turns out that $y_{i}$ is the only payoff-relevant information in our model, and therefore, we define an information structure for box $i$ directly as the induced CDF $F_{i}\in\Delta([\underline{x},\overline{x}])$ of posterior means $y_{i}$. 
\citet{blackwell1953equivalent} shows that some information structure induces $F_{i}\in\Delta([\underline{x},\overline{x}])$ as a distribution of posterior means if and only if $F_{i}$ has mean $\mu_{i}$.\footnote{This characterization relies on the assumption that state space for each box is binary. See Section \ref{sec: general state space} and Appendix \ref{sec:appendixc} for the feasibility condition under general prior distributions.}
Call such $F_{i}$ \textit{feasible}.

% Pandora's strategy 
Fixing a CDF profile $F=(F_{1},\dots,F_{n})$, Pandora's strategy is a sequential decision rule that determines, at each stage, whether to continue searching and, if so, which box to open next.
Each box can be opened only once.\footnote{See, for example, \citet{ke2019optimal} for a model that relaxes this assumption.}
When she stops searching, she can choose one of any previously opened boxes.\footnote{Unlike in \citeapos{doval2018whether} model, Pandora cannot choose a box without opening it.}  
It costs $c_{i}>0$ to open each box $i$. 
We assume that search costs are small enough to ensure that each box is opened with some positive probability under some feasible CDF profile.
Specifically, assume $c_{i}<\lambda_{i}(\overline{x}-\underline{x})$ for each $i$.\footnote{This is just for exposition; the inequality is eventually implied by our Assumption \ref{assum: regularity}.} 

% Search behavior
A \textit{search behavior} is a vector $v\in \mR^{n}$ that describes each probability $v_{i}$ that each box $i$ is opened. 
Call a search behavior $v$ \textit{feasible} if there exists a feasible CDF profile which induces it along with an optimal (possibly mixed) strategy. 
We mainly explore the structure of the set of all feasible search behaviors, which we denote by $V^{*}$.
In Section \ref{sec: general state space}, we also investigate Pandora's choice behavior and show that the same conclusions hold for the set of all feasible pairs of search and choice behaviors. 

\begin{remark}\label{rem: interpretation}
    One may think of $F_{i}$ as the distribution of actual rewards $x_{i}$ instead of the distribution of posterior means $y_{i}$. 
    Under this interpretation, Pandora observes true reward $x_{i}$ when opening box $i$ without facing any uncertainty. 
    Then, our research question turns to derive all feasible search behaviors varying distributions with a fixed profile of average values $\mu_{1},\dots,\mu_{n}$.
\end{remark}

%%%%%%%%%%%%%%%%%%%%%%%%%%%%%%%%%%%%%%%%%%%%%%%%%%%%%%%%%%%%%%%%%%%%%%%%%%%%%%
\subsection{Pandora's Rule} \label{sec:prep}

% Weitzman: reservation values
If we fix a CDF profile $F$, our model is mathematically equivalent to \citeapos{weitzman1979optimal} original model, wherein he characterizes the optimal search strategy.  
Define the \textit{reservation value} of box $i$ under distribution $F_{i}$ as the value $z_{i}$ that solves the following equation:% \footnote{The last equation follows from integration by parts. See, e.g.,  \citet{dogan2022consumer}.} 
\begin{align*}
    c_{i} = \int_{z_{i}}^{\overline{x}} (y_{i}-z_{i})dF_{i}(y_{i}) = \overline{x}-z_{i}-\int_{z_{i}}^{\overline{x}}F_{i}(y_{i})dy_{i}.
\end{align*}
The solution to the above equation is unique.  
The assumption $c_{i}<\lambda_{i}(\overline{x}-\underline{x})$ ensures $\underline{z}_{i}\leq z_{i}\leq \overline{z}_{i}$, where $\underline{z}_{i}=\mu_{i}-c_{i}$ and $\overline{z}_{i}=\overline{x}-c_{i}/\lambda_{i}$.\footnote{One can see that $\underline{z}_{i}$ and $\overline{z}_{i}$ correspond to the reservation values under no information and full information, respectively. 
Namely, $z_{i}=\underline{z}_{i}$ if $F_{i}$ is degenerate at prior mean $\mu_{i}$ and $z_{i}=\overline{z}_{i}$ if $F_{i}$ is equal to the prior distribution.} 
One can also show the converse: For any $z_{i}\in [\underline{z}_{i},\overline{z}_{i}]$, there exists a feasible CDF $F_{i}$ with the reservation value $z_{i}$.

% Pandora's rule
Pandora's optimal strategy, \textit{Pandora's rule}, is as follows: 
\begin{itemize}
\vspace{-0.25em}
    \item[$\square$] {\rm\scshape Selection Rule:}
    If a box is to be opened, it should be an unopened box with the highest reservation value.
    \vspace{-0.5em}
    \item[$\square$] {\rm\scshape Stopping Rule:}
    Terminate search when the maximum sampled expected reward exceeds the reservation values of unopened boxes; in this case, choose an opened box with the highest expected reward.
\vspace{-0.25em}
\end{itemize}

% Remark
Any strategy consistent with Pandora's rule is optimal and multiple optimal strategies may exist. 
For example, when the reservation values of all boxes are the same, any search order is optimal. 
Likewise, Pandora is indifferent between stopping and continuing search at any point where the maximum sampled value is exactly equal to the highest reservation value of the remaining boxes.

% Seach Behavior: Illustration
To illustrate how search behavior can be represented in terms of a distribution profile, consider a case where Pandora opens boxes in the ascending order of indices $i$ under a CDF profile $F$. 
Then, Pandora's rule claims that the reservation values under $F$ must satisfy $z_{1}\geq z_{2}\geq \cdots\geq z_{n}$. 
Moreover, each box $i>1$ is opened with some probability $v_{i}$ such that 
\begin{align*}
    F_{1}(z_{i}-)\cdot F_{2}(z_{i}-)\cdot \cdots\cdot F_{i-1}(z_{i}-)
    \leq v_{i}
    \leq F_{1}(z_{i})\cdot F_{2}(z_{i})\cdot \cdots\cdot F_{i-1}(z_{i}), 
\end{align*}
where $F_{i}(y_{i}-)$ is the left limit of $F_{i}$ at $y_{i}$. 
Note that $F_{j}(z_{i}-)\neq F_{j}(z_{i})$ means that $F_{j}$ has an atom at point $z_{i}$, in which case Pandora would be indifferent between opening box $i$ and not with some positive probability, and there exist multiple search behaviors compatible with $F$. 
Moreover, $F$ can also rationalize multiple search behaviors if $z_{i}=z_{i+1}$ for some $i$, in which case multiple search orders are optimal.

% Remark: insight
As discussed in the Introduction, the CDF profile, and hence the information structure, affects both the ex-ante and ex-post values of opening a box. The ex-ante value of opening each box is summarized by its reservation value, which shapes the ordered path through which the agent learns the ex-post values of the available options. These ex-post values, in turn, shape the agent's stopping strategy. One implication of our model is that, although the values of the options are statistically independent, the information structure of one option can affect how much information the decision maker acquires about the other options through its effect on the learning path. We view our model as one of the simplest models for analyzing how this dual structure of information determines the agent's eventual learning path, which is one of the main economic contributions of this paper. 
 
%%%%%%%%%%%%%%%%%%%%%%%%%%%%%%%%%%%%%%%%%%%%%%%%%%%%%%%%%%%%%%%%%%%%%%%%%%%%%%
%%%%%%%%%%%%%%%%%%%%%%%%%%%%%%%%%%%%%%%%%%%%%%%%%%%%%%%%%%%%%%%%%%%%%%%%%%%%%%
\section{Feasible Search Behaviors} \label{sec:results}

% Assumption
We impose one parametric assumption throughout. 
Assumption \ref{assum: regularity} essentially requires that the ranges of reservation values $[\underline{z}_{i},\overline{z}_{i}]$ have a non-empty intersection and any search order can be optimal under some information structure.
\begin{assumption} \label{assum: regularity}
    $\overline{z}\geq \underline{z}$ where $\underline{z}=\max_{i}\underline{z}_{i}$ and $\overline{z}=\min_{i}\overline{z}_{i}$.
\end{assumption}
Equivalently, the assumption holds if and only if $\overline{z}_{i}\geq \underline{z}_{j}$ for any $i$ and $j$. 
Since $\overline{z}_{i}=\overline{x}-c_{i}/\lambda_{i}$ and $\underline{z}_{j}=\mu_{j}-c_{j}$, it is satisfied if heterogeneity across boxes is not too large. 
In particular, Assumption \ref{assum: regularity} holds under the symmetric case where $\lambda_{i}=\lambda$ and $c_{i}=c$ for each $i$.
The reason for imposing this assumption will be clear from the statement of Theorem \ref{thm: feasible search behavior}, which will be elaborated on later in this section.

% Math prep
Before presenting our main theorem, let us introduce a few notations.
An order $\pi:\{1,\dots,n\}\rightarrow\{1,\dots,n\}$ is a permutation over the set of boxes. 
Then, for each order $\pi$ and index $i$, define 
\begin{align*}
    \overline{v}_{\pi,i} &= \prod_{j<i} p_{\pi(j)}, 
    \quad
    p_{i}=1-\frac{c_{i}}{\overline{x}-\underline{z}}, \\
    \underline{v}_{\pi,i} &= \prod_{j<i} q_{\pi(j)}, 
    \quad
    q_{i}=1-\frac{\underline{z}_{i}-\underline{x}}{\underline{z}-\underline{x}},
\end{align*}
for each $i$. 
Intuitively, it will turn out that $p_{i}$ is roughly an upper bound for the probability that Pandora continues searching after opening box $i$.\footnote{To be more precise, Lemma \ref{lem: point max} in Appendix \ref{sec:appendixa} implies that $p_{i}=\max_{F_{i}:z_{i}\geq\underline{z}}F_{i}(\underline{z})$, where $z_{i}$ is the reservation value of a feasible distribution $F_{i}$. In other words, $p_{i}$ is the upper bound for the probability of continuing the search, if the reservation values of all later boxes equal $\underline{z}$. In general, however, the probability may be strictly higher than $p_{i}$, depending on the reservation values of the remaining boxes. Lemma \ref{lem: lower bound} shows an analogous interpretation for $q_{i}$.} 
% Hence, given that Pandora opens boxes with order $\pi$, an upper bound for $v_{\pi(i)}$ is $\overline{v}_{\pi,i}$. 
Hence, given that Pandora opens boxes with order $\pi$, $\overline{v}_{\pi,i}$ is the highest possible value of $v_{\pi(i)}$ once the opening probabilities of all preceding boxes have been maximized successively in that order.
One may analogously interpret $q_{i}$ and $\underline{v}_{\pi,i}$ as lower bounds. 
To simplify the notation, we write $\overline{v}_{i}=\overline{v}_{\pi,i}$ and $\underline{v}_{i}=\underline{v}_{\pi,i}$ when $\pi(i)=i$ for each $i$. 

% Main theorem introduction
The main theorem is stated as follows.
We say that a point $v\in V$ is a \textit{vertex} of a set $V\subset \mR^{n}$ if there exists a vector $r$ such that $r\cdot v> r\cdot u$ for all $u\in V\setminus\{v\}$, where $r\cdot v=\sum_{i=1}^{n}r_{i}v_{i}$.
Call a set $V$ a \textit{polytope} if it is a convex hull of its finite vertexes. 

\begin{theorem}\label{thm: feasible search behavior}
    The set of all feasible search behaviors $V^{*}$ is a polytope. 
    Any vector $v$ with $v_{1}\geq v_{2}\geq\cdots\geq v_{n}$ is a vertex of $V^{*}$ if and only if $v=(\overline{v}_{1},\dots,\overline{v}_{i},\underline{v}_{i+1},\dots,\underline{v}_{n})$ for some $i$.
    Moreover, a single distribution profile $\overline{F}$ induces all feasible search behaviors, under which the reservation value of every box equals $\underline{z}$ and Pandora's expected payoff is minimized among all feasible distribution profiles, which also equals $\underline{z}$. 
\end{theorem} 

% Illustration 
In the following, we elaborate on the statement. 

% Vertex
The former part of the theorem argues that $V^{*}$ is a polytope. 
Moreover, all vertexes are derived in closed form; that is, we explicitly derive the set of all feasible search behaviors.
Intuitively, each vertex is characterized by a search order $\pi$ and a box that represents a threshold $i$.\footnote{Without loss of generality, the statement of Theorem \ref{thm: feasible search behavior} focuses on vertexes for which $\pi$ is the identity. Any vertex induced by a general search order $\pi$ can be obtained by replacing each index $j$ with $\pi(j)$, each $\overline{v}_{j}$ with $\overline{v}_{\pi,j}$, and each $\underline{v}_{j}$ with $\underline{v}_{\pi,j}$ in the statement.} 
The search behavior that induces the vertex is such that Pandora opens boxes in order $\pi$ with the highest probability until she reaches the threshold box $i$ and with the lowest probability for the remaining boxes. 

% Illustration
Figure \ref{fig: feasible search behaviors} illustrates $V^{*}\subset [0,1]^{n}$ when $n=3$.\footnote{To produce Figure \ref{fig: feasible search behaviors}, we set $\overline{x}=1$, and $\underline{x}=0$. The other parameter values are given by $\mu_{1}=0.55$, $\mu_{2}=0.45$, $\mu_{3}=0.4$, and $c_{1}=0.3$, $c_{2}=0.1$, $c_{3}=0.2$.} 
Note that $v_{1}+v_{2}+v_{3}\geq 1$ for all feasible search behaviors as there exists no outside option and Pandora opens at least one box for sure. 
One may also notice that box $2$ is always opened with some positive probability. 
In general, the number of such boxes is at most one: Under the distribution profile that is degenerate at prior mean, it is an optimal strategy for Pandora to open box $i$ with the highest $\underline{z}_{i}=\mu_{i}-c_{i}$ and terminate search immediately afterwards, without opening any other boxes.

\begin{figure}[ht]
\begin{center}
\begin{tikzpicture}[scale=0.85, transform shape]

    % Define such that underline{z}=underline{z}_{Z}
    \pgfmathsetmacro{\muX}{0.55} %0.8
    \pgfmathsetmacro{\muY}{0.45} %0.5
    \pgfmathsetmacro{\muZ}{0.40} %0.6
    \pgfmathsetmacro{\cX}{0.30} %0.3
    \pgfmathsetmacro{\cY}{0.10} %0.1
    \pgfmathsetmacro{\cZ}{0.20} %0.2

    \pgfmathsetmacro{\zmin}{max(\muX-\cX,\muY-\cY,\muZ-\cZ)}

    \pgfmathsetmacro{\pX}{1-\cX/(1-\zmin)}
    \pgfmathsetmacro{\pY}{1-\cY/(1-\zmin)}
    \pgfmathsetmacro{\pZ}{1-\cZ/(1-\zmin)}
    \pgfmathsetmacro{\qX}{1-(\muX-\cX)/\zmin}
    \pgfmathsetmacro{\qY}{1-(\muY-\cY)/\zmin}
    \pgfmathsetmacro{\qZ}{1-(\muZ-\cZ)/\zmin}

    \coordinate (O) at (0,0,0);
    \coordinate (X) at (4,0,0); 
    \coordinate (Y) at (0,4,0); 
    \coordinate (Z) at (0,0,4); 

    \coordinate (XY) at (4,4,0);
    \coordinate (YZ) at (0,4,4);
    \coordinate (XZ) at (4,0,4);
    \coordinate (XYZ) at (4,4,4);

    % Cube
    \draw (O) -- (X) -- (XY) -- (Y) -- cycle;
    \draw (O) -- (Z);
    \draw (X) -- (XZ);
    \draw (Y) -- (YZ);
    \draw (XY) -- (XYZ);
    \draw (Z) -- (XZ) -- (XYZ) -- (YZ) -- cycle; 

    % Extreme Points
    \coordinate (XYZ0) at (4,\pX*4,\pX*\pY*4); \filldraw (XYZ0) circle (2pt);
    \coordinate (XYZ1) at (4,\pX*4,\qX*\qY*4); \filldraw (XYZ1) circle (2pt);
    \coordinate (XYZ2) at (4,\qX*4,\qX*\qY*4); \filldraw (XYZ2) circle (2pt);
    \coordinate (XZY0) at (4,\pX*\pZ*4,\pX*4); \filldraw (XZY0) circle (2pt);
    \coordinate (XZY1) at (4,\qX*\qZ*4,\pX*4); \filldraw (XZY1) circle (2pt);
    \coordinate (XZY2) at (4,\qX*\qZ*4,\qX*4); \filldraw (XZY2) circle (2pt);
    \coordinate (YXZ0) at (\pY*4,4,\pY*\pX*4); \filldraw (YXZ0) circle (2pt);
    \coordinate (YXZ1) at (\pY*4,4,\qY*\qX*4); \filldraw (YXZ1) circle (2pt);
    \coordinate (YXZ2) at (\qY*4,4,\qY*\qX*4); \filldraw (YXZ2) circle (2pt);
    \coordinate (YZX0) at (\pY*\pZ*4,4,\pY*4); \filldraw (YZX0) circle (2pt);
    \coordinate (YZX1) at (\qY*\qZ*4,4,\pY*4); \filldraw (YZX1) circle (2pt);
    \coordinate (YZX2) at (\qY*\qZ*4,4,\qY*4); \filldraw (YZX2) circle (2pt);
    \coordinate (ZXY0) at (\pZ*4,\pX*\pZ*4,4); \filldraw (ZXY0) circle (2pt);
    \coordinate (ZXY1) at (\pZ*4,\qX*\qZ*4,4); \filldraw (ZXY1) circle (2pt);
    \coordinate (ZXY2) at (\qZ*4,\qX*\qZ*4,4); \filldraw (ZXY2) circle (2pt);
    \coordinate (ZYX0) at (\pY*\pZ*4,\pZ*4,4); \filldraw (ZYX0) circle (2pt);
    \coordinate (ZYX1) at (\qY*\qZ*4,\pZ*4,4); \filldraw (ZYX1) circle (2pt);
    \coordinate (ZYX2) at (\qY*\qZ*4,\qZ*4,4); \filldraw (ZYX2) circle (2pt);

    % Feasible Search Behaviors
    \filldraw[fill=blue, opacity=0.4] (XYZ0) -- (YXZ0) -- (YZX0) -- (ZYX0) -- (ZXY0) -- (XZY0) -- cycle; 
    \filldraw[fill=blue, opacity=0.0] (XYZ1) -- (YXZ1) -- (YZX1) -- (ZYX1) -- (ZXY1) -- (XZY1) -- cycle; 
    \filldraw[fill=gray, opacity=0.1] (XYZ2) -- (YXZ2) -- (YZX2) -- (ZYX2) -- (ZXY2) -- (XZY2) -- cycle; 
    
    \filldraw[fill=blue, opacity=0.25] (XYZ0) -- (YXZ0) -- (YXZ1) -- (XYZ1) -- cycle;
    \filldraw[fill=blue, opacity=0.3] (YXZ0) -- (YZX0) -- (YZX1) -- (YZX2) -- (YXZ2) -- (YXZ1) -- cycle;
    \filldraw[fill=blue, opacity=0.45] (YZX0) -- (ZYX0) -- (ZYX1) -- (YZX1) -- cycle;
    \filldraw[fill=blue, opacity=0.5] (ZYX0) -- (ZXY0) -- (ZXY1) -- (ZXY2) -- (ZYX2) -- (ZYX1) -- cycle;
    \filldraw[fill=blue, opacity=0.375] (ZXY0) -- (XZY0) -- (XZY1) -- (ZXY1) -- cycle;
    \filldraw[fill=blue, opacity=0.3] (XZY0) -- (XYZ0) -- (XYZ1) -- (XYZ2) -- (XZY2) -- (XZY1) -- cycle;
    \filldraw[fill=blue, opacity=0] (XYZ1) -- (YXZ1) -- (YXZ2) -- (XYZ2) -- cycle;
    \filldraw[fill=blue, opacity=0] (YZX1) -- (ZYX1) -- (ZYX2) -- (YZX2) -- cycle;
    \filldraw[fill=blue, opacity=0] (ZXY1) -- (XZY1) -- (XZY2) -- (ZXY2) -- cycle;

    % Coordinates
    \draw[->][dashed] (0,0,0) -- (6,0,0) node[right]{$v_{1}$};
    \draw[->][dashed] (0,0,0) -- (0,6,0) node[left]{$v_{2}$};
    \draw[->][dashed] (0,0,0) -- (0,0,6) node[left]{$v_{3}$};
\end{tikzpicture}
\end{center}
\caption{The set of feasible search behaviors $V^{*}$.}
\label{fig: feasible search behaviors}
\end{figure}
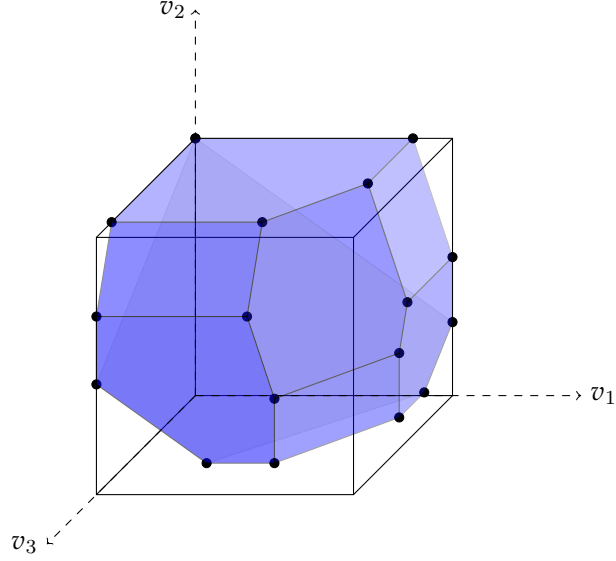

% Convexity
The convexity of $V^{*}$ is not a trivial consequence. This is especially because the search order is a discrete object and is therefore discontinuous in information structures. To make the point more concrete, consider two feasible search behaviors $v$ and $v'$ associated with any feasible CDF profiles $F$ and $F'$. Then the convex combination $v''=\alpha v+(1-\alpha) v'$ would generally satisfy $v''_{i}<1$ for every $i$. Thus, to implement $v''$, Pandora must use a mixed strategy, which requires constructing a CDF profile $F''$ under which the reservation values of multiple boxes coincide. However, at this point, we do not have enough information about the reservation values under $F$ and $F'$ other than their ordinal orders among options. Therefore, a naive convex combination of these CDF profiles $F$ and $F'$ need not have this property. 

% Face
We can also characterize each face of $V^{*}$ by a sub-group of those vertexes. 
Here, we say that a set $X\subset V$ is a \textit{face} of $V$ if $X=\arg\max_{v\in V}r\cdot v$ for some $r\in\mR^{n}\setminus\{0\}$. 
Remember from the previous discussion that each vertex is pinned down by a search order $\pi$ and a threshold $i$. 
In the next proposition, we sort the boxes for any given weight $r$ and show that the vertexes corresponding to the resulting orders are the ones that maximize $r\cdot v$.

\begin{proposition}\label{prop: face}
    Take any $r\in\mR^{n}\setminus\{0\}$. 
    Let $\Pi_{r}$ be the set of orders $\pi$ such that 
    \begin{align*}
    r_{\pi(1)}\cdot \frac{\overline{x}-\underline{z}}{c_{\pi(1)}}
    \geq
    \cdots 
    &\geq 
    r_{\pi(i)}\cdot \frac{\overline{x}-\underline{z}}{c_{\pi(i)}} \\
    &\geq 
    0 \\
    &>
    r_{\pi(i+1)}\cdot \frac{\underline{z}-\underline{x}}{\underline{z}_{\pi(i+1)}-\underline{x}}
    \geq
    \cdots
    \geq 
    r_{\pi(n)}\cdot \frac{\underline{z}-\underline{x}}{\underline{z}_{\pi(n)}-\underline{x}}.    
    \end{align*}
    Then, a face $X=\arg\max_{v\in V^{*}}r\cdot v$ is the convex hull of all points $v$ such that, for some $\pi\in \Pi_{r}$, we have $v_{\pi(j)}=\overline{v}_{\pi,j}$ for all $j\leq i$ and $v_{\pi(j)}=\underline{v}_{\pi,j}$ for all $j>i$.
\end{proposition}

% Implications
Proposition \ref{prop: face} characterizes the set of vertexes maximizing a weighted search length for any non-zero weight. 
Considering weights $r$ and $-r$, we can obtain a robust prediction for the range of feasible search lengths weighted by $r$.
For another application, if we take $r$ such that $r_{i}=1$ for some $i$ and $r_{j}=0$ for all $j\neq i$, Proposition \ref{prop: face} gives us the face that characterizes the set of all feasible search behaviors such that Pandora opens box $i$ first. 

% Distribution
The latter statement of Theorem \ref{thm: feasible search behavior} argues that we have the strongest form of reduction principle, i.e., one can focus on a single welfare-minimizing distribution profile $\overline{F}$ to derive all feasible search behaviors. 
The convexity of the set $V^{*}$ is obtained as a corollary of this result, rather than being proved directly. 
Figure \ref{fig: dist} illustrates the distribution $\overline{F}_{i}$.
It is a step function and the realizations of posterior means can take one of three values, $\underline{x}$, $\underline{z}$, and $\overline{x}$. 

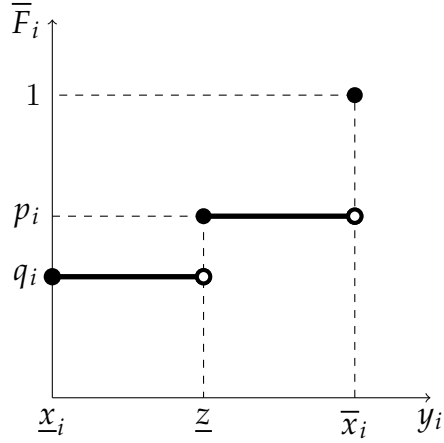
\begin{figure}[ht]
\centering
\begin{tikzpicture}[scale=1, transform shape]

    % Define such that underline{z}=underline{z}_{Z}
    \pgfmathsetmacro{\mu}{0.5}
    \pgfmathsetmacro{\c}{0.2}

    \pgfmathsetmacro{\zmin}{max(\mu-\c,0.5)}
    \pgfmathsetmacro{\p}{1-\c/(1-\zmin)}
    \pgfmathsetmacro{\q}{1-(\mu-\c)/\zmin}

    \coordinate (O) at (0,4) node[below]{$\underline{x}$};
    \coordinate (X) at (4,0); \filldraw (X) node[below]{$\overline{x}$};
    \coordinate (Y) at (0,4); \filldraw (Y) node[left]{$1$};
    \coordinate (XY) at (4,4); \filldraw[line width=2pt] (XY) circle (2pt);

    \coordinate (Up) at (4*\zmin,4*\p); \filldraw[line width=2pt] (Up) circle (2pt);
    \coordinate (UpRight) at (4,4*\p); \draw[line width=1.5pt] (UpRight) circle (2.5pt);
    \coordinate (UpLeft) at (0,4*\p); 
    \coordinate (Down) at (4*\zmin,4*\q); \draw[line width=1.5pt] (Down) circle (2.5pt);
    \coordinate (DownLeft) at (0,4*\q); \filldraw[line width=1.5pt] (DownLeft) circle (2.5pt); 
    \coordinate (JumpPoint) at (4*\zmin,0); 

    % Cube
    \draw[dashed] (X) -- (XY) -- (Y); 

    % Distribution
    \draw[line width=2pt] (Up) -- (4-0.1,4*\p); 
    \draw[dashed] (Up) -- (UpLeft) node[left]{$p_{i}$};
    \draw[line width=2pt] (4*\zmin-0.1,4*\q) -- (DownLeft) node[left]{$q_{i}$}; 
    \draw[dashed] (Up) -- (JumpPoint) node[below]{$\underline{z}$};

    % Coordinates
    \draw[->] (0,0) -- (5,0) node[below]{$y_{i}$};
    \draw[->] (0,0) -- (0,5) node[left]{$\overline{F}_{i}$};

    % Points
    \draw[fill=white, line width=1.5pt] (UpRight) circle (2.5pt);
    \draw[fill=white, line width=1.5pt] (Down) circle (2.5pt);
\end{tikzpicture}
\caption{The distribution of posterior means $\overline{F}_{i}$.}
\label{fig: dist}
\end{figure}

% Dist explanations
The construction is based on the following intuition.
In order for a single CDF profile to rationalize all vertexes, Pandora must be indifferent among all search orders.
Consistent with this, the following equation shows that the reservation value of every box equals $\underline{z}$: 
\begin{align*}
    \overline{x}- \underline{z} - \int_{\underline{z}}^{\overline{x}}\overline{F}_{i}(y_{i})dy_{i}
    = (\overline{x} - \underline{z}) - (\overline{x}-\underline{z})p_{i}
    = (\overline{x} - \underline{z})(1-p_{i})
    = c_{i}.
\end{align*}
The welfare minimization comes from this feature.\footnote{Specifically, Lemma \ref{lem: welfare min} in Appendix \ref{sec:appendixa} shows that any CDF profile $F$ with the reservation values being lower than $\underline{z}$ minimizes Pandora's welfare.}
Moreover, $\overline{F}_{i}$ is constructed to simultaneously maximize $F_{i}(\underline{z})$ and minimize $F_{i}(\underline{z}-)$. 
Roughly speaking, this means that we identify a single CDF profile that can both maximize and minimize the probability of continuing the search, depending on how Pandora breaks ties. % given that the reservation values of the remaining boxes equal $\underline{z}$
This flexibility allows us to implement any feasible search behavior with the single CDF profile, by selecting an appropriate search strategy consistent with Pandora's rule.

% Assumption 1
Assumption \ref{assum: regularity} is to guarantee $0\leq q_{i}\leq p_{i}\leq 1$, ensuring that each $\overline{F}_{i}$ is a non-decreasing function and therefore a well-defined CDF. 
See Appendix \ref{sec:appendixa} for further details. 
Conversely, if every $\overline{F}_{i}$ is well-defined, the definition of $\overline{z}_{i}$ implies that $\overline{z}_{i}\geq \underline{z}$ for all boxes $i$, thereby satisfying Assumption \ref{assum: regularity}.
Thus, the assumption is both necessary and sufficient for our construction $\overline{F}$.

% Proof strategy
The proof of Theorem \ref{thm: feasible search behavior} is in Appendix \ref{sec:appendixa}, which consists of three steps. 
The first step is the most involved, where we solve an auxiliary problem of choosing a CDF profile to maximize a search duration $r\cdot v$ with any non-negative weights $r$. 
Intuitively, we solve it by adjusting the reservation values, and show that making $z_{i}$ closer to $\underline{z}$ increases the search length. 
Hence, consistent with Theorem \ref{thm: feasible search behavior}, the optimum satisfies $z_{i}=\underline{z}$.
This adjustment generally involves a trade‑off, however: when $z_{i}<\underline{z}$, increasing $z_{i}$ raises the probability that box $i$ is opened, increasing the search length, but it also makes box $i$ more attractive; Pandora is then more likely to draw a high realization and terminate the search, reducing the search length. 
The key part of the proof is to show that the first effect dominates. 
A challenge is that each $v_{i}$ is discontinuous in CDF profiles, as it also depends on the search order. Although a naive approach would then analyze this trade‑off within each set of CDF profiles with a common search order, it turned out that this yields an ambiguous relation between the two effects.\footnote{For instance, if $r_{i}=0$, the first effect disappears. Hence, if the weight on the next-opened box is positive, generally the second effect dominates, and perturbing $F$ to raise $z_{i}$ shortens the weighted search length. 
Intuitively, Lemma \ref{lem: joint monotone} in Appendix \ref{sec:appendixa} shows that the first effect dominates the second one when the boxes are ordered consistently with the optimal search behavior.} 
Therefore, we jointly optimize over search orders and the CDF profiles.
In the second step, we use this result to construct a polytope that gives a bound for the set of all feasible search behaviors $V^{*}$ and characterize its vertexes. 
The final step is a heuristic procedure of constructing a welfare-minimizing profile $\overline{F}$ and verifying that it induces all of these vertexes, hence all points in the set. 

% Subsections introduced
In the remainder of this section, we discuss three applications of the main theorem. 
A literal reading of Theorem \ref{thm: feasible search behavior} would suggest that it characterizes the information-free prediction set in a simple environment where information structures also affect the order of learning about heterogeneous options. 
However, the two statements in Theorem \ref{thm: feasible search behavior} also yield separate implications for information design and for identification problems, respectively.

%%%%%%%%%%%%%%%%%%%%%%%%%%%%%%%%%%%%%%%%%%%%%%%%%%%%%%%%%%%%%%%%%%%%%%%%%%%%%%
\subsection{Information Design}\label{subsec: information design} 

% Introduction
One application is \citeapos{kamenica2011bayesian} \textit{information design}. 
Suppose that there exists another agent, called an information designer, who chooses an information structure $F$ to maximize an objective function $g:\mR^{n}\rightarrow\mR$ defined over search behaviors.
Equivalently, the information designer maximizes $g$ over $V^{*}$.\footnote{As in \citet{kamenica2011bayesian}, we assume that the information designer can pick a preferred equilibrium if there are multiple strategies consistent with Pandora's rule.} 
As discussed in the Introduction, the information designer in our model can be interpreted as an online platform that controls the amount of information disclosed on each product page in order to maximize pay-per-click advertising revenue.\footnote{Although their models are quite different from ours, platforms with similar objective functions are also studied in \citet{eliaz2011simple}, \citet{hagiu2011intermediaries}, \citet{de2016search}, \citet{gossner2021attention}, \citet{koh2022attention}, and \citet{hebert2026engagement}. Theorem \ref{thm: joint distribution} in the next section shows that the welfare implications change when such a platform also earns sales commissions.}

% Practical example
In applications, information designers may influence agents' behavior not only through information but also through other instruments. 
For example, when the information designer is interpreted as an online platform, it may not only design product web pages but also directly guide consumers toward a desired search order by placing products in prominent positions. Here, we implicitly assume that the designer manipulates only information.\footnote{\citet{armstrong2009prominence} develops a search model on this topic.} This is an essential simplification for isolating the implications of the design problem that arise from the dual role of information structures emphasized in this paper from those generated by other instruments specific to each application.

In this setting, our findings induce the following result.
Call a search behavior $v$ \textit{fully implementable} if, for any $\delta>0$, there exists a CDF profile $F$ such that for any Pandora's optimal strategy, the induced search behavior $v'$ satisfies $\max_{i}|v_{i}-v'_{i}|<\delta$.
Recall that $\Pi_{r}$ is defined in Proposition \ref{prop: face}.

\begin{proposition}\label{prop: optimal}
$\overline{F}$ is optimal for any upper-continuous objective function $g$, which minimizes Pandora's welfare.
Moreover, if $g(v)=r\cdot v$ for some $r\in\mR^{n}\setminus\{0\}$, we have the following: 
(i) it is optimal to let Pandora open boxes according to any order in $\Pi_{r}$, and 
(ii) if Assumption \ref{assum: regularity} holds with strict inequality, an optimal search behavior is fully implementable. 
\end{proposition}

% Discussion: optimal design
The first statement shows that any information designer who wishes to induce a preferred search behavior may result in minimizing Pandora's welfare among all feasible distributions.
Surprisingly, the optimal information structure $\overline{F}$ does not depend on the shape of the objective function. 
The upper-continuity of $g$ is assumed just for ensuring the existence of a maximizer.

% Discussion: linear case 1
The second statement is about the case of linear objective functions. 
In the former part, we characterize the designer-optimal search orders.
Recall that Pandora is indifferent among all search orders under $\overline{F}$ and thus the designer is free to choose any preferred search order.
For example, if $c_{i}=c$ and $r_{i}\geq 0$ for each $i$, we obtain an intuitive result that the information designer prefers to let Pandora open boxes in the descending order of $r_{i}$. 

% Discussion: linear case 2
The latter part concerns full implementation.
Although $ \overline{F}$ is optimal, it is not robustly optimal because Pandora can adopt an alternative optimal strategy that yields the designer-worst search behavior.
We show that every vertex of $V^{*}$ is fully implementable and, consequently, so is any optimal search behavior for any linear objective function.
In the proof in Appendix \ref{sec:appendixb}, we perturb the benchmark distributions $\overline{F}$ in a manner tailored to each given vertex we wish to implement. 

\subsection{Identification Problems}\label{subsec: identification} 

% Introduction
Consider the following problem: 
Suppose analysts have access to data on search behavior $v$.\footnote{In practice, constructing the vector $v$ requires detailed session-level data on which options were inspected. Recently, such information is becoming increasingly available in clickstream and browsing datasets used in the empirical search literature, including those from online retail, travel, and marketplace platforms. See, for example, \citet{ursu2025sequential}.}
They wish to use this data $v$ and estimate the model parameters. 
However, they do not know to what extent the agent learns from each inspection. 
Assumptions about information structures can greatly affect the estimations and the predictions, and therefore, they wish to take a conservative stance and derive the set of all parameters such that this data set $v$ is consistent with an optimal play under some information structure. 

% Derivation
The closed-form characterization in Theorem \ref{thm: feasible search behavior} can be directly mapped to the closed-form characterization for this \textit{partial identification problem}. 
Let $\Gamma\subset \mR^{2}\times \mR^{2n}$ be any subset of all parameter profiles $(\overline{x},\underline{x})$ and $(\lambda_{i},c_{i})_{i}$ satisfying Assumption \ref{assum: regularity}. 
For each parameter $\gamma\in \Gamma$, write $V^{*}(\gamma)$ to be the set of all feasible search behaviors which we derive in Theorem \ref{thm: feasible search behavior}. 
Then, 
\begin{align*}
    I(v; \Gamma) = \{\gamma\in \Gamma\mid v\in V^{*}(\gamma)\},
\end{align*}
is the \textit{identified set} for a given observation $v\in \mR^{n}$, i.e., the set of model parameters that rationalize the data $v$ with some information structure. 

% Implication from identification
The identification problem considered here has a natural economic significance for the empirical search literature, apart from the paper's theoretical motivation of analyzing the dual role of information structures. First, the parameter $\gamma$ consists of primitives such as the qualities of the options and search costs. Although such primitives are not directly observable in practice, they are directly relevant for welfare, and hence their estimation is economically important. Second, as emphasized in the empirical search literature, estimating model parameters also provides a basis for counterfactual analysis. For example, in our setting, the identified set for $\gamma$ would allow us to analyze what search behavior may arise under alternative search environments, such as changes in search costs, product quality, or the set of available alternatives, in a way that is robust to the underlying information structure. 

% This is a closed-form characterization: symmetric case
Theorem \ref{thm: feasible search behavior} shows that $V^{*}(\gamma)$ is a polytope, which implies that the identified set $I(v; \Gamma)$ takes a simple form. 
In particular, it is characterized by finitely many inequality constraints.
This is clearer in a simpler setup, as illustrated in the next result.   
Let $\Gamma^{*}\subset \Gamma$ be the set of all symmetric parameters such that $\lambda_{i}=\lambda$, and $c_{i}=c$ for all $i$. 
For simplicity, we regard $\Gamma^{*}$ as a four-dimensional space and write any element as $\gamma=(\overline{x},\underline{x},\lambda,c)$. 

\begin{proposition}\label{prop: identified set}
    Take any search behavior $v$ such that $v_{1}\geq v_{2}\geq \cdots \geq v_{n}\geq 0$. 
    Then the identified set for the symmetric parameter space $\Gamma^{*}$ is
    \begin{align*}
        I(v; \Gamma^{*}) = 
        \left\{\gamma\in \Gamma^{*} ~\middle|~ \sum_{i=1}^{n}v_{i}\geq 1 ~\text{and}~ \sum_{j\leq k}v_{j}\leq \sum_{j\leq k}p(\gamma)^{j-1} ~\text{for all}~ k\right\}, 
    \end{align*} 
    where for each $\gamma=(\overline{x},\underline{x},\lambda,c)$, we set $p(\gamma)=1-c/(\overline{x}-\underline{z})$. 
\end{proposition}

% This is a closed-form characterization: symmetric case cont'd
Therefore, the identified set in the space of symmetric primitives is characterized by $n+1$  inequalities.\footnote{The latter set of $n$ inequalities is often referred to as a \textit{majorization constraint}.} 
We can check that $p(\gamma)$ is decreasing in $c$, and therefore, the marginal of $I(v; \Gamma^{*})$ over search costs is a closed interval of a form $[0,\overline{c}]$. 
Likewise, $p(\gamma)$ is also decreasing in $\lambda$; hence, the projection of $I(v;\Gamma^{*})$ onto the $(\lambda,c)$-space is described by an upper-bound constraint.

% Remark
It is natural for analysts to observe \textit{choice behavior} as well, namely the probabilities with which each box is ultimately selected. 
In the next section, we extend our analysis, and Theorem \ref{thm: joint distribution} delivers a closed-form characterization of the set of all pairs of search and choice behaviors. 
This result enables us to characterize the identified set consistent with any search–choice pair.
Furthermore, as we discuss later, we can also partially relax the binary-reward assumption by allowing each $x_{i}$ to follow a general distribution $F_{i}$.
Accordingly, we can also extend our framework to the nonparametric identification problem, under which the space $\Gamma$ would consist of pairs of prior distributions and search costs. 

% Beyond the above information-design perspective, Theorem \ref{thm: feasible search behavior} can also be seen as an information-free prediction of the set of equilibrium outcomes. 
% In a job-search context, for example, a potential worker sequentially reviews job postings but receives only noisy signals of match quality.  
% The theorem characterizes the entire set of search behaviors that can emerge in such settings. 

%%%%%%%%%%%%%%%%%%%%%%%%%%%%%%%%%%%%%%%%%%%%%%%%%%%%%%%%%%%%%%%%%%%%%%%%%%%%%%
\subsection{Comparative Statics}\label{subsec: comparative statics} 

% Comparative Statics
Whether we view our analysis as an information-free prediction, an information-design problem, or an identification problem, it is natural to ask how the feasible set varies with the primitives.  
Because Theorem \ref{thm: feasible search behavior} provides all vertices of $V^{*}$ in closed form, we can trace how each vertex responds to changes in the underlying parameters.  
This observation yields the following result.

\begin{proposition}\label{prop: comparative statics}
    If $\underline{z}_{i}=\underline{z}$, the set $V^{*}$ shrinks when $\lambda_{i}$ increases. 
    If $\underline{z}_{i}<\underline{z}$, then $V^{*}$ expands when $\lambda_{i}$ increases. 
    In the symmetric case where $\lambda_{i}=\lambda$ and $c_{i}=c$ for all $i$, $V^{*}$ shrinks when $\lambda$ increases or $c$ increases. 
\end{proposition}

% Note
Proposition \ref{prop: comparative statics} summarizes how the set of feasible search behaviors responds to each primitive parameter.
The effect of increasing $\lambda_{i}$ on the set of all feasible search behaviors depends on the relative ex-ante attractiveness of box $i$. 
The effect of increasing $c_{i}$ is also ambiguous. 
However, it turns out that the set is monotone in the symmetric primitives. 
Figure \ref{fig: symmetry case} illustrates this point. 

%\vspace{1em}
\begin{figure}[ht]
\begin{minipage}[t]{.45\textwidth}
\begin{center}
\begin{tikzpicture}[scale=0.7, transform shape]

    % Define such that underline{z}=underline{z}_{Z}
    \pgfmathsetmacro{\muX}{0.8}
    \pgfmathsetmacro{\muY}{0.8}
    \pgfmathsetmacro{\muZ}{0.8}
    \pgfmathsetmacro{\cX}{0.2}
    \pgfmathsetmacro{\cY}{0.2}
    \pgfmathsetmacro{\cZ}{0.2}

    \pgfmathsetmacro{\zmin}{max(\muX-\cX,\muY-\cY,\muZ-\cZ)}

    \pgfmathsetmacro{\pX}{1-\cX/(1-\zmin)}
    \pgfmathsetmacro{\pY}{1-\cY/(1-\zmin)}
    \pgfmathsetmacro{\pZ}{1-\cZ/(1-\zmin)}
    \pgfmathsetmacro{\qX}{1-(\muX-\cX)/\zmin}
    \pgfmathsetmacro{\qY}{1-(\muY-\cY)/\zmin}
    \pgfmathsetmacro{\qZ}{1-(\muZ-\cZ)/\zmin}

    \coordinate (O) at (0,0,0);
    \coordinate (X) at (4,0,0); 
    \coordinate (Y) at (0,4,0); 
    \coordinate (Z) at (0,0,4); 

    \coordinate (XY) at (4,4,0);
    \coordinate (YZ) at (0,4,4);
    \coordinate (XZ) at (4,0,4);
    \coordinate (XYZ) at (4,4,4);

    % Cube
    \draw (O) -- (X) -- (XY) -- (Y) -- cycle;
    \draw (O) -- (Z);
    \draw (X) -- (XZ);
    \draw (Y) -- (YZ);
    \draw (XY) -- (XYZ);
    \draw (Z) -- (XZ) -- (XYZ) -- (YZ) -- cycle; 

    % Extreme Points
    \coordinate (XYZ0) at (4,\pX*4,\pX*\pY*4); \filldraw (XYZ0) circle (2pt);
    \coordinate (XYZ1) at (4,\pX*4,\qX*\qY*4); \filldraw (XYZ1) circle (2pt);
    \coordinate (XYZ2) at (4,\qX*4,\qX*\qY*4); \filldraw (XYZ2) circle (2pt);
    \coordinate (XZY0) at (4,\pX*\pZ*4,\pX*4); \filldraw (XZY0) circle (2pt);
    \coordinate (XZY1) at (4,\qX*\qZ*4,\pX*4); \filldraw (XZY1) circle (2pt);
    \coordinate (XZY2) at (4,\qX*\qZ*4,\qX*4); \filldraw (XZY2) circle (2pt);
    \coordinate (YXZ0) at (\pY*4,4,\pY*\pX*4); \filldraw (YXZ0) circle (2pt);
    \coordinate (YXZ1) at (\pY*4,4,\qY*\qX*4); \filldraw (YXZ1) circle (2pt);
    \coordinate (YXZ2) at (\qY*4,4,\qY*\qX*4); \filldraw (YXZ2) circle (2pt);
    \coordinate (YZX0) at (\pY*\pZ*4,4,\pY*4); \filldraw (YZX0) circle (2pt);
    \coordinate (YZX1) at (\qY*\qZ*4,4,\pY*4); \filldraw (YZX1) circle (2pt);
    \coordinate (YZX2) at (\qY*\qZ*4,4,\qY*4); \filldraw (YZX2) circle (2pt);
    \coordinate (ZXY0) at (\pZ*4,\pX*\pZ*4,4); \filldraw (ZXY0) circle (2pt);
    \coordinate (ZXY1) at (\pZ*4,\qX*\qZ*4,4); \filldraw (ZXY1) circle (2pt);
    \coordinate (ZXY2) at (\qZ*4,\qX*\qZ*4,4); \filldraw (ZXY2) circle (2pt);
    \coordinate (ZYX0) at (\pY*\pZ*4,\pZ*4,4); \filldraw (ZYX0) circle (2pt);
    \coordinate (ZYX1) at (\qY*\qZ*4,\pZ*4,4); \filldraw (ZYX1) circle (2pt);
    \coordinate (ZYX2) at (\qY*\qZ*4,\qZ*4,4); \filldraw (ZYX2) circle (2pt);

    % Feasible Search Behaviors
    \filldraw[fill=blue, opacity=0.4] (XYZ0) -- (YXZ0) -- (YZX0) -- (ZYX0) -- (ZXY0) -- (XZY0) -- cycle; 
    \filldraw[fill=blue, opacity=0.0] (XYZ1) -- (YXZ1) -- (YZX1) -- (ZYX1) -- (ZXY1) -- (XZY1) -- cycle; 
    \filldraw[fill=gray, opacity=0.1] (XYZ2) -- (YXZ2) -- (YZX2) -- (ZYX2) -- (ZXY2) -- (XZY2) -- cycle; 
    
    \filldraw[fill=blue, opacity=0.25] (XYZ0) -- (YXZ0) -- (YXZ1) -- (XYZ1) -- cycle;
    \filldraw[fill=blue, opacity=0.3] (YXZ0) -- (YZX0) -- (YZX1) -- (YZX2) -- (YXZ2) -- (YXZ1) -- cycle;
    \filldraw[fill=blue, opacity=0.45] (YZX0) -- (ZYX0) -- (ZYX1) -- (YZX1) -- cycle;
    \filldraw[fill=blue, opacity=0.5] (ZYX0) -- (ZXY0) -- (ZXY1) -- (ZXY2) -- (ZYX2) -- (ZYX1) -- cycle;
    \filldraw[fill=blue, opacity=0.375] (ZXY0) -- (XZY0) -- (XZY1) -- (ZXY1) -- cycle;
    \filldraw[fill=blue, opacity=0.3] (XZY0) -- (XYZ0) -- (XYZ1) -- (XYZ2) -- (XZY2) -- (XZY1) -- cycle;
    \filldraw[fill=blue, opacity=0] (XYZ1) -- (YXZ1) -- (YXZ2) -- (XYZ2) -- cycle;
    \filldraw[fill=blue, opacity=0] (YZX1) -- (ZYX1) -- (ZYX2) -- (YZX2) -- cycle;
    \filldraw[fill=blue, opacity=0] (ZXY1) -- (XZY1) -- (XZY2) -- (ZXY2) -- cycle;

    % Coordinates
    \draw[->][dashed] (0,0,0) -- (5.5,0,0) node[right]{$v_{1}$};
    \draw[->][dashed] (0,0,0) -- (0,5.5,0) node[left]{$v_{2}$};
    \draw[->][dashed] (0,0,0) -- (0,0,5.5) node[left]{$v_{3}$};
\end{tikzpicture}
\caption*{$\lambda=0.8$ and $c=0.2$.}
\end{center}
\end{minipage}
\hfill
\begin{minipage}[t]{.45\textwidth}
\begin{center}
\begin{tikzpicture}[scale=0.7, transform shape]

    % Define such that underline{z}=underline{z}_{Z}
    \pgfmathsetmacro{\muX}{0.7}
    \pgfmathsetmacro{\muY}{0.7}
    \pgfmathsetmacro{\muZ}{0.7}
    \pgfmathsetmacro{\cX}{0.1}
    \pgfmathsetmacro{\cY}{0.1}
    \pgfmathsetmacro{\cZ}{0.1}

    \pgfmathsetmacro{\zmin}{max(\muX-\cX,\muY-\cY,\muY-\cY)}

    \pgfmathsetmacro{\pX}{1-\cX/(1-\zmin)}
    \pgfmathsetmacro{\pY}{1-\cY/(1-\zmin)}
    \pgfmathsetmacro{\pZ}{1-\cZ/(1-\zmin)}
    \pgfmathsetmacro{\qX}{1-(\muX-\cX)/\zmin}
    \pgfmathsetmacro{\qY}{1-(\muY-\cY)/\zmin}
    \pgfmathsetmacro{\qZ}{1-(\muZ-\cZ)/\zmin}

    \coordinate (O) at (0,0,0);
    \coordinate (X) at (4,0,0); 
    \coordinate (Y) at (0,4,0); 
    \coordinate (Z) at (0,0,4); 

    \coordinate (XY) at (4,4,0);
    \coordinate (YZ) at (0,4,4);
    \coordinate (XZ) at (4,0,4);
    \coordinate (XYZ) at (4,4,4);

    % Cube
    \draw (O) -- (X) -- (XY) -- (Y) -- cycle;
    \draw (O) -- (Z);
    \draw (X) -- (XZ);
    \draw (Y) -- (YZ);
    \draw (XY) -- (XYZ);
    \draw (Z) -- (XZ) -- (XYZ) -- (YZ) -- cycle; 

    % Extreme Points
    \coordinate (XYZ0) at (4,\pX*4,\pX*\pY*4); \filldraw (XYZ0) circle (2pt);
    \coordinate (XYZ1) at (4,\pX*4,\qX*\qY*4); \filldraw (XYZ1) circle (2pt);
    \coordinate (XYZ2) at (4,\qX*4,\qX*\qY*4); \filldraw (XYZ2) circle (2pt);
    \coordinate (XZY0) at (4,\pX*\pZ*4,\pX*4); \filldraw (XZY0) circle (2pt);
    \coordinate (XZY1) at (4,\qX*\qZ*4,\pX*4); \filldraw (XZY1) circle (2pt);
    \coordinate (XZY2) at (4,\qX*\qZ*4,\qX*4); \filldraw (XZY2) circle (2pt);
    \coordinate (YXZ0) at (\pY*4,4,\pY*\pX*4); \filldraw (YXZ0) circle (2pt);
    \coordinate (YXZ1) at (\pY*4,4,\qY*\qX*4); \filldraw (YXZ1) circle (2pt);
    \coordinate (YXZ2) at (\qY*4,4,\qY*\qX*4); \filldraw (YXZ2) circle (2pt);
    \coordinate (YZX0) at (\pY*\pZ*4,4,\pY*4); \filldraw (YZX0) circle (2pt);
    \coordinate (YZX1) at (\qY*\qZ*4,4,\pY*4); \filldraw (YZX1) circle (2pt);
    \coordinate (YZX2) at (\qY*\qZ*4,4,\qY*4); \filldraw (YZX2) circle (2pt);
    \coordinate (ZXY0) at (\pZ*4,\pX*\pZ*4,4); \filldraw (ZXY0) circle (2pt);
    \coordinate (ZXY1) at (\pZ*4,\qX*\qZ*4,4); \filldraw (ZXY1) circle (2pt);
    \coordinate (ZXY2) at (\qZ*4,\qX*\qZ*4,4); \filldraw (ZXY2) circle (2pt);
    \coordinate (ZYX0) at (\pY*\pZ*4,\pZ*4,4); \filldraw (ZYX0) circle (2pt);
    \coordinate (ZYX1) at (\qY*\qZ*4,\pZ*4,4); \filldraw (ZYX1) circle (2pt);
    \coordinate (ZYX2) at (\qY*\qZ*4,\qZ*4,4); \filldraw (ZYX2) circle (2pt);

    % Feasible Search Behaviors
    \filldraw[fill=blue, opacity=0.4] (XYZ0) -- (YXZ0) -- (YZX0) -- (ZYX0) -- (ZXY0) -- (XZY0) -- cycle; 
    \filldraw[fill=blue, opacity=0.0] (XYZ1) -- (YXZ1) -- (YZX1) -- (ZYX1) -- (ZXY1) -- (XZY1) -- cycle; 
    \filldraw[fill=gray, opacity=0.1] (XYZ2) -- (YXZ2) -- (YZX2) -- (ZYX2) -- (ZXY2) -- (XZY2) -- cycle; 
    
    \filldraw[fill=blue, opacity=0.25] (XYZ0) -- (YXZ0) -- (YXZ1) -- (XYZ1) -- cycle;
    \filldraw[fill=blue, opacity=0.3] (YXZ0) -- (YZX0) -- (YZX1) -- (YZX2) -- (YXZ2) -- (YXZ1) -- cycle;
    \filldraw[fill=blue, opacity=0.45] (YZX0) -- (ZYX0) -- (ZYX1) -- (YZX1) -- cycle;
    \filldraw[fill=blue, opacity=0.5] (ZYX0) -- (ZXY0) -- (ZXY1) -- (ZXY2) -- (ZYX2) -- (ZYX1) -- cycle;
    \filldraw[fill=blue, opacity=0.375] (ZXY0) -- (XZY0) -- (XZY1) -- (ZXY1) -- cycle;
    \filldraw[fill=blue, opacity=0.3] (XZY0) -- (XYZ0) -- (XYZ1) -- (XYZ2) -- (XZY2) -- (XZY1) -- cycle;
    \filldraw[fill=blue, opacity=0] (XYZ1) -- (YXZ1) -- (YXZ2) -- (XYZ2) -- cycle;
    \filldraw[fill=blue, opacity=0] (YZX1) -- (ZYX1) -- (ZYX2) -- (YZX2) -- cycle;
    \filldraw[fill=blue, opacity=0] (ZXY1) -- (XZY1) -- (XZY2) -- (ZXY2) -- cycle;

    % Coordinates
    \draw[->][dashed] (0,0,0) -- (5.5,0,0) node[right]{$v_{1}$};
    \draw[->][dashed] (0,0,0) -- (0,5.5,0) node[left]{$v_{2}$};
    \draw[->][dashed] (0,0,0) -- (0,0,5.5) node[left]{$v_{3}$};
\end{tikzpicture}
\caption*{$\lambda=0.7$ and $c=0.1$.}
\end{center}
\end{minipage}
\vspace{1em}
\caption{The set of feasible search behaviors under symmetry.}
\label{fig: symmetry case}
\end{figure}

% Intuition
A rough intuition is as follows:
As $\lambda_{i}$ increases, Pandora is more likely to sample a high value and stop searching.  
Hence, intuitively, the maximum search length $r\cdot v$ weighted by $r\geq 0$ decreases.  
Moreover, if the primitives are symmetric, one can see that the minimum value of $r\cdot v$ equals $\min_{i}r_{i}$, which does not depend on the primitives.\footnote{Since there exists no outside option, $\sum_{i=1}^{n}v_{i}\geq 1$ for any feasible search behavior $v$. 
Hence, we obtain $r\cdot v\geq \min_{i}r_{i}$. 
This lower bound is achieved by a CDF profile that is degenerate at prior mean, if the primitives are symmetric.} 
Therefore, the range of $r \cdot v$ becomes smaller, which suggests that $V^{*}$ shrinks.  
However, if the primitives are asymmetric, the lower bound for $r \cdot v$ may depend on the primitive values and can decrease, in which case $V^{*}$ may neither shrink nor expand. 

%%%%%%%%%%%%%%%%%%%%%%%%%%%%%%%%%%%%%%%%%%%%%%%%%%%%%%%%%%%%%
\section{Extensions}\label{sec: general state space}
%%%%%%%%%%%%%%%%%%%%%%%%%%%%%%%%%%%%%%%%%%%%%%%%%%%%%%%%%%%%% 

% Remark
In this section, we extend our main analysis in two directions.
Throughout, we will impose a common symmetry assumption on the primitive parameters while still considering the full space of feasible information structures.
Under this symmetry, we can dispense with the earlier implicit assumption that true rewards take only binary values, as discussed below. 

%%%%%%%%%%%%%%%%%%%%%%%%%%%%%%%%%%%%%%%%%%%%%%%%%%%%%%%%%%%%%
\subsection{Feasible Search and Choice Behaviors}\label{subsec: general state space}
%%%%%%%%%%%%%%%%%%%%%%%%%%%%%%%%%%%%%%%%%%%%%%%%%%%%%%%%%%%%%

% Introduction
In the previous section, we defined search behavior as the vector of probabilities with which each box is opened. 
However, this definition does not fully characterize the realizations of Pandora's actions, because she must also decide which box to select in the end.
Here, we examine the extent to which our analysis can be extended by also considering Pandora's eventual selections of boxes. 

% Definition: choice
A \textit{choice behavior} is a vector $d\in\mR^{n}$ that describes each probability $d_{i}$ that each box $i$ is eventually chosen. 
Call a choice behavior $d$ \textit{feasible} if a feasible distribution profile and an optimal strategy induce it.
A pair $(v,d)\in \mR^{n}\times\mR^{n}$ is \textit{feasible} if a feasible distribution profile and an optimal strategy induce $v$ as the search behavior and $d$ as the choice behavior. 
We let $S^{*}\subset \mR^{n}\times \mR^{n}$ denote the set of all feasible search and choice behaviors. 

% Main result
The next result shows that, if the primitives are to some extent symmetric and $\underline{z}_{i}=\mu_{i}-c_{i}$ is homogeneous across boxes, then we can generalize Theorem \ref{thm: feasible search behavior} and characterize the set of all feasible pairs of search and choice behaviors. 
This assumption is satisfied if $\mu_{i}=\mu$ and $c_{i}=c$ for all boxes $i$. 

% ===========================================
\begin{theorem}\label{thm: joint distribution}
    Suppose $\underline{z}_{i}=\underline{z}$ for all $i$.
    The set of all feasible pairs $S^{*}$ is a polytope. 
    There exists a single distribution profile $F^{*}$ that induces all feasible pairs, under which the reservation value of every box equals $\underline{z}$ and Pandora's expected payoff is minimized among all feasible distribution profiles, which also equals $\underline{z}$.
\end{theorem}

\begin{remark}\label{rem: general prior}
    Thus far, we assume that the prior distributions of true rewards $x_{i}$ are binary, i.e., $x_{i}\in \{\underline{x}, \overline{x}\}$ for every box $i$. 
    Theorem \ref{thm: joint distribution} no longer needs this assumption. 
    In Appendix \ref{sec:appendixc}, we prove Theorem \ref{thm: joint distribution} under a general assumption that the true reward $x_{i}$ follows any distribution $F_{i}^{0}$ over the interval $[\underline{x},\overline{x}]$. %\footnote{We are grateful to an anonymous referee for the subsumed paper for suggesting this extension.} 
    If each $F_{i}^{0}$ has mean $\mu_{i}$ and $\underline{z}_{i}=\mu_{i}-c_{i}$ is homogeneous across boxes, then the statement of Theorem \ref{thm: joint distribution} does not change at all.
\end{remark}

% Discussion: implciation 1
The main implications mirror those of Theorem \ref{thm: feasible search behavior}. 
For example, paralleling Proposition \ref{prop: optimal}, the theorem implies that an information design intended to implement a specific distribution of search and choice behavior may still minimize the welfare of searching agents. 
For instance, consider online platforms that derive revenue from both user visits and sales commissions. 
The expected revenue of such platforms is a weighted sum of search and choice outcomes. 
Then, Theorem \ref{thm: joint distribution} suggests that the welfare implication of Proposition \ref{prop: optimal} remains intact in this richer setting.

% \blue{
% % Discussion: implciation 2
% Recall also that Theorem \ref{thm: feasible search behavior} implies that observed search behavior alone does not yield a meaningful identification set for the welfare of searching agents. 
% In particular, since a welfare-minimizing distribution profile is compatible with any feasible search behavior, the lower bound on Pandora's expected payoff is independent of search behavior. 
% Although most studies may impose symmetric priors, Theorem \ref{thm: joint distribution} then delivers an even stronger result: meaningful welfare bounds remain elusive even when we additionally observe choice behavior. 
% } 

% Vertex
As in Theorem \ref{thm: feasible search behavior}, we can also characterize the vertexes of the set of all feasible pairs $S^{*}$, which is summarized in the following proposition.  
Moreover, analogous to Proposition \ref{prop: optimal}, we can fully implement each vertex.\footnote{A pair $(v,d)$ is \textit{fully implementable} if, for any $\delta>0$, there exists a feasible CDF profile $F$ such that for any Pandora's optimal strategy under $F$, the induced pair $(v',d')$ satisfies $\max_{i}|v_{i}-v'_{i}|<\delta$ and $\max_{i}|d_{i}-d'_{i}|<\delta$.} 
In Appendix \ref{sec:appendixc}, we prove its generalization under general priors. 

% ===========================================
\begin{proposition}\label{prop: joint distribution vertex}
    Suppose $\underline{z}_{i}=\underline{z}$ for all $i$.
    Take any $(v,d)\in S^{*}$ such that $v_{1}\geq v_{2}\geq\cdots \geq v_{n}$. 
    Then, $(v,d)$ is a vertex of $S^{*}$ if and only if, for some $i$ and $j\leq i$, 
    \begin{align*}
        v_{k}
        &= 
        \begin{cases}
            \overline{v}_{k} \quad & \text{if} \quad k\leq i,\\
            0 \quad & \text{otherwise},
        \end{cases}\\
        d_{k}
        &= 
        \begin{cases}
            \overline{v}_{k}-\overline{v}_{k+1} \quad & \text{if} \quad k\leq i \text{ and } k\neq j,\\
            \overline{v}_{k}-\overline{v}_{k+1}+p_{i}\cdot \overline{v}_{i} \quad & \text{if} \quad k=j, \\
            0 \quad & \text{otherwise},
        \end{cases}
    \end{align*}
    for each $k$. 
    Moreover, any vertex is fully implementable. 
\end{proposition} 

\begin{remark}\label{rem: full implementation general priors}
    For full implementation, Proposition \ref{prop: optimal} requires Assumption \ref{assum: regularity} to hold with strict inequalities, that is, $\overline{z}> \underline{z}$. 
    In Proposition \ref{prop: joint distribution vertex}, the symmetry condition guarantees that this requirement is satisfied. 
    Under general priors, we prove in Appendix \ref{sec:appendixc} that any vertex can be fully implemented if, for each box $i$, the reservation value under the prior is strictly higher than $\underline{z}_{i}=\mu_{i}-c_{i}$. 
\end{remark}

% Intuition
Each vertex is identified by three elements: a search order $\pi$, a threshold box $i$, and a box $j$ invoked for the "return demand."  
Intuitively, each distribution $F^{*}_{k}$ which we construct in the proof can take one of two possible values, $\underline{z}$ and $h_{k}>\underline{z}$, and Pandora implements the vertex $(\pi,i,j)$ using the following strategy. 
She opens boxes in the order $\pi$, with the highest possible probabilities to boxes up to box $i$.  
Pandora stops search and selects box $k$ whenever she observes its high realization $h_{k}$, and she also selects box $j$ when all realizations observed up to the threshold box $i$ are low, i.e. $y_{k}=\underline{z}$ for all $k\leq i$. 
Boxes after $i$ are never opened and therefore never chosen.

% Proof sketch
The proof of Theorem \ref{thm: joint distribution} parallels the approach used to prove Theorem \ref{thm: feasible search behavior}.
We first fix an arbitrary feasible search behavior $v$ and derive a tight bound on the choice behaviors $d$ for which the pair $(v,d)$ is feasible.
This derivation relies on \citet{armstrong2017ordered} and \citet{choi2018consumer}, which characterize equilibrium choice behaviors in terms of the distribution of \textit{effective values}.
The symmetry assumption $\underline{z}_{i}=\underline{z}$ ensures that this bound is exact.
Then, we extend this analysis to obtain a tight bound on the joint set of feasible search and choice behaviors and construct a profile of binary-support distributions $F^{*}$ that attains every point in this set.
Under binary priors, $F^{*}$ coincides with the distribution $\overline{F}$ derived in Theorem \ref{thm: feasible search behavior}; for general priors under which we prove Theorem \ref{thm: joint distribution} in Appendix \ref{sec:appendixc}, no closed-form expression for $F^{*}$ is generally available. 

%%%%%%%%%%%%%%%%%%%%%%%%%%%%%%%%%%%%%%%%%%%%%%%%%%%%%%%%%%
\subsection{Price Competition}\label{subsec: price competition}

% Introduction
We have thus far focused on the demand side of search markets, treating the supply side as exogenous. 
In some applications, each box corresponds to a product whose seller sets its price. 
This section analyzes an ordered search market in which sellers compete by setting prices. 
Assume $n\geq 2$.

% Setting
Suppose that each box $i$ corresponds to one seller $i$. 
In the beginning of the game, each seller $i$ simultaneously posts a price $t_{i}\geq 0$. 
Then, Pandora observes the price realizations and engages in sequential search. 
If Pandora chooses box $i$ in the end, seller $i$ obtains payoff $t_{i}$ and Pandora receives payoff $x_{i}-t_{i}$ net of the total search cost.\footnote{We assume that each seller has zero marginal cost of production. This is without loss of generality because the consumer has no outside option: the proof of Theorem \ref{thm: price competition} remains valid as long as all sellers share the same production cost.} 
The other sellers' payoffs are zero. 
We consider (possibly mixed) Nash equilibria of this game. 

% Definition: equilibrium search behaviors
An \textit{equilibrium search behavior} is one induced by a Nash equilibrium under some feasible distribution profile $F$. 
We analogously define an \textit{equilibrium pair of search and choice behavior}.
A major challenge in incorporating price competition into the ordered search model is that the structure of equilibrium price distributions remains unresolved in the literature. 
For instance, \citet{obradovits2023price} solve the case in which each $F_{i}$ is a common binary-support distribution, yet the symmetric equilibrium is still a complex mixed strategy and its form depends on the primitives.\footnote{Whether Pandora observes prices before searching has a crucial effect on equilibrium price distributions. We assume she observes realized prices; see \citet{au2024attraction} for related discussion.} 

% Main result 
Since we do not generally know the shape of equilibrium price distributions, there is no formula for an equilibrium search behavior as a function of feasible distribution profiles. 
However, we still obtain the following. 
Again, Theorem \ref{thm: price competition} remains valid under general priors, which we show in the Appendix. 

% ===========================================
\begin{theorem}\label{thm: price competition}
    Suppose $\underline{z}_{i}=\underline{z}$ for all $i$.
    The set of equilibrium pairs coincides with the set of feasible pairs $S^{*}$. 
    A single distribution profile $\overline{F}$ induces all equilibrium pairs, under which every seller obtains zero surplus. 
\end{theorem}

% Implications
Therefore, the set of rationalizable behaviors is invariant to whether option sellers compete on prices. 
The same reduction principle applies: a single distribution profile rationalizes every equilibrium outcome. 
However, Theorem \ref{thm: price competition} has a rather different welfare implication from those in Theorems \ref{thm: feasible search behavior} and \ref{thm: joint distribution}. 
When sellers' price-setting behavior is accounted for, the supply side earns zero surplus. 
Because each seller charges a price of zero, the demand side may now prefer this distribution profile over another. 

% Proof 1
A simple yet crucial observation for this result is that the Bertrand outcome in which all prices are zero is an equilibrium under the CDF profile $F^{*}$ constructed in Theorem \ref{thm: joint distribution}. 
This immediately implies that every feasible pair is an equilibrium pair. 
The intuition is as follows: posting a positive price $t_{i}$ lowers the reservation value to $\underline{z}-t_{i}$. 
If another box $j$ is priced at zero, Pandora visits $j$ before $i$.
The distribution $F^{*}$ is such that every realization is weakly above $\underline{z}$; hence, Pandora never visits seller $i$ who posts a positive price. 

% Proof 2
The remaining proof shows the converse, i.e., equilibrium pairs are in the set $S^{*}$ we derive in Theorem \ref{thm: joint distribution}. 
The key idea is that introducing a positive price can only shorten the weighted search length relative to the benchmark model without pricing.
Hence, any equilibrium search behavior still lies within the polytope characterized in Theorem \ref{thm: feasible search behavior}. 
We then use this fact to show that, conditional on a given equilibrium search behavior, the associated equilibrium choice behaviors should also belong to the polytope identified in Theorem \ref{thm: joint distribution}.

% Technical contribution 
Equilibrium price competition in consumer search markets is widely recognized in the literature as a difficult, yet largely unexplored, problem. 
Theorems \ref{thm: feasible search behavior} and \ref{thm: joint distribution} provide tools to bypass this problem, yielding Theorem \ref{thm: price competition}.

%%%%%%%%%%%%%%%%%%%%%%%%%%%%%%%%%%%%%%%%%%%%%%%%%%%%%%%%%%
%%%%%%%%%%%%%%%%%%%%%%%%%%%%%%%%%%%%%%%%%%%%%%%%%%%%%%%%%%

\section{Conclusion} \label{sec:conclusion}

% Results summary
This paper studies how information structures shape \textit{ordered learning paths} by extending \citeapos{weitzman1979optimal} Pandora's problem in the simplest possible way.
In Theorem \ref{thm: feasible search behavior}, we characterize all search behaviors rationalized by some information structure. 
The closed-form expression allows us to provide a tractable approach to a partial identification problem. 
A single welfare-minimizing information structure implements any feasible search behavior, implying that information designed to induce a particular search behavior may reduce searchers' welfare.
Under a symmetry assumption, Theorems \ref{thm: joint distribution} and \ref{thm: price competition} extend the characterization to feasible pairs of search and choice behaviors, with and without price competition.
Although price competition among sellers does not alter the rationalizable set, it leads to distinct welfare implications. 
These findings suggest that our model provides a tractable framework for analyzing the underexplored effects of information structures on the ordered learning paths. 

% Future research
We choose to study the classic Pandora’s box problem as our base model because we aim to establish a theoretical benchmark. 
Naturally, some possible applications would have extra features beyond those captured here.
For instance, searching agents may learn an information structure itself over time \citep{adam2001learning}, choose an item without inspecting it \citep{doval2018whether}, or receive information about additional products while viewing a particular one \citep{ke2020informational, bao2022search}. 
Tailoring the model to each application is a promising avenue for future research. 
However, such extensions may need to extract additional structures from the application, as most variants of the Pandora's box problem lack a general characterization of optimal search strategies. 

%%%%%%%%%%%%%%%%%%%%%%%%%%%%%%%%%%%%%%%%%%%%%%%%%%%%%%%%%%%%%%%%%%%%%%%%%%%%%%

%%%%%%%%%%%%%%%%%%%%%%%%%%%%%%%%%%%%%%%%%%%%%%%%%%%%%%%%%%%%%%%%%%%%%%%%%%%%%%
%%%%%%%%%%%%%%%%%%%%%%%%%%%%%%%%%%%%%%%%%%%%%%%%%%%%%%%%
\titleformat{\section}
		{\Large\bfseries}     
         {Appendix \thesection:}% the label and number
        {0.5em}% space between label/number and subsection title
        {}% formatting commands applied just to subsection title
        []% punctuation or other commands following subsection title

        % Change theorem numbering to A.1, A.2, etc.
\renewcommand{\thetheorem}{A.\arabic{theorem}}
\setcounter{theorem}{0}

 \appendix 

%%%%%%%%%%%%%%%%%%%%%%%%%%%%%%%%%%%%%%%%%%%%%%%%%%%%%%%%%%%%%%%%%%%%%%%%%%%%%%
\section{Proof of Theorem \ref{thm: feasible search behavior}}\label{sec:appendixa}

% Description
We prove Theorem \ref{thm: feasible search behavior} in this section. 
For each order $\pi$, let $V_{\pi}$ be the set of all vectors $v$ such that $v_{\pi(1)}\geq\cdots\geq v_{\pi(n)}$, 
\begin{align*}
    v_{\pi(i)} \geq \underline{v}_{\pi,i}
    \quad \text{and} \quad
    \sum_{j\leq i} \frac{c_{\pi(j)}}{\overline{x}-\underline{z}}v_{\pi(j)} \leq \sum_{j\leq i} \frac{c_{\pi(j)}}{\overline{x}-\underline{z}}\overline{v}_{\pi,j},
\end{align*}
for each $i$. 
Recall $\overline{v}_{\pi,i}=\prod_{j<i} p_{\pi(j)}$ and $\underline{v}_{\pi,i}=\prod_{j<i} q_{\pi(j)}$. 
We show that the set of all feasible search behaviors $V^{*}$ is the convex hull of $V=\bigcup_{\pi}V_{\pi}$. 

% Step 1: Upper bound
The proof starts with proving the following result, which characterizes the feasible search behavior that maximizes a total search length weighted by any non-negative weight. 
This is the most involved part of the entire proof, and therefore, we prove it later in Subsection \ref{subsec: appendix information design} for readability.

% ----------------------------------
\begin{theorem}\label{thm: information design positive weight}
    For any feasible search behavior $v$ and any non-negative weight vector $r\in\mR^{n}$ such that
    \begin{align*}
        r_{1}\cdot \frac{\overline{x}-\underline{z}}{c_{1}}
        \geq r_{2}\cdot \frac{\overline{x}-\underline{z}}{c_{2}}
        \geq \cdots
        \geq r_{n}\cdot \frac{\overline{x}-\underline{z}}{c_{n}}
        \geq 0,
    \end{align*}
    we have $r\cdot \overline{v}\geq r\cdot v$. 
    Moreover, if the first $i$ inequalities are strict with $r_{i}>0$, and $v$ maximizes $r\cdot v$, then $v_{j}=\overline{v}_{j}$ for each $j\leq i$. 
\end{theorem} 

% Step 2: Lower bound
The next auxiliary lemma is used to derive a lower bound for each probability that each box is opened. 

% ----------------------------------
\begin{lemma}\label{lem: lower bound}
    Take any $z$ with $\underline{z}_{i}\leq z\leq \overline{z}_{i}$. 
    Then, we have 
    \begin{align*}
        1-\frac{\underline{z}_{i}-\underline{x}}{z-\underline{x}} = \min_{F_{i}:z_{i}\geq z} F_{i}(z-),
    \end{align*}
    where $F_{i}$ is a feasible distribution and $z_{i}$ is the associated reservation value.
    %In particular, $F^{*}_{i}$ is a solution to this problem when $z=\underline{z}$.
\end{lemma}

\begin{proof}
    % Bound
    Take any feasible distribution $F_{i}$ such that the associated reservation value $z_{i}$ satisfies  $z_{i}\geq z$.
    Then, the definition of reservation value implies
    \begin{align*}
        c_{i} 
        = \overline{x}-z_{i}-\int_{z_{i}}^{\overline{x}}F_{i}(y_{i})dy_{i} 
        &= \int_{z_{i}}^{\overline{x}}(1-F_{i}(y_{i}))dy_{i} \\
        &\leq \int_{z}^{\overline{x}}(1-F_{i}(y_{i}))dy_{i} 
        = \overline{x}-z-\int_{z}^{\overline{x}}F_{i}(y_{i})dy_{i}, 
    \end{align*}
    where the last inequality follows from $z_{i}\geq z$.
    Since $F_{i}$ has mean $\mu_{i}$ by feasibility, integration by parts implies that 
    \begin{align*}
        \overline{x}-\mu_{i} 
        &= \int_{\underline{x}}^{\overline{x}}F_{i}(y_{i})dy_{i} \\
        &= \int_{\underline{x}}^{z}F_{i}(y_{i})dy_{i} + \int_{z}^{\overline{x}}F_{i}(y_{i})dy_{i} 
        \leq (z-\underline{x})\cdot F_{i}(z-) + \overline{x}-z-c_{i},
    \end{align*}
    where the last inequality follows from the first remark. 
    Thus, from $\underline{z}_{i}=\mu_{i}-c_{i}$, we get $(z-\underline{x})\cdot F_{i}(z-)\geq z-\underline{z}_{i}$, which shows 
    \begin{align*}
        \min_{F_{i}:z_{i}\geq z} F_{i}(z-)
        \geq \frac{z-\underline{z}_{i}}{z-\underline{x}}
        = 1-\frac{\underline{z}_{i}-\underline{x}}{z-\underline{x}}.
    \end{align*}

    % Tight Bound
    It remains to show that there exists a feasible distribution $F_{i}$ such that $F_{i}(z-)=1-(\underline{z}_{i}-\underline{x})/(z-\underline{x})$ and the reservation value $z_{i}$ of $F_{i}$ satisfies $z_{i}\geq z$. 
    Define $F_{i}$ as follows:
    \begin{align*}
        F_{i}(y_{i}) = 
        \begin{cases}
            1-(\underline{z}_{i}-\underline{x})/(z-\underline{x}) & \quad \text{if} \quad y_{i}<z,\\
            1-c_{i}/(\overline{x}-z) & \quad \text{if} \quad z\leq y_{i}< \overline{x},\\
            1 & \quad \text{if} \quad \overline{x} \leq y_{i}.
        \end{cases}
    \end{align*} 
    To see that $F_{i}$ is a well-defined CDF, note that 
    \begin{align*}
        \overline{z}_{i}
        = \frac{\lambda_{i}\overline{x} - c_{i}}{\lambda_{i}}
        = \frac{\mu_{i}-(1-\lambda_{i})\underline{x} - c_{i}}{\lambda_{i}}
        = \frac{\underline{z}_{i}-(1-\lambda_{i})\underline{x}}{\lambda_{i}}.
    \end{align*}
    Therefore, $z\leq \overline{z}_{i}$ implies that $\lambda_{i}\leq (\underline{z}_{i}-\underline{x})/(z-\underline{x})$.
    Similarly, $z\leq \overline{z}_{i}=\overline{x}-c_{i}/\lambda_{i}$ implies that $\lambda_{i}\geq c_{i}/(\overline{x}-z)$.
    By combining them, one can check $1-(\underline{z}_{i}-\underline{x})/(z-\underline{x})\leq 1-c_{i}/(\overline{x}-z)$. 
    Therefore, $F_{i}$ is increasing in $y_{i}$.
    Moreover, $F_{i}$ has mean $\mu_{i}$, which follows from
    \begin{align*}
        \int_{\underline{x}}^{\overline{x}}F_{i}(y_{i})dy_{i}
        = (z-\underline{x})\left(1-\frac{\underline{z}_{i}-\underline{x}}{z-\underline{x}}\right)
        + (\overline{x}-z)\left(1-\frac{c_{i}}{\overline{x}-z}\right)
        %= z-\underline{z}_{i}+\overline{x}-z-c_{i}
        = \overline{x}-\mu_{i}.
    \end{align*}
    Since $F_{i}$ is right continuous and has mean $\mu_{i}$, and therefore, the distribution $F_{i}$ is feasible.
    
    Finally, notice that
    \begin{align*}
        \int_{z}^{\overline{x}}F_{i}(y_{i})dy_{i}
        = (\overline{x}-z)\left(1-\frac{c_{i}}{\overline{x}-z}\right)
        = \overline{x}-z-c_{i},
    \end{align*}
    which implies $z_{i}=z$ by definition, hence $z_{i}\geq z$. 
    This completes the proof.
\end{proof}

% Step 3: Bounding set
Combining Theorem \ref{thm: information design positive weight} and Lemma \ref{lem: lower bound}, we can conclude that any feasible search behavior must be in the convex hull of $V=\bigcup_{\pi}V_{\pi}$.

% ----------------------------------
\begin{lemma}\label{lem: necessary condition}
    If $v$ is feasible, then $v$ is in the convex hull of $V$.
\end{lemma}

\begin{proof}
    % Intro
    Assume w.l.g. that $v_{1}\geq \cdots\geq v_{n}$.
    By assumption, some feasible distribution profile $F$ and an optimal strategy induce $v$. 
    To begin with, assume that this optimal strategy is a pure strategy. 
    We prove that $v\in V_{\pi}$, where $\pi$ is an identity mapping. 

    % FOSD
    First, define $\overline{u}_{i}=c_{i}\overline{v}_{i}/(\overline{x}-\underline{z})$ and $u_{i}=c_{i}v_{i}/(\overline{x}-\underline{z})$ for each $i$. 
    Then, since $v$ is feasible, Theorem \ref{thm: information design positive weight} shows that $r\cdot u \leq r\cdot \overline{u}$ for all non-negative vectors $r\in \mR^{n}$ such that $r_{1}\geq \cdots\geq r_{n}$. 
    Equivalently, we have $\sum_{j\leq i}u_{j}\leq \sum_{j\leq i}\overline{u}_{j}$ for each $i$. 
    Therefore, $\sum_{j\leq i}c_{j}v_{j}/(\overline{x}-\underline{z})\leq \sum_{j\leq i}c_{j}\overline{v}_{j}/(\overline{x}-\underline{z})$ for each $i$.

    % Lower bound
    Second, we prove $v_{i}\geq \underline{v}_{i}$ for all $i$. 
    Let $z_{i}$ be the reservation value of each distribution $F_{i}$. 
    By assumption, $z_{1}\geq z_{2}\geq\cdots\geq z_{n}$ must hold to be consistent with Pandora's rule. 
    Take any $i$. 
    If $\underline{z}>\max\{\underline{z}_{i},\dots,\underline{z}_{n}\}$, the definition of $\underline{z}$ implies that there exists $j<i$ such that $\underline{z}=\underline{z}_{j}$, thus $\underline{v}_{i}=q_{1}\cdots q_{j}\cdots q_{i-1}=0$, and therefore, we clearly have $v_{i}\geq \underline{v}_{i}$.  
    Next, suppose $\underline{z}=\max\{\underline{z}_{i},\dots,\underline{z}_{n}\}$. 
    Since $z_{i}\geq z_{j}\geq \underline{z}_{j}$ for each $j\geq i$, 
    \begin{align*}
        v_{i}
        \geq \prod_{j<i}F_{j}(z_{i}-)
        \geq \prod_{j<i}F_{j}(\max\{\underline{z}_{i},\dots,\underline{z}_{n}\}-)
        = \prod_{j<i}F_{j}(\underline{z}-).
    \end{align*} 
    For each $j<i$, Lemma \ref{lem: lower bound} at $z=\underline{z}$ further implies\footnote{Note that Assumption \ref{assum: regularity} implies $\overline{z}_{j}\geq z\geq \underline{z}_{j}.$} 
    \begin{align*}
        F_{j}(\underline{z}-) 
        \geq 1-\frac{\underline{z}_{j}-\underline{x}}{\underline{z}-\underline{x}}.
    \end{align*}
    The right-hand side equals $q_{j}$ for all $j<i$, and therefore, we have $v_{i}\geq \underline{v}_{i}$ by definition. 

    % Summary
    The above two arguments prove $v\in V_{\pi}$. 
    If $v$ is induced by some mixed strategy, it is a convex combination of feasible search behaviors induced by each pure strategy in the support of the mixed strategy. 
    Therefore, $v$ must be in the convex hull of $V$. 
    This finishes the proof.
\end{proof}

% Step 4: Vertexes of the bounding set 
The convex hull of a finite union of polytopes is also a polytope, and therefore, the convex hull of $V$ is a polytope. 
In the following lemma, we essentially characterize all vertexes of this set.

% ----------------------------------
\begin{lemma}\label{lem: vertex any direction}
    Take any weight $r\in\mR^{n}$ such that
    \begin{align*}
    r_{\pi(1)}\cdot \frac{\overline{x}-\underline{z}}{c_{\pi(1)}}
    \geq
    \cdots 
    &\geq 
    r_{\pi(i)}\cdot \frac{\overline{x}-\underline{z}}{c_{\pi(i)}} \\
    &\geq 
    0 \\
    &>
    r_{\pi(i+1)}\cdot \frac{\underline{z}-\underline{x}}{\underline{z}_{\pi(i+1)}-\underline{x}}
    \geq
    \cdots
    \geq 
    r_{\pi(n)}\cdot \frac{\underline{z}-\underline{x}}{\underline{z}_{\pi(n)}-\underline{x}},
    \end{align*}
    for some order $\pi$. 
    Then, we have $r\cdot \tilde{v}\geq r\cdot u$ for all $u$ in the convex hull of $V$, where
    $\tilde{v}_{\pi(j)}=\overline{v}_{\pi,j}$ for all $j\leq i$ and $\tilde{v}_{\pi(j)}=\underline{v}_{\pi,j}$ for all $j>i$. 
    Moreover, $r\cdot \tilde{v}> r\cdot u$ for all $u\neq \tilde{v}$ if all of the above inequalities are strict. 
    Namely, $\tilde{v}$ is a vertex. 
\end{lemma}

\begin{proof}
    % Basics
    Assume w.l.g. that $\pi(i)=i$ for all $i$. 

    % Intro 
    Take any vector $u$ in the convex hull of $V$. 
    We prove that $r\cdot \tilde{v}\geq r\cdot u$. 
    From Bauer's Maximum Principle, there must exist an extreme point $\tilde{u}$ of the convex hull of $V$ such that $r\cdot \tilde{u}\geq r\cdot u$.
    Therefore, we may assume from the beginning that $u$ is an extreme point.
    In particular, $u\in V_{\sigma}$ for some $\sigma$. 

    % Reordering the former boxes
    Construct an order $\rho$ from $\sigma$ as follows:
    First, $\rho(j)=j$ for each $j\leq i$. 
    Second, for all $j$ and $k$ with $j> i$ and $k> i$, $\sigma^{-1}(j)\leq \sigma^{-1}(k)$ if and only if $\rho^{-1}(j)\leq \rho^{-1}(k)$. 
    Intuitively, $\rho$ drops the ranking of each $j>i$ from $\sigma$ on the end, and then reorder the remaining first $i$ boxes $j\leq i$ in the order of indices. 

    % An auxiliary vector
    Here, consider a vector $v$ such that $v_{\rho(j)}=\overline{v}_{\rho,j}$ for each $j\leq i$ and $v_{\rho(j)}=\underline{v}_{\rho,j}$ for each $j>i$. 
    Note that $v\in V_{\rho}$.
    Then, define an auxiliary vector $r^{+}$ such that $r^{+}_{j}=\max\{r_{j},0\}$ for each $j$. 
    By an analogous argument with the proof of Theorem \ref{thm: information design positive weight}, since $r^{+}\geq 0$, we can show $r^{+}\cdot (v-u)\geq 0$. 
    Since $u\in V_{\sigma}$ implies $u_{\sigma(j)}\geq \underline{v}_{\sigma,j}$ for each $j$, the construction of $\rho$ implies that $u_{j}\geq v_{j}$ for each $j>i$. 
    Therefore, $r_{j}<0$ implies $v_{j}-u_{j}\leq 0$, implying $r\cdot (v-u)\geq 0$. 

    % Reordering the latter boxes
    Finally, we prove $r\cdot (\tilde{v}-u)\geq 0$. 
    If we have $\rho(j)=j$ for each $j>i$, then we have $\rho=\pi$ and thus $v=\tilde{v}$, finishing the proof. 
    Suppose not. 
    Then, by assumption, there exists $j>i$ such that
    \begin{align*}
        r_{\rho(j+1)}\cdot \frac{\underline{z}-\underline{x}}{\underline{z}_{\rho(j+1)}-\underline{x}}
        \geq 
        r_{\rho(j)}\cdot \frac{\underline{z}-\underline{x}}{\underline{z}_{\rho(j)}-\underline{x}}.
    \end{align*}
    Consider an order $\tau$ that swaps the order between $\rho(j)$ and $\rho(j+1)$ while keeping other relative orders unchanged. 
    Then, define a vector $\hat{v}$ such that $\hat{v}_{\tau(j)}=\overline{v}_{\tau,j}$ for each $j\leq i$ and $\hat{v}_{\tau(j)}=\underline{v}_{\tau,j}$ for each $j>i$. 
    Note that $\hat{v}\in V_{\tau}$. 
    Then, 
    \begin{align*}
        r\cdot (\hat{v}-v) 
        &= \sum_{k>i} r_{\tau(k)}\hat{v}_{\tau(k)} - \sum_{k>i} r_{\rho(k)}v_{\rho(k)} \\
        &= (r_{\tau(j)}\hat{v}_{\tau(j)}+r_{\tau(j+1)}\hat{v}_{\tau(j+1)}) - (r_{\rho(j)}v_{\rho(j)}+r_{\rho(j+1)}v_{\rho(j+1)}) \\
        &= (r_{\rho(j+1)}\hat{v}_{\tau(j)}+r_{\rho(j)}\hat{v}_{\tau(j+1)}) - (r_{\rho(j)}v_{\rho(j)}+r_{\rho(j+1)}v_{\rho(j+1)}) \\
        &= \hat{v}_{\tau(j)} \cdot (r_{\rho(j+1)}+r_{\rho(j)}q_{\tau(j)}) - v_{\rho(j)} \cdot (r_{\rho(j)}+r_{\rho(j+1)}q_{\rho(j)}) \\
        &= \hat{v}_{\tau(j)} \cdot (r_{\rho(j+1)}+r_{\rho(j)}q_{\rho(j+1)}) - v_{\rho(j)} \cdot (r_{\rho(j)}+r_{\rho(j+1)}q_{\rho(j)}) \\
        &= v_{\rho(j)} \cdot \left\{r_{\rho(j+1)}\cdot \frac{\underline{z}_{\rho(j)}-\underline{x}}{\underline{z}-\underline{x}} - r_{\rho(j)}\cdot \frac{\underline{z}_{\rho(j+1)}-\underline{x}}{\underline{z}-\underline{x}}\right\} \\
        &\geq 0,
    \end{align*}
    where the last inequality follows by assumption and other equalities follow by the construction of $\tau$.
    Therefore, swapping the order achieves a higher value. 
    Repeating this procedure, we eventually get $r\cdot (\tilde{v}-u)\geq 0$.

    The latter statement on the uniqueness of the maximizer will follow from the above reasoning along with Theorem \ref{thm: information design positive weight}.
\end{proof}

% Step 5: Characterization 
The next result shows that a single CDF profile $\overline{F}$ induces all vertexes in the convex hull of $V$. 
Formally, each $\overline{F}_{i}$ is defined as follows: 
\begin{align*}
    \overline{F}_{i}(y_{i}) = 
    \begin{cases}
        q_{i} & \quad \text{if} \quad y_{i}<\underline{z},\\
        p_{i} & \quad \text{if} \quad \underline{z}\leq y_{i}< \overline{x},\\
        1 & \quad \text{if} \quad \overline{x} \leq y_{i}.
    \end{cases}
\end{align*} 
To see that this is a well-defined CDF, note first that $0\leq q_{i}$ and $p_{i}\leq 1$ follow by definition. 
Therefore, it is sufficient to check that $q_{i}\leq p_{i}$. 
Assumption \ref{assum: regularity} implies that $\underline{z}\leq \overline{z}\leq \overline{z}_{i}$. 
Note that $\underline{z}_{i}=\mu_{i}-c_{i}$ implies two expressions for $\overline{z}_{i}$ as follows.
\begin{align*}
    \overline{z}_{i} 
    = \overline{x}-\frac{c_{i}}{\lambda_{i}}
    = \frac{\underline{z}_{i}}{\lambda_{i}}-\frac{\underline{x}(1-\lambda_{i})}{\lambda_{i}}.
\end{align*}
The first expression and $\underline{z}\leq \overline{z}_{i}$ imply $\lambda_{i}\geq c_{i}/(\overline{x}-\underline{z})$. 
Also, the second expression and $\underline{z}\leq \overline{z}_{i}$ imply $\lambda_{i}\leq (\underline{z}_{i}-\underline{x})/(\underline{z}-\underline{x})$.
These two inequalities with appropriate rearrangements induce
\begin{align*}
    \frac{\underline{z}-\underline{z}_{i}}{\underline{z}-\underline{x}} 
    \leq 
    1-\lambda_{i}
    \leq 
    1-\frac{c_{i}}{\overline{x}-\underline{z}},
\end{align*}
which is equivalent to $ q_{i}\leq p_{i}$.
Therefore, the function $\overline{F}_{i}$ is a well-defined CDF. 
Moreover, an analogous argument as the one in the proof of Lemma \ref{lem: lower bound}, one can see by direct computations that $\overline{F}_{i}$ has mean $\mu_{i}$ and reservation value $\underline{z}$, respectively.

% ----------------------------------
\begin{lemma}\label{lem: sufficient condition}
    If $v$ is a vertex of the convex hull of $V$, then $\overline{F}$ induces $v$.
\end{lemma}

\begin{proof} 
    % Construction
    Assume w.l.g. that $v_{1}\geq v_{2}\geq\cdots\geq v_{n}$.
    Since $v$ is a vertex, Lemma \ref{lem: vertex any direction} along with this order implies that $v = (\overline{v}_{1},\dots,\overline{v}_{i},\underline{v}_{i+1},\dots,\underline{v}_{n})$ for some $i$. 
    
    Now, the following strategy is consistent with Pandora's rule under the distribution profile $\overline{F}$:
    Open boxes in the ascending order of indices. 
    Suppose that box $j-1$ is opened and $y$ is the maximum value of past observations. 
    If $j\leq i$, then we let Pandora continue search and proceed to open box $j$, if and only if $y\leq \underline{z}$. 
    If $j>i$, we let Pandora open box $j$ if and only if $y<\underline{z}$.
    
    To see that the above strategy induces $v$, take any $j$. 
    If $j\leq i$, then box $j$ is opened with probability $p_{1}\cdots p_{j-1}=\overline{v}_{j}$.
    If $j>i$, the probability that Pandora opens box $j$ equals $q_{1}\dots q_{j-1}=\underline{v}_{j}$. 
    Hence, $\overline{F}$ induces $v$.
\end{proof}

% Step 6: Welfare minimization 
Theorem \ref{thm: feasible search behavior} argues that the CDF profile $\overline{F}$ minimizes Pandora's welfare. 
The next lemma proves what generalizes this statement for later use.

% ----------------------------------
\begin{lemma}\label{lem: welfare min}
    Take any feasible $F$ such that the reservation value $z_{i}$ of every box $i$ satisfies $z_{i}\leq \underline{z}$. 
    Then, Pandora's expected payoff is $\underline{z}$, which is the minimum possible expected payoff among all feasible distributions.
\end{lemma}

\begin{proof}
    Note that there must exist a box $i$ such that $z_{i}=\underline{z}_{i}=\underline{z}$. 
    Then, for box $i$, we have $F_{i}(\underline{z}-)=0$, because otherwise we have
    \begin{align*}
        \overline{x}-\mu_{i} 
        = \int_{\underline{x}}^{\overline{x}}F_{i}(y_{i})dy_{i}
        > \int_{\underline{z}}^{\overline{x}}F_{i}(y_{i})dy_{i}
        = \overline{x}-\underline{z}-c_{i},
    \end{align*}
    which together with $\underline{z}=\underline{z}_{i}=\mu_{i}-c_{i}$ implies $0>0$, a contradiction. 
    Therefore, it is optimal for Pandora to first open box $i$ and stop search immediately afterwards. 
    This gives the expected payoff of $\underline{z}_{i}=\underline{z}$. 
    Moreover, this strategy is feasible under any distribution and provides the payoff $\underline{z}$, and thus, $\underline{z}$ is the minimum possible expected payoff among all feasible distributions.
\end{proof}

% Final Step
Finally, we are ready to prove Theorem \ref{thm: feasible search behavior}.

% ----------------------------------
\begin{proof}[Proof of Theorem \ref{thm: feasible search behavior}] 
    First, we prove that $V^{*}$ is the convex hull of $V$ and $\overline{F}$ induces all points in the set.
    If $v$ is a feasible search behavior, then Lemma \ref{lem: necessary condition} shows that $v$ is in the convex hull of $V$. 
    Conversely, take any $v$ in the convex hull of $V$. 
    This is a compact and convex set, and thus, Krein–Milman theorem implies $v$ is a convex combination of the extreme points.  
    By Lemma \ref{lem: sufficient condition}, for each extreme point $u$, which is a vertex as the convex hull of $V$ is a polytope, there exists an optimal strategy which induces $u$ under $\overline{F}$. 
    Therefore, an appropriate mixture of these strategies induces $v$ under $\overline{F}$. 
    Hence, $v$ is a feasible search behavior. 

    Second, we prove the remaining statements. 
    Take any $v$ with $v_{1}\geq\cdots\geq v_{n}$.
    If $v$ is a vertex of $V^{*}$, then Lemma \ref{lem: vertex any direction} shows that $v = (\overline{v}_{1},\dots,\overline{v}_{i},\underline{v}_{i+1},\dots,\underline{v}_{n})$ for some $i$. 
    Lemma \ref{lem: vertex any direction} also proves the other direction. 
    Then, Lemma \ref{lem: welfare min} implies the last statement, completing the proof.
\end{proof}

% =======================================================================
% =======================================================================
\subsection{Proof of Theorem \ref{thm: information design positive weight}}\label{subsec: appendix information design}

% Description
To prove Theorem \ref{thm: information design positive weight}, we consider the information design problem of choosing a feasible CDF profile $F$ that maximizes the induced weighted search length
\begin{align*}
    \max_{\pi\in\Pi_{F}}\left\{\sum_{k=1}^{n}r_{\pi(k)}\prod_{i=1}^{k-1}F_{\pi(i)}(z_{\pi(k)})\right\},
\end{align*}
where $z_{k}$ is the reservation value under $F_{k}$ for each $k$ and $\Pi_{F}$ is the set of all orders $\pi$ such that $z_{\pi(1)}\geq \cdots\geq z_{\pi(n)}$.
Note that we assume Pandora breaks ties in favor of continuing search. 
This is without loss of generality as we assume $r\geq 0$. 
We prove that an optimal CDF profile induces $\overline{v}$.

We begin with the following preliminary result. 

%%%%%%%%%%%%%%%%%%%%%%%%%%%%%%%%
%%%%%%%%%%%%%%%%%%%%%%%%%%%%%%%%
\begin{lemma}\label{lem: point max}
Take any $z\in [\underline{z}_{i},\overline{z}_{i}]$. 
Then, there exists a solution to the problem 
\begin{align*}
    1-\frac{c_{i}}{\overline{x}-z} = \max_{F_{i}:z_{i}\geq z} F_{i}(z),
\end{align*}
for any $i$, where $z_{i}$ is the reservation value under a feasible distribution $F_{i}$. 
The reservation value under the solution equals $z$.
\end{lemma}
%%%%%%%%%%%%%%%%%%%%%%%%%%%%%%%%
%%%%%%%%%%%%%%%%%%%%%%%%%%%%%%%%

\begin{proof}
Take any feasible distribution $F_{i}$ for box $i$ and the associated reservation value $z_{i}$ such that $z_{i}\geq z$. 
The definition of the reservation value implies
\begin{align*}
    c_{i} = \overline{x}-z_{i}-\int_{z_{i}}^{\overline{x}}F_{i}(y_{i})dy_{i} 
    & = \int_{z_{i}}^{\overline{x}}(1-F_{i}(y_{i}))dy_{i} \\
    & \leq \int_{z}^{\overline{x}}(1-F_{i}(y_{i}))dy_{i}
    \leq (\overline{x}-z)(1-F_{i}(z)),
\end{align*}
where the inequalities follow because $z_{i}\geq z$ and thus $F_{i}(z)\leq F_{i}(z_{i})$.
Rearranging terms, $1-c_{i}/(\overline{x}-z)\geq F_{i}(z)$ holds. 
Therefore, the value of the maximization problem is bounded above by $1-c_{i}/(\overline{x}-z)$.

Now, consider a distribution $G_{i}$ whose support equals $\{l_{i},\overline{x}\}$ with the probability that value $l_{i}$ realizes being equal to $\tilde{p}_{i}=1-c_{i}/(\overline{x}-z)$, where
\begin{align*}
    l_{i} = \frac{(\overline{x}-z)\mu_{i}-c_{i}\overline{x}}{\overline{x}-z-c_{i}} = \mu_{i}-c_{i}\cdot \frac{\overline{x}-\mu_{i}}{\overline{x}-z-c_{i}}.
\end{align*}
One can see that $\tilde{p}_{i}l_{i}+(1-\tilde{p}_{i})\overline{x}=\mu_{i}$ and thus $G_{i}$ is feasible. 
Moreover, a direct calculation using $z\geq \underline{z}_{i}\geq l_{i}$ shows that the reservation value equals $z\geq l_{i}$. 
Then, $G_{i}(z)=\tilde{p}_{i}=1-c_{i}/(\overline{x}-z)$ completes the proof.
\end{proof}

The next lemma derives the designer-optimal search order for a special class of distribution profiles such that the reservation value of every box is the same. 

%%%%%%%%%%%%%%%%%%%%%%%%%%%%%%%%
%%%%%%%%%%%%%%%%%%%%%%%%%%%%%%%%
\begin{lemma}\label{lem: optimal order}
Fix a feasible distribution profile $F$ such that the associated reservation value of every box $i$ equals some constant $z$ and $F_{i}(z)=1-c_{i}/(\overline{x}-z)$. 
Then, it is
% uniquely
optimal to let Pandora search in the descending order of $r_{i}(\overline{x}-z)/c_{i}$. 
\end{lemma}
%%%%%%%%%%%%%%%%%%%%%%%%%%%%%%%%
%%%%%%%%%%%%%%%%%%%%%%%%%%%%%%%%

\begin{proof}
Suppose that Pandora opens boxes in the ascending order of $i$. 
Pick any $m<n$. 
As the reservation value of every box is the common constant $z$, it is also optimal for Pandora to swap the order of opening box $m$ and box $m+1$. 
This alternative strategy yields the following objective value:
\begin{align*}
    & \sum_{k=1}^{m-1}r_{k} \prod_{i=1}^{k-1}F_{i}(z) + r_{m+1}\prod_{i=1}^{m-1}F_{i}(z) + r_{m}\prod_{i=1}^{m-1}F_{i}(z)\cdot F_{m+1}(z) + \sum_{k=m+2}^{n}r_{k}\prod_{i=1}^{k-1}F_{i}(z) \\
    = &  \sum_{k=1}^{n}r_{k}\prod_{i=1}^{k-1}F_{i}(z) - \prod_{i=1}^{m-1}F_{i}(z)\left\{ r_{m} - r_{m}F_{m+1}(z) + r_{m+1}F_{m}(z) - r_{m+1} \right\} \\
    = & \sum_{k=1}^{n}r_{k}\prod_{i=1}^{k-1}F_{i}(z) - \prod_{i=1}^{m-1}F_{i}(z)\left\{ r_{m}\frac{c_{m+1}}{\overline{x}-z} -r_{m+1}\frac{c_{m}}{\overline{x}-z} \right\},
\end{align*}
where the last equality follows by assumption $F_{i}(z)=1-c_{i}/(\overline{x}-z)$.

If the original search order maximizes the objective value, the second term of the last expression must be non-negative. 
% This means that $r_{m}c_{m+1}\geq r_{m+1}c_{m}$ must hold.
Therefore, $r_{m}(\overline{x}-z)/c_{m}\geq r_{m+1}(\overline{x}-z)/c_{m+1}$ using $c_{m}>0$ and $c_{m+1}>0$. 
\end{proof}

The following sequence of lemmas narrows down the candidates of optimal distribution profiles.
First, Lemma \ref{lem: finite support} implies the existence of an optimal distribution with finite support. 
In particular, whenever there exists an optimal distribution profile under which the reservation values are homogeneous, that suggests that a binary distribution profile is optimal. 

%%%%%%%%%%%%%%%%%%%%%%%%%%%%%%%%
%%%%%%%%%%%%%%%%%%%%%%%%%%%%%%%%
\begin{lemma}\label{lem: finite support}
Take any feasible $F$ under which $z_{1}\geq z_{2}\geq \cdots\geq z_{n}$. 
Then, there exists a feasible $G$ that generates the same objective and reservation values such that the finite support $\{a_{i-1},a_{i},\dots,a_{n}\}$ of posterior means under $G_{i}$ satisfies $z_{m+1}\leq a_{m}\leq z_{m}$ for $m\geq i$ and $z_{i}\leq a_{i-1}\leq \overline{x}$, where $z_{n+1}=\underline{x}$. 
\end{lemma}
%%%%%%%%%%%%%%%%%%%%%%%%%%%%%%%%
%%%%%%%%%%%%%%%%%%%%%%%%%%%%%%%%

\begin{proof}
The proof is by construction. 
Fix any index $i$. 
To begin with, we define $n-i+2$ parameters $a_{i-1}\geq a_{i}\geq \dots\geq a_{n}$ according to the following procedure.
Let $a_{i-1}$ be a solution to the equation
\begin{align*}
    \int_{z_{i}}^{\overline{x}}F_{i}(y_{i})dy_{i} = (\overline{x}-a_{i-1}) + (a_{i-1}-z_{i})F_{i}(z_{i}),
\end{align*}
where $z_{i}\leq a_{i-1}\leq \overline{x}$. 
Note that such an $a_{i-1}$ exists because the right-hand side is continuous in $a_{i-1}$ and takes values from $(\overline{x}-z_{i})$ to $(\overline{x}-z_{i})F_{i}(z_{i})$. 
Similarly, for each $m\geq i$, let $a_{m}$ with $z_{m+1}\leq a_{m}\leq z_{m}$ be a solution to the equation
\begin{align*}
    \int_{z_{m+1}}^{z_{m}}F_{i}(y_{i})dy_{i} = (z_{m}-a_{m})F_{i}(z_{m}) + (a_{m}-z_{m+1})F_{i}(z_{m+1}),
\end{align*}
which also exists by the same argument, where $z_{n+1}=\underline{x}$.
For $m=n$, the term $F_{i}(z_{n+1})$ in the preceding equation is replaced by $F_{i}(\underline{x}-)=0$.

Now, define $G_{i}$ as follows. 
Let $G_{i}(y_{i})=1$ for all $y_{i}\geq a_{i-1}$ and $G_{i}(y_{i})=0$ for all $y_{i}<a_{n}$.  
Then, for each $y_{i}$ with $a_{m}\leq y_{i}< a_{m-1}$, define $G_{i}(y_{i})=F_{i}(z_{m})$.
It follows by construction that 
\begin{align*}
    & \int_{\underline{x}}^{\overline{x}}G_{i}(y_{i})dy_{i} \\
     = & \int_{\underline{x}}^{z_{n}}G_{i}(y_{i})dy_{i} + \sum_{m=i+1}^{n}\int_{z_{m}}^{z_{m-1}}G_{i}(y_{i})dy_{i} + \int_{z_{i}}^{\overline{x}}G_{i}(y_{i})dy_{i} \\
     = & \int_{\underline{x}}^{z_{n}}F_{i}(y_{i})dy_{i} + \sum_{m=i+1}^{n}\int_{z_{m}}^{z_{m-1}}F_{i}(y_{i})dy_{i} + \int_{z_{i}}^{\overline{x}}F_{i}(y_{i})dy_{i} \\
     = & \int_{\underline{x}}^{\overline{x}}F_{i}(y_{i})dy_{i},
\end{align*}
where the last term equals $\overline{x}-\mu_{i}$.
This implies that $G_{i}$ has mean $\mu_{i}$, and therefore, $G_{i}$ is feasible.

The reservation value under $G_{i}$ is equal to $z_{i}$ by the construction of $a_{i-1}$. 
Moreover, we have $F_{i}(z_{m})=G_{i}(z_{m})$ for all $m>i$ by the definitions of $a_{m}$ for $m\geq i$. 
These two equivalences imply that replacing $F_{i}$ with $G_{i}$ generates the same designer's expected payoff. 

Finally, the support of $G_{i}$ equals $\{a_{i-1},a_{i},\dots,a_{n}\}$ by construction, which clearly satisfies the stated condition. 
Repeating this procedure for each box $i$, we complete the proof.
\end{proof}

The following lemma is the key to our result.
Intuitively, the expression in the lemma is proportional to the probability that the last $n-m+1$ boxes are opened. 
The monotonicity of this function eventually means that the weighted search length weakly increases when we increase the reservation values of the last $n-m+1$ boxes while decreasing the reservation value of a previously inspected box. 

%%%%%%%%%%%%%%%%%%%%%%%%%%%%%%%%
%%%%%%%%%%%%%%%%%%%%%%%%%%%%%%%%
\begin{lemma}\label{lem: joint monotone}
Take any real number $a\leq \overline{x}$ and a box index $m\leq n$ with $m\geq 2$. 
Suppose that we have $r_{m}(\overline{x}-\underline{z})/c_{m}\geq r_{m+1}(\overline{x}-\underline{z})/c_{m+1}\geq \cdots \geq r_{n}(\overline{x}-\underline{z})/c_{n}$.
Define a function $g$ for $z$ with $z\leq \underline{z}$ and $z<a$, where
\begin{align*}
    g(z) = \frac{1}{a-z} \left\{ \sum_{k=m}^{n}r_{k}\prod_{i=m}^{k-1} \left( \frac{\overline{x}-z-c_{i}}{\overline{x}-z} \right) \right\}. 
\end{align*}
Then, $g$ is
% strictly
increasing in $z$ in the domain. 
\end{lemma}
%%%%%%%%%%%%%%%%%%%%%%%%%%%%%%%%
%%%%%%%%%%%%%%%%%%%%%%%%%%%%%%%%

\begin{proof}
We show that $dg(z)/dz\geq 0$ holds for each $z$ with $z\leq \underline{z}$ and $z<a$. 
In particular, we prove by mathematical induction with respect to $n\geq m$ that 
\begin{align*}
    \frac{dg(z)}{dz} \geq \frac{1}{(a-z)^{2}}\sum_{k=m}^{n}r_{k}\left\{ \prod_{i=m}^{n}\left( \frac{\overline{x}-z-c_{i}}{\overline{x}-z} \right)\cdot \frac{\overline{x}-z}{\overline{x}-z-c_{k}} \right\}
    \geq 
    0,
\end{align*}
provided that $r_{m}(\overline{x}-\underline{z})/c_{m}\geq \cdots \geq r_{n}(\overline{x}-\underline{z})/c_{n}$. 
This is trivial when $n=m$. 
Now, assume the inductive hypothesis at $n-1\geq m$.
Observe that the function $g$ is expressed as a summation of $n-m+1$ terms. 
The derivative of the last term equals
\begin{align*}
    & \frac{r_{n}}{(a-z)^{2}} \prod_{i=m}^{n-1}\left(\frac{\overline{x}-z-c_{i}}{\overline{x}-z}\right) - \frac{1}{a-z}\sum_{k=m}^{n-1}r_{n}\frac{c_{k}}{(\overline{x}-z)(\overline{x}-z-c_{k})}\prod_{i=m}^{n-1}\left(\frac{\overline{x}-z-c_{i}}{\overline{x}-z}\right) \\
    = & \frac{r_{n}}{(a-z)^{2}} \prod_{i=m}^{n-1}\left(\frac{\overline{x}-z-c_{i}}{\overline{x}-z}\right) - \frac{1}{a-z}\sum_{k=m}^{n-1}\frac{r_{n}c_{k}}{(\overline{x}-z)^{2}}\prod_{i=m}^{n-1}\left(\frac{\overline{x}-z-c_{i}}{\overline{x}-z}\right)\frac{\overline{x}-z}{\overline{x}-z-c_{k}} \\
    \geq & \frac{r_{n}}{(a-z)^{2}} \prod_{i=m}^{n-1}\left(\frac{\overline{x}-z-c_{i}}{\overline{x}-z}\right) \\
    &\quad\quad\quad - \frac{1}{a-z}\sum_{k=m}^{n-1}\frac{r_{k}c_{n}}{(\overline{x}-z)(\overline{x}-z)}\prod_{i=m}^{n-1}\left(\frac{\overline{x}-z-c_{i}}{\overline{x}-z}\right)\frac{\overline{x}-z}{\overline{x}-z-c_{k}} \\
    \geq & \frac{1}{(a-z)^{2}} \left\{ r_{n}\prod_{i=m}^{n-1}\left(\frac{\overline{x}-z-c_{i}}{\overline{x}-z}\right) - \sum_{k=m}^{n-1}\frac{r_{k}c_{n}}{\overline{x}-z}\prod_{i=m}^{n-1}\left(\frac{\overline{x}-z-c_{i}}{\overline{x}-z}\right)\frac{\overline{x}-z}{\overline{x}-z-c_{k}}\right\},
\end{align*}
where the first inequality follows from $r_{k}(\overline{x}-z)c_{n}\geq r_{n}(\overline{x}-z)c_{k}$ and the second inequality holds by $a\leq \overline{x}$. 
Note that, for the above inequalities to hold, we also use the fact that $\overline{x}-z-c_{i}\geq 0$ for $i\geq m$, which follows from $z\leq \underline{z}\leq \overline{z}_{i}=\overline{x}-c_{i}/\lambda_{i}$. 

From the above calculation, the inductive hypothesis for case $n-1$ implies that $dg(z)/dz$ multiplied by $(a-z)^{2}$ is bounded below by
\begin{align*}
    & \sum_{k=m}^{n-1}r_{k}\left\{\prod_{i=m}^{n-1}\left(\frac{\overline{x}-z-c_{i}}{\overline{x}-z}\right)\cdot\frac{\overline{x}-z}{\overline{x}-z-c_{k}}\right\} \\
    & \quad\quad\quad + \left\{r_{n}\prod_{i=m}^{n-1}\left(\frac{\overline{x}-z-c_{i}}{\overline{x}-z}\right) - \sum_{k=m}^{n-1}\frac{r_{k}c_{n}}{\overline{x}-z}\prod_{i=m}^{n-1}\left(\frac{\overline{x}-z-c_{i}}{\overline{x}-z}\right)\frac{\overline{x}-z}{\overline{x}-z-c_{k}}\right\} \\
    & = \sum_{k=m}^{n-1}r_{k}\frac{\overline{x}-z-c_{n}}{\overline{x}-z}\left\{\prod_{i=m}^{n-1}\left(\frac{\overline{x}-z-c_{i}}{\overline{x}-z}\right)\cdot\frac{\overline{x}-z}{\overline{x}-z-c_{k}}\right\} + r_{n}\prod_{i=m}^{n-1}\left(\frac{\overline{x}-z-c_{i}}{\overline{x}-z}\right) \\
    & = \sum_{k=m}^{n-1}r_{k}\prod_{i=m}^{n}\left(\frac{\overline{x}-z-c_{i}}{\overline{x}-z}\right)\cdot\frac{\overline{x}-z}{\overline{x}-z-c_{k}} \\
    &\quad\quad\quad + r_{n}\left\{\prod_{i=m}^{n-1}\left(\frac{\overline{x}-z-c_{i}}{\overline{x}-z}\right)\right\}\cdot\frac{\overline{x}-z-c_{n}}{\overline{x}-z}\cdot\frac{\overline{x}-z}{\overline{x}-z-c_{n}} \\
    & = \sum_{k=m}^{n}r_{k}\left\{\prod_{i=m}^{n}\left(\frac{\overline{x}-z-c_{i}}{\overline{x}-z} \right)\cdot \frac{\overline{x}-z}{\overline{x}-z-c_{k}} \right\}\geq 0,
\end{align*}
where the first equality holds by combining the first and third terms. 
The last inequality follows from $\overline{x}-z-c_{i}\geq 0$ for each $i\geq m$. 
Therefore, the hypothesis holds for any $n\geq m$. 
In particular, $dg(z)/dz\geq 0$. 
\end{proof}

Using the next lemma, it remains to prove that the reservation value of every box is weakly higher than $\underline{z}$ under an optimal CDF profile.

%%%%%%%%%%%%%%%%%%%%%%%%%%%%%%%%
%%%%%%%%%%%%%%%%%%%%%%%%%%%%%%%%
\begin{lemma}\label{lem: reservation value bound}
For any feasible distribution profile under which the reservation value of every box is weakly higher than $\underline{z}$, the distribution profile $\overline{F}$ yields a weakly higher objective value.
\end{lemma}
%%%%%%%%%%%%%%%%%%%%%%%%%%%%%%%%
%%%%%%%%%%%%%%%%%%%%%%%%%%%%%%%%

\begin{proof} 
Take any feasible distribution profile $F$ such that the reservation value $z_{i}$ of every box $i$ satisfies $z_{i}\geq \underline{z}$. 
Assume w.l.g. that $z_{1}\geq z_{2}\geq \cdots\geq z_{n}$ and Pandora opens boxes in the ascending order of $i$. 
From Lemma \ref{lem: point max} and $z_{k}\geq \underline{z}\geq \underline{z}_{i}$,  
\begin{align*}
    F_{i}(z_{k})
    \leq 1-\frac{c_{i}}{\overline{x}-z_{i}}
    \leq 1-\frac{c_{i}}{\overline{x}-z_{k}}
    \leq 1-\frac{c_{i}}{\overline{x}-\underline{z}} 
    = \overline{F}_{i}(\underline{z}),
\end{align*}
for any $i$ and $k$ with $k>i$.
Therefore, the objective value under $F$ satisfies
\begin{align*}
    \sum_{k=1}^{n}r_{k}\left( \prod_{i=1}^{k-1}F_{i}(z_{k}) \right)\leq \sum_{k=1}^{n}r_{k}\left( \prod_{i=1}^{k-1}\overline{F}_{i}(\underline{z}) \right).
\end{align*}
As the reservation values of all boxes under the distributions $\overline{F}$ equal $\underline{z}$, the right-hand side of the above inequality is exactly the objective value under $\overline{F}$ if Pandora opens boxes in the same order. 
This concludes the proof.
\end{proof}

The last lemma considers the remaining case where Lemma \ref{lem: reservation value bound} is not applicable, where we find that a binary-support distribution profile is optimal. 

%%%%%%%%%%%%%%%%%%%%%%%%%%%%%%%%
%%%%%%%%%%%%%%%%%%%%%%%%%%%%%%%%
\begin{lemma}\label{lem: optimal binary}
Suppose that the reservation value of some box is lower than $\underline{z}$ for any optimal CDF profile. 
Then, there exists an optimal CDF profile $F$ such that each $F_{i}$ has a binary support.
\end{lemma}
%%%%%%%%%%%%%%%%%%%%%%%%%%%%%%%%
%%%%%%%%%%%%%%%%%%%%%%%%%%%%%%%%

\begin{proof}
Take any feasible CDF profile $F$ that is optimal. 
Let $z_{i}$ be the reservation value under $F_{i}$. 
By assumption, $\underline{z}>z_{i}$ for some box $i$. 
Assume, without loss of generality, that $z_{1}\geq z_{2}\geq \cdots\geq z_{n}$.

Note that there exists an index $m\leq n$ such that $z_{m-1}>z_{m}=\dots=z_{n}$. 
If there is no such index $m\leq n$, then we have $\underline{z}>z_{n}=z_{i}\geq \underline{z}_{i}$ for all $i$, which implies a contradiction $\underline{z}>\underline{z}$. 
Additionally, we assume, without loss of generality, that there is no other optimal distribution profile whose reservation values satisfy $z_{1}\geq z_{2}\geq \cdots\geq z_{n}$ and $z_{m-1}=z_{m}=\dots=z_{n}$. 
Otherwise, we replace $F$ with such a distribution profile.

Suppose that there exists a box $l$ such that the realization from $F_{l}$ is not binary.
From $z_{m}=\dots=z_{n}$ and Lemma \ref{lem: finite support}, for each $i\geq m$, there exists a binary distribution $G_{i}$ such that replacing $F_{i}$ with $G_{i}$ yields the same objective value. 
Assume $l<m$. 
From Lemma \ref{lem: finite support}, let $F_{l}$ have a finite support $\{a_{l-1},\dots,a_{n}\}$ such that $z_{k+1}\leq a_{k}\leq z_{k}$ for all $k>l$ and $z_{l}\leq a_{l-1}\leq \overline{x}$, where $z_{n+1}=\underline{x}$. 
For simplicity, we also write $b_{k}=F_{l}(a_{k})$ for each $k\geq l-1$. 
Note that $b_{l-1}=1$. 

If $b_{i}=1$ or $b_{i}=b_{n}$ for all $i\geq l$, the realization of $F_{l}$ is binary.
Therefore, $1>b_{i}>b_{n}$ for some $i\geq l$. 
In the following, we prove that there exists a feasible distribution for box $l$ that yields the same objective along with $F_{i}$ for $i\neq l$ and the number of possible realizations is $|\{1,\dots,b_{n}\}|-1$. 
Repeating this process will yield a binary-support distribution for box $l$ while keeping optimality.

Now, we construct a class of feasible distributions $G_{l}^{z}$ for box $l$ parameterized by $z\leq z_{m-1}$. 
Let $i<m$ be the highest index such that $b_{i}>b_{n}$. 
We have $i\geq l$ as discussed above. 
Then, define $G_{l}^{z}$ as the feasible distribution for box $l$ such that the support of posterior means is $\{a_{l-1},\dots,a_{i},z\}$ and $G_{l}^{z}(a_{k})=b_{k}$ for each $k\leq i$. 
The feasibility of $G_{l}^{z}$ implies that $b(z)=G_{l}^{z}(z)$ satisfies
\begin{align*}
    b(z) 
    = \frac{\overline{x}-\mu_{l}-\int_{a}^{\overline{x}}G_{l}^{z}(x_{l})dx_{l}}{a-z} 
    = \frac{\overline{x}-\mu_{l}-\int_{a}^{\overline{x}}F_{l}(x_{l})dx_{l}}{a-z},
\end{align*}
where we write $a=a_{i}$ for convenience. 
The second equality follows by construction.
Note that $G_{l}^{z}$ is a well-defined CDF if $b(z)\leq b$, where $b=b_{i}$. 

Note that the numerator of $b(z)$ does not depend on the choice of $z$.
Moreover, as $a\leq a_{l}\leq  z_{l}\leq \overline{x}$ and $G_{l}^{z}(x_{l})=F_{l}(x_{l})$ for any $x_{l}$ with $a\leq x_{l}\leq \overline{x}$ by assumption, the reservation value under $G_{l}^{z}$ does not depend on $z$ and equals $z_{l}$ by definition. 
For each $z\leq z_{m-1}$ and $i\geq m$ with $z\geq \underline{z}_{i}$, we also let $G_{i}^{z}$ be the solution to the maximization in Lemma \ref{lem: point max} at $z$. 
Lemma \ref{lem: point max} indicates that $G_{i}^{z}(z)=1-c_{i}/(\overline{x}-z)$ and $z$ is the reservation value under $G_{i}^{z}$.

By definition, we have $F_{l}=G_{l}^{a_{n}}$. 
As $a_{n}\leq z_{n}$ and $b(z)$ is weakly increasing in $z$, the above observation shows that replacing $F_{l}$ with $G_{l}^{z_{n}}$ weakly increases the objective value for the same search order. 
Therefore, we assume $F_{l}=G_{l}^{z_{n}}$, and the objective value under $F$ is calculated as follows: 
\begin{align*}
    & \sum_{k=1}^{n}r_{k}\prod_{i=1}^{k-1}F_{i}(z_{k}) \\
    =& \sum_{k=1}^{m-1}r_{k}\prod_{i=1}^{k-1}F_{i}(z_{k}) 
    + \left\{\prod_{1\leq i<m: i\neq l}F_{i}(z_{n})\right\} F_{l}(z_{n}) \sum_{k=m}^{n}r_{k}\prod_{i=m}^{k-1}F_{i}(z_{n}) \\
    \leq& % <
    \sum_{k=1}^{m-1}r_{k}\prod_{i=1}^{k-1}F_{i}(z_{k}) 
    + \left\{\prod_{1\leq i<m: i\neq l}F_{i}(z_{n})\right\} b(z_{n}) \sum_{k=m}^{n}r_{k}\prod_{i=m}^{k-1} \left( \frac{\overline{x}-z_{n}-c_{i}}{\overline{x}-z_{n}} \right),
\end{align*}
where the first equality follows from $z_{m}=\cdots=z_{n}$ and the last inequality follows from Lemma \ref{lem: point max} with inequality $z_{n}=z_{i}\geq \underline{z}_{i}$ for each $i\geq m$. 

Here, suppose that the numerator of $b(z)$ above is $0$. 
Then, it follows from the choice of $a>z_{n}$ that the probability that Pandora opens boxes $i\geq m$ is $0$ under $F$. 
Therefore, for any feasible distribution $G_{i}$ for each box $i\geq m$, if its reservation value does not exceed $z_{m-1}$, replacing $F_{i}$ with $G_{i}$ results in the same objective value. 
% In particular, as $\overline{z}_{i}\geq \underline{z}$ for each $i\geq m$, we can pick $G_{i}$ whose reservation value equals $\max\{\underline{z},z_{m-1}\}$. 
% If $\underline{z}\geq z_{m-1}$, it implies the existence of optimal CDF profile under which all reservation values are weakly higher than $\underline{z}$, which is a contradiction. 
% However, $\underline{z}< z_{m-1}$ contradicts the choice of $m$. 
In particular, since $\underline{z}_{i}\leq z_{n}<\min\{\underline{z},z_{m-1}\}\leq\underline{z}\leq\overline{z}_{i}$ for each $i\geq m$, we can choose each $G_{i}$ to have reservation value $\min\{\underline{z},z_{m-1}\}$.
If $\underline{z}\geq z_{m-1}$, the resulting optimal CDF profile has equal reservation values for boxes $m-1,\dots,n$, contradicting the choice of $m$.
If $\underline{z}<z_{m-1}$, every reservation value in the resulting optimal CDF profile is weakly higher than $\underline{z}$, contradicting the assumption of this lemma.
Therefore, the numerator of $b(z)$ is strictly positive.

The optimality of $F_{i}$ implies that the above inequality on the objective value must hold with equality, because otherwise, we may replace $F_{i}$ with $G_{i}^{z_{n}}$ for all $i\geq m$ and increase the objective value.
Moreover, restricting attention to boxes $m,\dots,n$, Lemma \ref{lem: optimal order} shows that we have $r_{m}(\overline{x}-\underline{z})/c_{m}\geq \cdots \geq r_{n}(\overline{x}-\underline{z})/c_{n}$. 
Thus, because $z_{n}<\underline{z}$ and $z_{n}<a$ by assumption, Lemma \ref{lem: joint monotone} implies that the objective is bounded above by
\begin{align*}
    \sum_{k=1}^{m-1}r_{k}\prod_{i=1}^{k-1}F_{i}(z_{k}) 
    + \left\{\prod_{1\leq i<m: i\neq l}F_{i}(z_{n})\right\} G_{l}^{z}(z) \left\{ \sum_{k=m}^{n}r_{k}\prod_{i=m}^{k-1} G_{i}^{z}(z) \right\},
\end{align*}
where $z$ is the maximum value $z$ such that all $z\leq z_{m-1}$, $z\leq \underline{z}$, $z\leq a$, and $b(z)\leq b$ hold. 
Note that there exists such $z$ with $z<a$, which follows because the function $b(z)$ continuously increases towards $\infty$ as $z$ increases towards $a$.

Note that the reservation value of box $l$ under $G_{l}^{z}$ remains to be $z_{l}$. 
Moreover, we have $G_{l}^{z}(z_{k})\geq F_{l}(z_{k})$ for all $k$ by construction.
Therefore, as $F_{i}(z_{n})\leq F_{i}(z)$ for all $i<m$, the above inequality suggests that replacing $F_{l},F_{m},\dots,F_{n}$ with $G_{l}^{z},G_{m}^{z},\dots,G_{n}^{z}$ generates a weakly higher objective value.
If $z=z_{m-1}$, it contradicts the choice of $m$. 
If $z=\underline{z}$, the reservation value of every box is above $\underline{z}$, a contradiction.
Therefore, $z<z_{m-1}$ and $z<\underline{z}$. 
Then, $b(z)=b$ must hold by the choice of $z$. 

Finally, $b(z)=b$ implies that the number of realizations under $G_{l}^{z}$ equals $|\{1,\dots,b_{i}\}|=|\{1,\dots,b_{i},b_{n}\}|-1$, where $b_{i}=b$. 
As the above procedure can be iterated as long as the distribution is not binary, we eventually obtain an optimal binary distribution for box $l$ while keeping other distributions equal. 
Therefore, a binary distribution profile is optimal.
\end{proof}  

Finally, we prove Theorem \ref{thm: information design positive weight}.

%%%%%%%%%%%%%%%%%%%%%%%%%%%%%%%%%%%%%%%%%%%%%%%%%%
%%%%%%%%%%%%%%%%%%%%%%%%%%%%%%%%%%%%%%%%%%%%%%%%%%
\begin{proof}[Proof of Theorem \ref{thm: information design positive weight}]
Take any optimal CDF profile $F$. 
Let $z_{1},z_{2},\dots,z_{n}$ be the associated reservation values.
Lemmas \ref{lem: reservation value bound} and \ref{lem: optimal order} complete the proof if $z_{i}\geq \underline{z}$ for all $i$. 
Therefore, assume $z_{i}<\underline{z}$ for some $i$. 
Lemma \ref{lem: optimal binary} implies that we may assume w.l.g. that each $F_{i}$ has a binary support $\{l_{i},h_{i}\}$ with $l_{i}\leq h_{i}$.
Assume that Pandora opens boxes in the ascending order of $i$.

Take the maximal index $m$ such that $l_{i}\leq z_{m}$ for all $i\leq m-1$. 
First, suppose that $m<n$. 
Then, the probability that each box $i> m$ is opened is $0$. 
Now, either $z_{m}<\underline{z}$ or $z_{m}\geq \underline{z}$ holds. 
In each case, replacing each $F_{i}$ with a feasible distribution under which the reservation value equals $z_{m}$ or $\underline{z}$ achieves the same 
% we want to say "strcitly higher" in this place but it may be wrong if $l_{i} > \underline{z}$, is it possible ?? 
optimal designer's expected payoff. 
Note that each of these transformations is feasible because Assumption \ref{assum: regularity} implies $\overline{z}_{i}\geq \underline{z}$ for any $i> m$.

Two cases must be considered.
If $z_{m}\geq \underline{z}$ holds, it suggests the existence of an optimal CDF profile under which $z_{i}\geq \underline{z}$ for all $i$, in which case Lemma \ref{lem: reservation value bound} and Lemma \ref{lem: optimal order} complete the proof.
If $z_{m}<\underline{z}$ holds, it suggests the existence of an optimal CDF profile such that $l_{i}\leq z_{n}$ for all $i\leq n-1$.

Therefore, we assume that $m=n$. 
Now, with an abuse of notation, take the smallest index $m$ such that $z_{m}=\cdots=z_{n}$. 
We prove by seeking a contradiction that $m=1$ holds. 
Suppose that $m>1$.
Then, we have $z_{m-1}>z_{m}$ by the choice of $m$. 
If $z_{n}\geq \underline{z}$ holds, then again Lemma \ref{lem: reservation value bound} completes the proof. 
Therefore, assume that $z_{n}<\underline{z}$ holds. 
As in the proof of Lemma \ref{lem: optimal binary}, we also assume w.l.g. that there is no other optimal CDF profile such that the reservation values satisfy $z_{1}\geq z_{2}\geq \cdots\geq z_{n}$ and $z_{m-1}=z_{m}=\dots=z_{n}$. 

Here, we claim that we may assume $l_{m-1}=z_{n}$. 
Note that $l_{m-1}<h_{m-1}$. 
Otherwise, we must have $l_{m-1}=h_{m-1}=\mu_{m-1}$ and $z_{m-1}=\underline{z}_{m-1}<\mu_{m-1}$; hence, $l_{m-1}=\mu_{m-1}>z_{m-1}>z_{n}$, which is a contradiction.
Thus, from the feasibility of $F_{m-1}$ implying that
\begin{align*}
    F_{m-1}(z_{n})=F_{m-1}(l_{m-1})=\frac{h_{m-1}-\mu_{m-1}}{h_{m-1}-l_{m-1}},
\end{align*}  
the probability $F_{m-1}(l_{m-1})$ is increasing in $l_{m-1}$. 
Moreover, as each $F_{i}$ with $i<m-1$ is binary such that $l_{i}\leq z_{n}$, the probability $F_{i}(z_{m-1})$ is independent of $z_{m-1}$ as long as $z_{m-1}\geq z_{n}$. 

Let $G_{m-1}$ be the distribution of box $m-1$ that is identified with $\{l, h_{m-1}\}$, where $l$ is the maximum value such that both $l\leq z_{n}$ and the reservation value $z_{m-1}^{l}$ is no lower than $z_{n}$.
One can see that $z_{m-1}^{l}$ is continuously decreasing in $l$, and, thus, such an $l$ exists. 
Then, the above discussion implies that replacing $F_{m-1}$ with $G_{m-1}$ weakly increases the objective value. 
As $z_{m-1}^{l}=z_{n}$ contradicts the choice of $m$, we obtain $l=z_{n}$ by the choice of $l$. 
Therefore, we assume from the beginning that $l_{m-1}=z_{n}$ holds. 

Now, note that Lemma \ref{lem: point max} shows $F_{i}(z_{n})\leq 1-c_{i}/(\overline{x}-z_{n})$ for each $i\geq m$ because $z_{n}=z_{i}\geq \underline{z}_{i}$ holds. 
From the optimality of the given distributions, the objective value equals
\begin{align*}
    & \sum_{k=1}^{n}r_{k}\prod_{i=1}^{k-1}F_{i}(z_{k}) \\
    =& \sum_{k=1}^{m-1}r_{k}\prod_{i=1}^{k-1}F_{i}(z_{k}) 
    + \left\{\prod_{i=1}^{m-2}F_{i}(z_{n})\right\} \frac{h_{m-1}-\mu_{m-1}}{h_{m-1}-z_{n}} \left\{ \sum_{k=m}^{n}r_{k}\prod_{i=m}^{k-1} \left( \frac{\overline{x}-z_{n}-c_{i}}{\overline{x}-z_{n}} \right) \right\},
\end{align*}
where, by restricting attention to the set of boxes $i\geq m$, Lemma \ref{lem: optimal order} also implies the inequalities $r_{m}(\overline{x}-\underline{z})/c_{m}\geq \cdots \geq r_{n}(\overline{x}-\underline{z})/c_{n}$. 

If $l_{m-1}\geq \underline{z}_{m-1}$ holds, we have $z_{m-1}=\underline{z}_{m-1}$ by the definition of reservation values. 
As $z_{n}\geq l_{m-1}$ by assumption, this means $z_{m-1}=z_{n}$, which is a contradiction. 
Therefore, $l_{m-1}< \underline{z}_{m-1}$ holds, or equivalently, $z_{n}< \underline{z}_{m-1}$. 
Hence, as $\underline{z}_{m-1}<h_{m-1}$ and $\underline{z}_{m-1}\leq \underline{z}$ holds, Lemma \ref{lem: joint monotone} implies that the objective value is bounded above by
\begin{align*}
    \sum_{k=1}^{m-1}r_{k}\prod_{i=1}^{k-1}F_{i}(z_{k}) 
    + \left\{\prod_{i=1}^{m-2}F_{i}(\underline{z}_{m-1})\right\} \frac{h_{m-1}-\mu_{m-1}}{h_{m-1}-\underline{z}_{m-1}} \left\{ \sum_{k=m}^{n}r_{k}\prod_{i=m}^{k-1} \left( \frac{\overline{x}-\underline{z}_{m-1}-c_{i}}{\overline{x}-\underline{z}_{m-1}} \right) \right\}, 
\end{align*}
where we also use $F_{i}(\underline{z}_{m-1})\geq F_{i}(z_{n})$ for each $i<m$. 

Finally, let $G_{m-1}$ be the feasible distribution for box $m-1$ with support $\{\underline{z}_{m-1},h_{m-1}\}$. 
Then, for each $i\geq m$, let $G_{i}$ be the solution to the maximization in Lemma \ref{lem: point max} at $\underline{z}_{m-1}$, where we can actually apply Lemma \ref{lem: point max} from $\underline{z}_{m-1}> z_{n}= z_{i}\geq \underline{z}_{i}$. 
As $F_{i}(z_{m-1})=F_{i}(\underline{z}_{m-1})$ for each $i<m$ from $l_{i}\leq z_{n}\leq \underline{z}_{m-1}$, the above bound for the objective value implies that the CDF profile $F_{1},\dots,F_{m-2},G_{m-1},\dots,G_{n}$ is also optimal. % strcitly increases the designer's payoff (?)
Moreover, the associated reservation values satisfy $z_{m-1}=z_{m}=\cdots=z_{n}$. 
This is a contradiction.

In summary, we have $m=1$, which implies $z_{i}\geq \underline{z}$ for all $i$. 
Thus, Lemma \ref{lem: reservation value bound} proves that $\overline{F}$ is optimal. % uniquely optimal upto the quotient set by ignoring the last box. 
Since $\overline{F}$ induces $\overline{v}$ when Pandora opens boxes in the descending order of $r_{i}(\overline{x}-\underline{z})/c_{i}$, Lemma \ref{lem: optimal order} completes the proof of the first statement. 
The second statement on the uniqueness follows from the proof of Lemma \ref{lem: optimal order}.
\end{proof}

%%%%%%%%%%%%%%%%%%%%%%%%%%%%%%%%%%%%%%%%%%%%%%%%%%%%%%%%%%%%%%%%%%%%%%%%%%%%%%
\section{Proofs of Propositions \ref{prop: face}, \ref{prop: optimal}, \ref{prop: identified set}, and \ref{prop: comparative statics}}\label{sec:appendixb}

\begin{proof}[Proof of Proposition \ref{prop: face}]
    Since we show in the proof of Theorem \ref{thm: feasible search behavior} that the convex hull of $V$ is exactly equal to $V^{*}$, the proof is immediate from Lemma \ref{lem: vertex any direction}.
\end{proof}

\begin{proof}[Proof of Proposition \ref{prop: optimal}] 
    To see the former statement, note that Theorem \ref{thm: feasible search behavior} shows that the feasible search behaviors form a polytope, which is compact. 
    Since $g$ is upper-continuous, there exists $v$ that maximizes $g(v)$ among all feasible search behaviors. 
    Then, the latter statement in Theorem \ref{thm: feasible search behavior} completes the proof. 

    To prove the remaining statement, suppose that $g(v)=r\cdot v$ for some $r$.
    Take any $\pi\in\Pi_{r}$. 
    Proposition \ref{prop: face} shows that it is optimal for the designer to let Pandora induce $v$ such that, for some index $i$, $v_{\pi(j)}=\overline{v}_{\pi,j}$ for all $j\leq i$ and $v_{\pi(j)}=\underline{v}_{\pi,j}$ for all $j>i$. 
    Then, $v_{\pi(1)}\geq v_{\pi(2)}\geq \cdots\geq v_{\pi(n)}$ implies that Pandora opens boxes following the order $\pi$. 

    % Full implementation intro
    It remains to check that the search behavior $v$ described above is fully implementable. 
    Suppose that Assumption \ref{assum: regularity} holds with strict inequality, i.e. $\overline{z}>\underline{z}$.
    Assume as always that $\pi$ is the identity. 
    
    % Defining well-defined CDFs
    For each box $j$, consider the following function $F_{j}(y_{j};\varepsilon_{j}, \Delta_{j})$ parameterized by $\varepsilon_{j}>0$ and $\Delta_{j}>0$, each of which will be a sufficiently small number:
    \begin{align*}
    F_{j}(y_{j}; \varepsilon_{j}, \Delta_{j}) = 
    \begin{cases}
        q_{j}+\eta_{j}(\varepsilon_{j},\Delta_{j}) & \quad \text{if} \quad y_{j}<\underline{z}+\varepsilon_{j},\\
        p_{j}-\Delta_{j} & \quad \text{if} \quad \underline{z}+\varepsilon_{j}\leq y_{j}< \overline{x},\\
        1 & \quad \text{if} \quad \overline{x} \leq y_{j},
    \end{cases}
    \end{align*}     
    where $\eta_{j}(\varepsilon_{j},\Delta_{j})$ is defined such that the function as a distribution has mean $\mu_{j}$ for each parameter pair $\varepsilon_{j}$ and $\Delta_{j}$. 
    Note that $\eta_{j}(\varepsilon_{j},\Delta_{j})\geq 0$. 
    As discussed earlier, one can show that $p_{j}\geq q_{j}$ if and only if $\overline{z}\geq \underline{z}$. 
    Therefore, we have $p_{j}>q_{j}$ when $\overline{z}> \underline{z}$, which implies that $F_{j}(y_{j}; \varepsilon_{j}, \Delta_{j})$ is a well-defined, feasible CDF for sufficiently small $\varepsilon_{j}$ and $\Delta_{j}$. 

    % Constructing parameters
    We construct a profile of parameters such that the induced CDF profile and any optimal search strategy approximately induce the search behavior $v$. 
    Fix any sufficiently small numbers $\varepsilon_{+}$ and $(\varepsilon_{j})_{j>i}$ such that
    \begin{align*}
        \varepsilon_{+}>\varepsilon_{i+1}>\varepsilon_{i+2}>\cdots>\varepsilon_{n}>0.
    \end{align*}
    A computation shows that the reservation value $z_{j}$ under $F_{j}(y_{j}; \varepsilon_{+}, \Delta_{j})$ satisfies $z_{j}\geq \underline{z}+\varepsilon_{+}$ if and only if $\Delta_{j}\geq \underline{\Delta}_{j+}$, where 
    \begin{align*}
        \underline{\Delta}_{j+} \equiv \frac{c_{j}}{\overline{x}-\underline{z}-\varepsilon_{+}} - \frac{c_{j}}{\overline{x}-\underline{z}},
    \end{align*}
    which is arbitrarily close to zero if $\varepsilon_{+}$ is small enough. 
    % Hence, since the reservation value is strictly increasing and continuous in $\Delta_{j}$, there exists a profile of small positive numbers $(\hat{\Delta}_{j})_{j\leq i}$ such that $\hat{\Delta}_{j}\geq \Delta_{j+}$ for all $j\leq i$ and the associated reservation values satisfy
    Hence, since the reservation value is strictly increasing and continuous in $\Delta_{j}$, there exists a profile of small positive numbers $(\hat{\Delta}_{j})_{j\leq i}$ such that $\hat{\Delta}_{j}\geq \underline{\Delta}_{j+}$ for all $j\leq i$ and the associated reservation values satisfy
    \begin{align*}
        z_{1}>z_{2}>\cdots >z_{i}>\underline{z}+\varepsilon_{+}.
    \end{align*}
    % Likewise, for each box $j>i$, the reservation value $z_{j}$ under $F_{j}(y_{j}; \varepsilon_{j}, \Delta_{j})$ satisfies $z_{j}\geq \underline{z}+\varepsilon_{j}$ if and only if $\Delta_{j}\geq \underline{\Delta}_{j}$ for an appropriate threshold $\underline{\Delta}_{j}$, where $\underline{\Delta}_{j+}>\underline{\Delta}_{j}$ and $\underline{\Delta}_{i+1}>\underline{\Delta}_{i+2}>\cdots >\underline{\Delta}_{n}$.
    Likewise, for each box $j>i$, the reservation value $z_{j}$ under $F_{j}(y_{j}; \varepsilon_{j}, \Delta_{j})$ satisfies $z_{j}\geq \underline{z}+\varepsilon_{j}$ if and only if $\Delta_{j}\geq \underline{\Delta}_{j}$ for an appropriate threshold $\underline{\Delta}_{j}$, where $\underline{\Delta}_{j+}>\underline{\Delta}_{j}$.
    Thus, there exists $(\hat{\Delta}_{j})_{j> i}$ such that $\underline{\Delta}_{j+}\geq \hat{\Delta}_{j}\geq \underline{\Delta}_{j}$ for all $j> i$ and
    \begin{align*}
        \underline{z}+\varepsilon_{+}>z_{i+1}>\underline{z}+\varepsilon_{i+1}
        > z_{i+2} > \underline{z}+\varepsilon_{i+2}
        > \cdots 
        > z_{n} > \underline{z}+\varepsilon_{n}.
    \end{align*}
    Then, consider a CDF profile $F$ such that $F_{j}(y_{j})=F_{j}(y_{j};\varepsilon_{j},\hat{\Delta}_{j})$ for each box $j$ and value $y_{j}$, where we set $\varepsilon_{j}=\varepsilon_{+}$ for all $j\leq i$. 

    % Show full implementation
    Let $\hat{v}$ be any search behavior induced by $F$. 
    Under $F$, the reservation values are strictly ordered, and thus, any optimal search strategy must open boxes in the ascending order of indices. 
    Moreover, Pandora will never be indifferent between stopping and continuing search by construction, and therefore, 
    \begin{align*}
        \hat{v}_{j} 
        = \prod_{k<j} F_{k}(z_{j}),
    \end{align*}
    for each box $j$. 
    The construction of the CDF profile also implies that
    \begin{align*}
        F_{k}(z_{j}) 
        = 
        \begin{cases}
            F_{k}(\underline{z}+\varepsilon_{+}) = p_{k}-\hat{\Delta}_{k} \quad & \text{if} \quad j\leq i,\\
            F_{k}(\underline{z}+\varepsilon_{k}-) = q_{k}+\eta_{k}(\varepsilon_{k},\hat{\Delta}_{k}) \quad & \text{if} \quad j> i,
        \end{cases}
    \end{align*}
    for all $j$ and $k<j$, where the second case follows from $z_{j}<\underline{z}+\varepsilon_{k}$ and $F_{k}$ being constant in the region $[\underline{x},\underline{z}+\varepsilon_{k})$. 
    Finally, note that, if $\varepsilon_{+}$ and all $\varepsilon_{j}$ with $j>i$ are arbitrarily small, then $\hat{\Delta}_{j}$ and $\eta_{j}(\varepsilon_{j},\hat{\Delta}_{j})$ can also be made arbitrarily small by construction. 
    Therefore, $\hat{v}$ is arbitrarily close to $v$, which completes the proof.
\end{proof}

\begin{proof}[Proof of Proposition \ref{prop: identified set}]
    Note that for any symmetric primitive $\gamma\in \Gamma$, we have $p_{i}=p(\gamma)$ and $q_{i}=0$ for all $i$. 
    Therefore, up to a permutation of boxes, any vertex in $V^{*}(\gamma)$ takes the form $v(\gamma,i)=(1,p(\gamma),p(\gamma)^{2},\dots,p(\gamma)^{i-1},0,\dots,0)$ for some index $i$. 
    It is enough to show that for any parameter $\gamma\in \Gamma^{*}$, we have $v\in V^{*}(\gamma)$ if and only if $v$ satisfies the set of inequalities given in the statement. 

    % If part
    Suppose $v\in V^{*}(\gamma)$. 
    Clearly, we must have $\sum_{i=1}^{n}v_{i}\geq 1$. 
    Take any index $k$ and consider a weight $r$ such that $r_{j}=1$ if $j\leq k$ and $r_{j}=0$ otherwise. 
    Then, Proposition \ref{prop: face} applied to the vertex $v(\gamma,k)$ and the weight $r$ implies the $k$-th inequality in the statement.  

    % Only-if part
    Suppose next the set of inequalities. 
    By seeking a contradiction, let $v\notin V^{*}(\gamma)$. 
    Since Theorem \ref{thm: feasible search behavior} shows that $V^{*}(\gamma)$ is a polytope, the separating hyperplane theorem implies that there exists a non-zero weight $r$ such that $r\cdot v>r\cdot u$ for all $u\in V^{*}(\gamma)$. 
    Since $V^{*}(\gamma)$ is symmetric, assume w.l.g. $r_{1}\geq r_{2}\geq \cdots \geq r_{n}$. 
    If $r_{1}\leq 0$, then we have 
    \begin{align*}
        r_{1} 
        \geq r_{1}\cdot \sum_{i=1}^{n}v_{i} 
        \geq \sum_{i=1}^{n}r_{i}v_{i},
    \end{align*}
    where the second inequality follows from $r_{1}\geq r_{i}$ for all $i$. 
    However, for $u=(1,0,\dots,0)\in V^{*}(\gamma)$, we have $r_{1}=r\cdot u$, which is a contradiction. 
    Therefore, assume $r_{1}>0$. 

    Let $i$ be the largest index such that $r_{i}\geq 0$. 
    Note that $r_{1}>0$ implies that it is well-defined.
    Now, note that $q_{i}=0$ for all $i$. 
    Therefore, if $r_{k}<0$ for some box $k>i$, Proposition \ref{prop: face} and $k\neq 1$ imply that replacing $r_{k}$ with zero does not change the value $\max_{u\in V^{*}(\gamma)}\{r\cdot u\}$. 
    Since this operation only increases $r\cdot v$, we assume w.l.g. that $r_{k}=0$ for all $k>i$. 
    
    Note that, for each $k$, we have $r_{k}-r_{k+1}\geq 0$. 
    Therefore, the set of inequalities in the statement implies that, for each $k\in \{1,\dots,i\}$, we have 
    \begin{align*}
        % (r_{1}-r_{2}) \sum_{j\leq 1}v_{j} &\leq (r_{1}-r_{2}) \sum_{j\leq 1}p(\gamma)^{j-1}, \\
        % (r_{2}-r_{3}) \sum_{j\leq 2}v_{j} &\leq (r_{1}-r_{2}) \sum_{j\leq 2}p(\gamma)^{j-1}, \\
        % & \vdots \\
        (r_{k}-r_{k+1}) \sum_{j\leq k}v_{j} &\leq (r_{k}-r_{k+1}) \sum_{j\leq k}p(\gamma)^{j-1}.
    \end{align*}
    Note that $r_{i+1}=0$. 
    Hence, by summing up these inequalities, we obtain
    \begin{align*}
        \sum_{j=1}^{i}r_{j}v_{j} \leq \sum_{j=1}^{i}r_{j}p(\gamma)^{j-1}.
    \end{align*}
    The right-hand side of this inequality is the inner product of $r$ and the vertex $v(\gamma,i)$. 
    This is a contradiction, and therefore, we must have $v\in V^{*}(\gamma)$. 
\end{proof}

\begin{proof}[Proof of Proposition \ref{prop: comparative statics}]
    Recall that $\underline{z}_{i}=\mu_{i}-c_{i}$ and $\underline{z}=\max_{k}\{\underline{z}_{k}\}$. 
    Therefore, 
    \begin{align*}
        p_{j} = 1-\frac{c_{j}}{\overline{x}-\max_{k}\{\mu_{k}-c_{k}\}} 
        \quad \text{and} \quad
        q_{j} = 1-\frac{\mu_{j}-c_{j}-\underline{x}}{\max_{k}\{\mu_{k}-c_{k}\}-\underline{x}},
    \end{align*}
    for each $j$. 
    Note that $\mu_{i}=\lambda_{i}(\overline{x}-\underline{x})+\underline{x}$ is increasing in $\lambda_{i}$.
    Therefore, for any box $j$, if $\mu_{i}-c_{i}=\max_{k}\{\mu_{k}-c_{k}\}$ for some box $i$, then $p_{j}$ decreases and $q_{j}$ increases when $\lambda_{i}$ goes up.
    If $\mu_{i}-c_{i}\neq \max_{k}\{\mu_{k}-c_{k}\}$, then $q_{i}$ decreases while all the other variables $p_{j}$ and $q_{j}$ are kept constant when $\lambda_{i}$ goes up. 

    Since $V^{*}$ is the convex hull of the sets $V_{\pi}$ for all $\pi$, it is enough to prove the statement for each $V_{\pi}$ instead of directly looking at $V^{*}$.
    Take any $v\in V_{\pi}$ for some $\pi$. 
    By definition, we have $v_{\pi(i)} \geq \underline{v}_{\pi,i}$ and
    \begin{align*} 
        \sum_{j\leq i} \frac{c_{\pi(j)}}{\overline{x}-\underline{z}}v_{\pi(j)} &\leq \sum_{j\leq i} \frac{c_{\pi(j)}}{\overline{x}-\underline{z}}\overline{v}_{\pi,j},
    \end{align*}
    for each $i$. 
    By the definition of $\overline{v}_{\pi,j}$ and $\underline{v}_{\pi,j}$ together with the above discussion, these two inequalities are also satisfied for a smaller value of $\lambda_{i}$ when $\underline{z}_{i}=\underline{z}$ and for a larger value of $\lambda_{i}$ when $\underline{z}_{i}\neq \underline{z}$. 
    Therefore, if $\underline{z}_{i}=\underline{z}$, then $V_{\pi}$ expands when $\lambda_{i}$ goes down. 
    Likewise, if $\underline{z}_{i}\neq \underline{z}$, then $V_{\pi}$ expands when $\lambda_{i}$ goes up.

    Finally, suppose $\lambda_{i}=\lambda$ and $c_{i}=c$ for all $i$. 
    Then, $\mu_{i}=\mu$ for each box $i$, where $\mu=\lambda\overline{x}+(1-\lambda)\underline{x}$. 
    Therefore, $p_{i}=1-c/(\overline{x}-\mu+c)$ and $q_{i}=0$ for all $i$. 
    Note that $p_{i}$ is decreasing in $\lambda$ and in $c$.
    Therefore, an analogous argument as in the above paragraph finishes the proof.
\end{proof}

%%%%%%%%%%%%%%%%%%%%%%%%%%%%%%%%%%%%%%%%%%%%%%%%%%%%%%%%%%%%%%%%%%%%%%%%%%%%%%
% Change theorem numbering to A.1, A.2, etc.
\renewcommand{\theproposition}{C.\arabic{proposition}}
\setcounter{proposition}{0}
\section{Proofs of Theorem \ref{thm: joint distribution} and Proposition \ref{prop: joint distribution vertex}}\label{sec:appendixc} 

% Feasible distributions 
Here, we prove Theorem \ref{thm: joint distribution} and Proposition \ref{prop: joint distribution vertex} under general priors. 
Consider a general prior distribution $F_{i}^{0}$ over the interval $[\underline{x},\overline{x}]$ with mean $\mu_{i}$. 
In this case, \citet{blackwell1953equivalent} shows that a posterior mean distribution $F_{i}$ is feasible if and only if $F_{i}$ is a \textit{mean-preserving contraction} of $F_{i}^{0}$. 
That is, 
\begin{align*}
    \int_{\underline{x}}^{y_{i}}F_{i}(x_{i})dx_{i} \leq \int_{\underline{x}}^{y_{i}}F^{0}_{i}(x_{i})dx_{i},
\end{align*}
for each $y_{i}\in [\underline{x},\overline{x}]$, with equality at $y_{i}=\overline{x}$.\footnote{In our original setting of binary state space, we have $F^{0}_{i}(x_{i})=1-\lambda_{i}$ for all $x_{i}$ with $\underline{x}\leq x_{i}<\overline{x}$. 
Then, it is easy to see that $F_{i}$ is a mean-preserving contraction of $F^{0}_{i}$ if and only if the expectation under $F_{i}$ equals $\mu_{i}$, as explained in Section \ref{sec:model}.}
For being precise, let $\mathcal{F}_{i}$ be the set of all mean-preserving contractions of $F_{i}^{0}$ for each $i$. 

% Description
First, we prove the following lemma, which partly generalizes Lemma \ref{lem: point max}.

% -------------------------------------------
\begin{lemma}\label{lem: Lemma 1 general}
    Assume $\underline{z}_{i}=\underline{z}$ for all $i$.
    Then, for each box $i$, a binary-support distribution $F^{*}_{i}$ maximizes $F_{i}(z_{i})$ subject to $F_{i}\in \mathcal{F}_{i}$ where $z_{i}$ is the reservation value under $F_{i}$. 
    Moreover, $F^{*}_{i}$ has support $\{\underline{z},h_{i}\}$ for some $h_{i}\geq \underline{z}$ and has the reservation value of $\underline{z}$.
\end{lemma}

\begin{proof}[Proof of Lemma \ref{lem: Lemma 1 general}]
    To prove the first claim, take any distribution $F_{i}$ that solves the maximization. 
    Let $z_{i}$ be its reservation value. 
    Then, construct $F^{*}_{i}$ from $F_{i}$ by pooling values above and below $z_{i}$, respectively. 
    Formally, we have
    \begin{align*}
        F^{*}_{i}(y_{i}) 
        = 
        \begin{cases}
            0 \quad &\text{if} \quad y_{i} < l_{i} \\
            \tilde{p}_{i} \quad &\text{if} \quad l_{i}\leq y_{i}< h_{i}\\
            1 \quad &\text{if} \quad h_{i}\leq y_{i}
        \end{cases},
    \end{align*}
    where $l_{i}=\mathbb{E}_{y_{i}\sim F_{i}}[y_{i}|y_{i}\leq z_{i}]$ and $h_{i}=\mathbb{E}_{y_{i}\sim F_{i}}[y_{i}|y_{i}> z_{i}]$, with an appropriate $\tilde{p}_{i}$ such that $F^{*}_{i}$ has mean $\mu_{i}$. 
    Since $F_{i}\in\mathcal{F}_{i}$, we have $F^{*}_{i}\in\mathcal{F}_{i}$ by construction. 
    The construction also implies the reservation value under $F^{*}_{i}$ equals $z_{i}$. 
    Hence, $F^{*}_{i}$ also solves the maximization problem. 
    Note that $F^{*}_{i}$ has a binary support $\{l_{i},h_{i}\}$.

    To see the second claim, note that $l_{i}\leq \underline{z}$. 
    Then, it follows from the definition of reservation values that $z_{i}=\underline{z}$ if and only if $l_{i}=\underline{z}$. 
    Seeking a contradiction, assume $z_{i}>\underline{z}$. 

    Now, consider an alternative binary-support distribution $F^{**}_{i}$ that is characterized by a triplet $(l_{i}+\Delta_{l}, h_{i}; \tilde{p}_{i}+\Delta_{p})$, where $\Delta_{l}$ is a small enough positive number and $\Delta_{p}$ is defined so that $F^{**}_{i}$ has mean $\mu_{i}$. 
    Note that $\Delta_{p}>0$ by construction.
    If $\Delta_{l}$ is small enough, $\Delta_{p}$ is also sufficiently small, and hence $F^{**}_{i}$ is a well-defined CDF. 
    Moreover, using the facts that $F^{**}_{i}$ and $F^{*}_{i}$ have the same mean and that $\Delta_{l},\Delta_{p}>0$, one can see that $F^{**}_{i}$ is a mean-preserving contraction of $F^{*}_{i}$. 
    Therefore, $F^{**}_{i}\in\mathcal{F}_{i}$.

    Finally, let $z^{**}_{i}$ be the reservation value under $F^{**}_{i}$. 
    It follows that $z^{**}_{i}$ is arbitrarily close to $z_{i}$ when $\Delta_{l}$ and $\Delta_{p}$ are sufficiently small, hence $z^{**}_{i}>l_{i}+\Delta_{l}$. 
    This means that we have $F^{**}_{i}(z^{**}_{i})=\tilde{p}_{i}+\Delta_{p}>\tilde{p}_{i}=F^{*}_{i}(z_{i})$, a contradiction. 
    Therefore, $z_{i}=\underline{z}$ and thus $l_{i}=\underline{z}$.
\end{proof}

% Notation
Here, we introduce some notation for later use. 
For the distribution $F^{*}_{i}$ that we derive in Lemma \ref{lem: Lemma 1 general}, define $p^{*}_{i} = F^{*}_{i}(\underline{z})$. 
Then, let $v^{*}_{i}=\prod_{j<i}p^{*}_{j}$. 
Lemmas \ref{lem: point max} and \ref{lem: Lemma 1 general} imply $p_{i}^{*}=p_{i}$ when the priors are binary.
In general, the feasible set of distributions is narrower than under binary priors, hence $p_{i}^{*}\leq p_{i}$. 
Abusing notation, let $\overline{z}_{i}$ denote the reservation value under the prior $F_{i}^{0}$. 
Define $\overline{z}=\min_{i}\{\overline{z}_{i}\}$. 
When each $F_{i}^{0}$ has binary support on $\{\underline{x},\overline{x}\}$, both $\overline{z}_{i}$ and $\overline{z}$ coincide with the parametric values introduced in the main text. 

% Effective values
We rely on the results in \citet{armstrong2017ordered} and \citet{choi2018consumer} to prove our result. 
For a fixed distribution $F_{i}$ and its reservation value $z_{i}$ for box $i$, the \textit{effective value} is a random variable $w_{i}=\min\{y_{i},z_{i}\}$ with $y_{i}$ following $F_{i}$. 
\citet{armstrong2017ordered} and \citet{choi2018consumer} show that, given the realization $y_{i}$ of each box $i$, Pandora chooses a box with one of the highest effective values under optimal strategies. 
We typically use $H_{i}$ to denote the induced distribution of effective values. 

% -------------------------------------------
\begin{lemma}\label{lem: lower bound of choice}
    Suppose that $\underline{z}_{i}=\underline{z}$ for all $i$.
    Then, for any feasible pair $(v,d)$, we have $d_{i}\geq (1-p_{i}^{*})\cdot v_{i}$ for all $i$.
\end{lemma}

\begin{proof}
    Suppose that a distribution profile $F$ induces $(v,d)$.
    Let $z_{i}$ be the reservation value under $F_{i}$ for each box $i$.
    Since values are independent across boxes, the definition of Pandora's rule implies that the decision to open box $i$ is independent of its unobserved value $y_{i}$. 
    Moreover, by Pandora's rule, when box $i$ is opened, every previously sampled value and every remaining reservation value are weakly below $z_{i}$. 
    Thus, if $y_{i}>z_{i}$, it is strictly optimal to stop search and choose box $i$.
    Therefore,
    \begin{align*}
        d_{i} \geq (1-F_{i}(z_{i}))\cdot v_{i}\geq (1-p_{i}^{*})\cdot v_{i},
    \end{align*}
    where the last inequality follows from Lemma \ref{lem: Lemma 1 general}.
\end{proof}

Let $S_{\pi}$ be the set of all feasible pairs of search and choice behavior induced by some pure strategy with search order $\pi$. 
The next result lists a set of key conditions for feasible pairs.

% -------------------------------------------
\begin{lemma}\label{lem: feasible joint distribution necessary conditions}
    Suppose that $\underline{z}_{i}=\underline{z}$ for all $i$.
    Let $\pi$ be the identity permutation.
    Then, any $(v,d)\in S_{\pi}$ satisfies the following:
    \begin{itemize}
        \item[$\square$] $d_{i}\geq (1-p_{i}^{*})\cdot v_{i}$ for all $i$,
        \item[$\square$] $\sum_{i\leq j}d_{i} \geq 1-v_{j+1}$ for all $j$, and
        \item[$\square$] $p^{*}_{i}\cdot v_{i}\geq v_{i+1}$ for all $i$, 
    \end{itemize}
    where we set $v_{n+1}$ to be $0$.
\end{lemma}

\begin{proof}
    Lemma \ref{lem: lower bound of choice} implies the first inequality. 
    Consider a distribution profile $F$ and a pure strategy inducing $(v,d)$. 
    Pandora opens boxes in the order of indices by assumption. 
    Therefore, for any event where Pandora does not open box $j+1$, she chooses a box $i$ such that $i<j+1$. 
    This implies the second inequality. 
    Finally, note that for any box $i$ such that $v_{i}\neq 0$,
    \begin{align*}
        \frac{v_{i+1}}{v_{i}} 
        \leq F_{i}(z_{i+1}) 
        \leq F_{i}(z_{i})
        \leq p_{i}^{*},
    \end{align*}
    which implies the last inequality. 
    If $v_{i}=0$, then we have $v_{i+1}=0$, hence the last inequality. 
\end{proof}

% Intro
Then, we use this condition to prove that $F^{*}$ induces all points in $S_{\pi}$. 

% -------------------------------------------
\begin{lemma}\label{lem: neccesary conditions imply implementability}
    Suppose that $\underline{z}_{i}=\underline{z}$ for all $i$. 
    Let $\pi$ be the identity permutation.
    Then, the distribution profile $F^{*}$ induces any $(v,d)\in S_{\pi}$. 
\end{lemma}

\begin{proof}
    Take any $(v, d)\in S_{\pi}$.
    Lemma \ref{lem: feasible joint distribution necessary conditions} implies that we have the three conditions listed in the statement.
    Note that we clearly have $v_{1}=1$, $\sum_{i}d_{i}=1$, and $v_{i}\geq 0$ for all $i$.
    We construct an optimal strategy for Pandora that induces $(v, d)$ as the resulting pair of search and choice behavior.
    Note that every $F^{*}_{i}$ can only take one of two values, $\underline{z}$ or $h_{i}\in(\underline{z},\overline{x}]$.

    In the following two paragraphs, we construct Pandora's strategy.
    Let Pandora open boxes in the ascending order of box indices.
    First, we designate her stopping decision.
    Suppose Pandora opens a box $i$.
    If the sampled value is $h_{i}$, we let Pandora stop search and choose box $i$.
    If the sampled value is $\underline{z}$, we have Pandora stop search with probability
    \begin{align*}
        1-\frac{v_{i+1}}{p^{*}_{i}\cdot v_{i}} \in [0,1],
    \end{align*}
    where the inclusion follows by assumption whenever $p^{*}_{i}\cdot v_{i}>0$.
    At histories of probability zero, specify any optimal action.
    Note that, under this stopping strategy, Pandora opens box $i$ and stops search immediately afterwards with probability $v_{i}-v_{i+1}$, where we let $v_{n+1}=0$.

    Second, we also need to describe a box that we have Pandora choose when stopping search.
    Consider the following simultaneous equations:
    \begin{align*}
        \begin{bmatrix}
            d_{1}-(1-p_{1}^{*})v_{1} \\
            d_{2}-(1-p_{2}^{*})v_{2} \\
            \vdots \\
            d_{n}-(1-p_{n}^{*})v_{n}
        \end{bmatrix}
        =
        \begin{bmatrix}
            m_{11} & m_{12} & \dots & m_{1n}\\
            0 & m_{22} & \dots & m_{2n}\\
            \vdots &  \vdots & \ddots & \vdots \\
            0 & 0 & \dots & m_{nn}
        \end{bmatrix}
        \cdot
        \begin{bmatrix}
            p_{1}^{*}v_{1}-v_{2}\\
            p_{2}^{*}v_{2}-v_{3}\\
            \vdots\\
            p_{n}^{*}v_{n}
        \end{bmatrix}.
    \end{align*}
    Note that the first and the third conditions in Lemma \ref{lem: feasible joint distribution necessary conditions} imply that both vectors in the expression are nonnegative. Moreover, they have the same sum, and
    \begin{align*}
        \sum_{i\leq j}[d_{i}-(1-p_{i}^{*})v_{i}]
        \geq\sum_{i\leq j}[p_{i}^{*}v_{i}-v_{i+1}],
    \end{align*}
    for all $j$, with equality at $j=n$. This follows by
    \begin{align*}
        \sum_{i\leq j}[d_{i}-(1-p_{i}^{*})v_{i}]
        -\sum_{i\leq j}[p_{i}^{*}v_{i}-v_{i+1}]
        &= \sum_{i\leq j}[d_{i}-v_{i}+v_{i+1}] \\
        &\geq \sum_{i\leq j}[v_{i+1}-p_{i}^{*}v_{i}] \\ 
        & \geq 0,
    \end{align*}
    where the inequalities use the conditions in Lemma \ref{lem: feasible joint distribution necessary conditions}. The equality at $j=n$ is trivial. 
    Then, Theorem 3.4 in \citet{bruno2024note} together with this set of inequalities shows that there exists a solution $M=(m_{ij})_{i,j}\in \mR^{n\times n}$ to this system of simultaneous equations such that $m_{ij}=0$ if $i>j$ and $\sum_{i}m_{ij}=1$.\footnote{In other words, our problem is to find an upper-triangular stochastic matrix when the left-hand side vector \textit{majorizes} the right-hand side vector. \citet{bruno2024note} provide characterizations for majorization in terms of triangular matrices.}
    Now, if Pandora stops searching at box $j$ after sampling $\underline{z}$, we let Pandora choose box $i$ with probability $m_{ij}$.
    Note that this is a well-defined strategy by the above properties on matrix $M$.

    It remains to check that it is consistent with Pandora's rule and induces $(v,d)$ as the resulting pair of search and choice behavior.
    Note that any search order is optimal under $F^{*}$.
    Moreover, since every reservation value equals $\underline{z}$, any strategy such that Pandora stops search when she samples value $h_{i}$ is consistent with the optimal stopping.
    Finally, note that Pandora opens box $j$ only if all previously sampled values are $\underline{z}$.
    Therefore, if Pandora opens a box $i$ and samples value $h_{i}$, it is strictly optimal for Pandora to choose box $i$, and if Pandora samples value $\underline{z}$, any selection rule of choosing any box $j\leq i$ is optimal.
    In summary, the strategy described above is optimal.

    Finally, we check that the above strategy induces $(v,d)$.
    Note that, conditional on opening a box $i$ with $p^{*}_{i} v_{i}>0$, the above strategy lets Pandora stop search with probability
    \begin{align*}
        (1-p^{*}_{i}) + p^{*}_{i} \cdot \left(1-\frac{v_{i+1}}{p^{*}_{i}\cdot v_{i}}\right)
        = 1 - \frac{v_{i+1}}{v_{i}}.
    \end{align*}
    If $p^{*}_{i}\cdot v_{i}=0$, the strategy never proceeds from box $i$ to box $i+1$, and $v_{i+1}=0$ by assumption.
    With this expression, we can inductively compute the probability that each box $i$ is opened:
    Box $1$ is opened with probability $1$.
    Hence, box $2$ is opened with probability $1-(1-v_{2})=v_{2}$.
    Likewise, given that box $i$ is opened with probability $v_{i}>0$, box $i+1$ is opened with probability
    \begin{align*}
        v_{i} \cdot \left\{1- \left(1 - \frac{v_{i+1}}{v_{i}}\right)\right\} = v_{i+1}.
    \end{align*}
    Hence, each box $i$ is opened with probability $v_{i}$.
    Therefore, with probability $v_{i}-v_{i+1}$, Pandora stops searching immediately after opening box $i$.
    Hence, by the construction of $M$ and Pandora's selection strategy, the probability that each box $i$ is chosen equals
    \begin{align*}
        (1-p_{i}^{*})v_{i}
        +\sum_{j=i}^{n}m_{ij}\left(p_{j}^{*}v_{j}-v_{j+1}\right)=d_{i}.
    \end{align*}
    This completes the proof.
\end{proof}

Finally, we are ready to prove Theorem \ref{thm: joint distribution}.

% -------------------------------------------
\begin{proof}[Proof of Theorem \ref{thm: joint distribution}]
    Note that $S^{*}$ is the convex hull of $\bigcup_{\pi}S_{\pi}$. 
    Let $S^{**}$ be the set of all pairs that $F^{*}$ induces. 
    Lemma \ref{lem: neccesary conditions imply implementability} implies that $\bigcup_{\pi}S_{\pi} \subset S^{**}$.
    Therefore, since $S^{**}$ is clearly convex, we also have $S^{*} \subset S^{**}$. 
    Note that $S^{**}\subset S^{*}$ by definition, which implies $S^{*}=S^{**}$. 
    Therefore, the set of all feasible pairs $S^{*}$ is a polytope, and $F^{*}$ induces all points in $S^{*}$. 
    Lemma \ref{lem: Lemma 1 general} shows that the reservation value of every box equals $\underline{z}$, and thus, Lemma \ref{lem: welfare min} shows that $F^{*}$ minimizes Pandora's expected payoff among all feasible CDF profiles, which also equals $\underline{z}$. 
    This completes the proof.
\end{proof}

% Vertex characterization introduction
Since we consider general priors, we prove a generalized version of Proposition \ref{prop: joint distribution vertex}.  
Proving that each vertex is fully implementable may take more space than one might expect.  This is mainly because we allow for general priors. 
In the proof of Proposition \ref{prop: optimal}, we prove full implementation by perturbing the benchmark distribution $\overline{F}$. 
Under general priors, however, these perturbations must be crafted more carefully, as naive perturbations can violate the mean-preserving-contraction constraint. 

% ===========================================
\begin{proposition}\label{prop: joint distribution vertex appendix}
    Take any $(v,d)\in S^{*}$ such that $v_{1}\geq v_{2}\geq\cdots \geq v_{n}$. 
    Then, $(v,d)$ is a vertex of $S^{*}$ if and only if, for some $i$ and $j\leq i$, 
    \begin{align*}
        v_{k}
        &= 
        \begin{cases}
            v_{k}^{*} \quad & \text{if} \quad k\leq i,\\
            0 \quad & \text{otherwise},
        \end{cases}\\
        d_{k}
        &= 
        \begin{cases}
            (1-p_{k}^{*})\cdot v_{k}^{*} \quad & \text{if} \quad k\leq i \text{ and } k\neq j,\\
           (1-p_{k}^{*})\cdot v_{k}^{*}+p_{i}^{*}\cdot v_{i}^{*} \quad & \text{if} \quad k=j, \\
            0 \quad & \text{otherwise},
        \end{cases}
    \end{align*}
    for each $k$. 
    Moreover, if $\overline{z}>\underline{z}$, then any vertex is fully implementable. 
\end{proposition}

\begin{proof}[Proof of Proposition \ref{prop: joint distribution vertex appendix}]
    Let $\hat{S}$ be the set of all points $(v,d)$ which are expressed as in the statement for some indices $i$ and $j\leq i$.  

    Take any feasible pair of search and choice behavior $(v,d)$ such that $v_{1}\geq v_{2}\geq\cdots \geq v_{n}$. 
    Suppose it is a vertex. 
    Theorem \ref{thm: joint distribution} implies that $F^{*}$ induces $(v,d)$. 
    Moreover, since $(v,d)$ is a vertex, it must hold that a pure strategy induces $(v,d)$. 
    Under this strategy, by assumption, Pandora must open boxes in the ascending order of box indices. 
    Note that any realization from $F^{*}_{i}$ is either $\underline{z}$ or $h_{i}\in (\underline{z},\overline{x}]$. 
    
    Let $i$ be a box such that Pandora stops searching if the sampled value from the box $i$ is $\underline{z}$.
    Suppose that $i$ is the smallest such index. 
    Then, by construction, $v_{k}=d_{k}=0$ for all $k>i$. 
    Let $j$ be a box that Pandora chooses if all sampled values up to $i$ are $\underline{z}$. 
    Then, conditional on opening box $k<i$, Pandora samples $h_{k}$ with probability $1-p_{k}^{*}$, and therefore, $d_{k}=(1-p_{k}^{*})v_{k}^{*}$ for $k\neq j$. 
    If $k=j$, then we have $d_{j}=(1-p_{j}^{*})v_{j}^{*}+p_{i}^{*}v_{i}^{*}$. 
    Finally, by the choice of index $i$, we have $v_{k}=v_{k}^{*}$ for each $k\leq i$, and therefore, we obtain $(v,d)\in \hat{S}$. 

    It remains to check the converse, i.e., any point $(v,d)\in\hat{S}$ is a vertex. 
    Let $\tilde{S}$ be the union of all points in $\hat{S}$ and its permutations. 
    Suppose that $(v,d)$ has the expression with indices $i$ and $j\leq i$.
    Since the above paragraph implies that $S^{*}$ is the convex hull of $\tilde{S}$, it is sufficient to show that there exists a weight $(r,s)\in\mR^{n}\times \mR^{n}$ under which $(v,d)$ uniquely maximizes a linear function $r\cdot v + s\cdot d$ among all those points. 

    Specifically, we consider the following weight: for each $k$,
    \begin{align*}
        r_{k} &= 
        \begin{cases}
            (1-p_{k}^{*})10^{i-k} \quad &\text{if} \quad k\leq i,\\
            -\infty \quad &\text{otherwise},
        \end{cases} \\
        s_{k} &= 
        \begin{cases}
            \varepsilon \quad &\text{if} \quad k=j,\\
            0 \quad &\text{otherwise},
        \end{cases}
    \end{align*}
    where $\varepsilon\in (0,1)$ is a small positive number. 
    Let $(u,e)$ be any maximizer in $\tilde{S}$ with some search order $\pi$ and two indices $i^{*}$ and $j^{*}$. 
    First, it is clear that $u_{k}=0$ must hold for all $k>i$. 
    % Second, if $\pi(k)\neq k$ for some box $k$, then, 
    % \begin{align*}
    %     r\cdot u+s\cdot e
    %     \leq \sum_{k=1}^{i}10^{i-k}v_{\pi(k)}^{*} + \varepsilon 
    %     < \sum_{k=1}^{i}10^{i-k}v_{k}^{*}
    %     < r\cdot v + s\cdot d,
    % \end{align*}
    % where the second inequality follows from $\varepsilon$ being sufficiently small and $10^{i-k}$ being strictly decreasing in $k$. 
    Second, for $k<l\leq i$, we have
    \begin{align*}
        r_{k}(1-p_{l}^{*})-r_{l}(1-p_{k}^{*})
        =(1-p_{k}^{*})(1-p_{l}^{*})
        \left(10^{i-k}-10^{i-l}\right)>0.
    \end{align*}
    Since $r_{k}>0$ for all $k\leq i$, this adjacent-interchange comparison implies that $r\cdot u$ is maximized by opening boxes $1,\dots,i$ in this order and continuing after every low outcome.
    As $\tilde{S}$ is finite, sufficiently small $\varepsilon$ ensures that any maximizer of $r\cdot u+s\cdot e$ also maximizes $r\cdot u$.
    Therefore, $\pi(k)=k$ for all $k\leq i$. 
    Finally, given these observations, it is trivial to see that $j^{*}=j$ must hold. 
    Therefore, $(u,e)=(v,d)$, which shows that $(v,d)$ is a unique maximizer, hence a vertex.

    % Full implementation
    It remains to show that vertexes are fully implementable. 
    Suppose that $\overline{z}>\underline{z}$. 
    An analogous argument as in the proof of Lemma \ref{lem: Lemma 1 general} implies that $\overline{z}_{i}>\underline{z}_{i}$ if and only if $F_{i}^{0}(\underline{z}_{i}-)>0$. 
    Fix any $l_{k}$ which is close enough to $\underline{z}=\underline{z}_{i}$ such that $F_{k}^{0}(l_{k})>0$.
    Let $(v,d)$ be the vertex represented by the identity order $\pi(k)=k$ and two indices $i$ and $j\leq i$.  

    % Prep
    We start with one preparation.
    Let $F^{*}$ be the binary-support CDF profile which we define in Lemma \ref{lem: Lemma 1 general}. 
    Then, consider the difference
    \begin{align*}
        g^{*}_{k}(x_{k})
        = \int_{\underline{x}}^{x_{k}}F_{k}^{0}(y_{k})dy_{k}
        -\int_{\underline{x}}^{x_{k}}F_{k}^{*}(y_{k})dy_{k}.
    \end{align*}
    The feasibility of $F^{*}$ means that $g^{*}_{k}(x_{k})\geq 0$ for all $x_{k}$.  
    If $g^{*}_{k}(\underline{z})=0$, then the definition of the reservation values implies that the reservation value under the prior also equals $\underline{z}$, which contradicts $\overline{z}>\underline{z}$. 
    Therefore, we must have $g^{*}_{k}(\underline{z})>0$. 
    Since $g^{*}_{k}$ is a continuous function, we can find a point $m_{k}>\underline{z}$ close to $\underline{z}$ such that $g^{*}_{k}(x_{k})> 0$ for all $x_{k}\in [l_{k},m_{k}]$: note that we clearly have $g^{*}_{k}(x_{k})>0$ for all $x_{k}\in [l_{k},\underline{z})$. 
    Then, define $\overline{g}_{k}>0$ to be the minimum value of $g^{*}_{k}(x_{k})$ in this region, which exists by continuity. 
    Fix any such $l_{k}$ and $m_{k}$, and thus $\overline{g}_{k}$. 

    % Construction 
    As in the proof of Proposition \ref{prop: optimal}, we construct a CDF profile. 
    Consider the following function which is a perturbation of $F^{*}$ in Lemma \ref{lem: Lemma 1 general}: 
    \begin{align*}
    F_{k}(y_{k}; \varepsilon_{k}, \Delta_{k}) = 
    \begin{cases}
        0 & \quad \text{if} \quad y_{k}<l_{k}, \\
        \eta_{k}(\varepsilon_{k},\Delta_{k}) & \quad \text{if} \quad l_{k}\leq y_{k}<\underline{z}+\varepsilon_{k},\\
        p^{*}_{k}-\Delta_{k} & \quad \text{if} \quad \underline{z}+\varepsilon_{k}\leq y_{k}< m_{k},\\
        p^{*}_{k} & \quad \text{if} \quad m_{k}\leq y_{k}< h_{k},\\
        1 & \quad \text{if} \quad h_{k} \leq y_{k},
    \end{cases}
    \end{align*} 
    where $\eta_{k}(\varepsilon_{k},\Delta_{k})\geq 0$ is defined so that the function has mean equal to the prior mean $\mu_{k}$. 
    Since $p_{k}^{*}>0$, this is a well-defined CDF for small $\varepsilon_{k}>0$ and $\Delta_{k}>0$. 
    Note that it coincides with $F_{k}^{*}$ at each point $y_{k}\geq m_{k}$. 
    In other words, we locally perturb $F^{*}$ around $\underline{z}$.

    % Feasible
    Here, we check that $F_{k}(y_{k}; \varepsilon_{k}, \Delta_{k})$ is feasible, namely a mean-preserving contraction of the prior. 
    To see this, consider the difference
    \begin{align*}
        g_{k}(x_{k}; \varepsilon_{k}, \Delta_{k}) 
        &= 
        \int_{\underline{x}}^{x_{k}}F_{k}^{0}(y_{k})dy_{k}
        -\int_{\underline{x}}^{x_{k}}F_{k}(y_{k}; \varepsilon_{k}, \Delta_{k})dy_{k}. 
    \end{align*}
    First, note that both $F_{k}(y_{k}; \varepsilon_{k}, \Delta_{k})$ and $F_{k}^{*}(y_{k})$ have the same mean $\mu_{k}$ and they coincide at any $y_{k}\geq m_{k}$. 
    Therefore, for any $x_{k}\geq m_{k}$, integration by parts shows
    \begin{align*}
        \int_{\underline{x}}^{x_{k}}F_{k}(y_{k}; \varepsilon_{k}, \Delta_{k})dy_{k} 
        % &= \int_{\underline{x}}^{\overline{x}}F_{k}(y_{k}; \varepsilon_{k}, \Delta_{k})dy_{k} - \int_{x_{k}}^{\overline{x}}F_{k}(y_{k}; \varepsilon_{k}, \Delta_{k})dy_{k} \\
        & = \overline{x}-\mu_{k} - \int_{x_{k}}^{\overline{x}}F_{k}(y_{k}; \varepsilon_{k}, \Delta_{k})dy_{k} \\
        & = \overline{x}-\mu_{k} - \int_{x_{k}}^{\overline{x}}F^{*}_{k}(y_{k})dy_{k} \\
        & = \int_{\underline{x}}^{x_{k}}F^{*}_{k}(y_{k})dy_{k}. 
    \end{align*}
    Since $F_{k}^{*}$ is feasible, this implies that $g_{k}(x_{k}; \varepsilon_{k}, \Delta_{k})=g_{k}^{*}(x_{k})\geq 0$ for all $x_{k}\geq m_{k}$. 
    Second, since $F_{k}(y_{k}; \varepsilon_{k}, \Delta_{k})=0$ for all $y_{k}<l_{k}$, the difference is clearly non-negative for any $x_{k}<l_{k}$. 
    Third, take any $x_{k}$ such that $l_{k}\leq x_{k}\leq m_{k}$. 
    Then, 
    \begin{align*}
        g_{k}(x_{k}; \varepsilon_{k}, \Delta_{k}) 
        &\geq \overline{g}_{k} + \int_{\underline{x}}^{x_{k}} (F^{*}_{k}(y_{k})-F_{k}(y_{k}; \varepsilon_{k}, \Delta_{k})) dy_{k} \\
        &\geq \overline{g}_{k} - (m_{k}-l_{k})\cdot p_{k}^{*}.
    \end{align*}
    Note that, by definition, $\overline{g}_{k}$ becomes higher when $m_{k}$ and $l_{k}$ become closer to $\underline{z}$. 
    Therefore, if $m_{k}$ and $l_{k}$ are set to be close enough to $\underline{z}$, the difference is positive for any $x_{k}\in [l_{k},m_{k}]$. 
    Summarizing these three arguments, the distribution we construct above is feasible if $l_{k}$ and $m_{k}$ are sufficiently close to $\underline{z}$.  

    % Paremeter Construction
    Now, fix any sufficiently small numbers $\varepsilon_{+}$ and $\varepsilon_{-}$ such that $\varepsilon_{+}>\varepsilon_{-}>0$ and $\underline{z}+\varepsilon_{+}<m_{k}$ for all $k$. 
    Then, as we see in the proof of Proposition \ref{prop: optimal}, one can analogously show that there exists a profile $(\hat{\Delta}_{k})_{k\leq i}$ such that the associated reservation values satisfy
    \begin{align*}
        z_{1}>z_{2}>\cdots >z_{i}>\underline{z}+\varepsilon_{+}. 
    \end{align*}
    Then, consider a CDF profile such that $F_{k}(y_{k})=F_{k}(y_{k};\varepsilon_{-},\hat{\Delta}_{k})$ for all $k\leq i$ such that $k\neq j$ and $F_{j}(y_{j})=F_{j}(y_{j};\varepsilon_{+},\hat{\Delta}_{j})$. 
    Note that, since $F_{k}(y_{k})=F_{k}^{*}(y_{k})$ for all $y_{k}\geq m_{k}$, we must have $z_{k}<m_{k}$ for each $k$; otherwise, the reservation value under $F^{*}_{k}$ also equals $z_{k}>\underline{z}$. 
    For all remaining boxes $k>i$, we set $F_{k}$ to be a distribution that is degenerate at prior mean $\mu_{k}$.  

    % Show full implementation
    Finally, we show that the CDF $F$ approximately implements the pair $(v,d)$.
    Let $(v',d')$ be any pair induced by $F$. 
    Again, any optimal search strategy must open boxes in the ascending order of indices. 
    Moreover, Pandora will never sample a value equal to the reservation value of some box, hence
    \begin{align*}
        v'_{k} 
        = \prod_{l<k} F_{l}(z_{k}),
    \end{align*}
    for each box $k$. 
    The construction of the CDF profile implies 
    \begin{align*}
        F_{l}(z_{k}) 
        \in  
        \begin{cases}
            [p^{*}_{l}-\hat{\Delta}_{l}, p_{l}^{*}] \quad & \text{if} \quad k\leq i,\\
            [0, \eta_{l}(\varepsilon_{l},\hat{\Delta}_{l})] \quad & \text{if} \quad k> i,
        \end{cases}
    \end{align*}
    for all $l$ and $k>l$. 
    Finally, note that, if $\varepsilon_{+}$ and $\varepsilon_{-}$ are arbitrarily small, then $\hat{\Delta}_{j}$ and $\eta_{j}(\varepsilon_{j},\hat{\Delta}_{j})$ can also be made arbitrarily small by construction. 
    Therefore, $v'$ is arbitrarily close to $v$. 
    % Likewise, we can compute that
    % \begin{align*}
    %     d'_{k}
    %     &= 
    %     \begin{cases}
    %         [1-(p_{k}^{*}-\hat{\Delta}_{k}-\eta_{k}(\varepsilon_{k},\hat{\Delta}_{k}))]v'_{k} & \text{ if } k\leq i \text{ and } k\neq j,\\
    %         [1-(p_{k}^{*}-\hat{\Delta}_{k}-\eta_{k}(\varepsilon_{k},\hat{\Delta}_{k}))]v'_{k}+(p_{i}^{*}-\hat{\Delta}_{i}-\eta_{i}(\varepsilon_{i},\hat{\Delta}_{i})) v'_{i} & \text{ if } k=j,
    %     \end{cases}
    % \end{align*}
    % and $d'_{k}\leq \eta_{i}(\varepsilon_{i},\hat{\Delta}_{i})$ for all $k>i$. 
    % Therefore, since $v'$ is arbitrarily close to $v$, we have that $d'$ is also arbitrarily close to $d$. 
    Couple the target strategy under $F^{*}$ and the perturbed profile using the same uniform draws.
    Their choices coincide unless $l_{k}$ or $m_{k}$ occurs for some $k\leq i$, so
    \begin{align*}
        \max_{k}|d'_{k}-d_{k}|
        \leq\sum_{k=1}^{i}
        \left[\eta_{k}(\varepsilon_{k},\hat{\Delta}_{k})
        +\hat{\Delta}_{k}\right]\longrightarrow0.
    \end{align*}
    Therefore, since $v'$ is arbitrarily close to $v$, we have that $d'$ is also arbitrarily close to $d$.
    Hence, $(v,d)$ is fully implementable. 
\end{proof}

%%%%%%%%%%%%%%%%%%%%%%%%%%%%%%%%%%%%%%%%%%%%%%%%%%%%%%%%%
%%%%%%%%%%%%%%%%%%%%%%%%%%%%%%%%%%%%%%%%%%%%%%%%%%%%%%%%%%%%%%%%%%%%%%%%%%%%%%
% Change theorem numbering to A.1, A.2, etc.
\renewcommand{\theproposition}{D.\arabic{proposition}}
\setcounter{proposition}{0}
\section{Proof of Theorem \ref{thm: price competition}}\label{sec:appendixd}

% Preliminary
Note that Pandora observes price realizations prior to search. 
Pandora's optimal strategy is then described as follows. 
Consider any feasible CDF profile $F$ and any equilibrium strategy profile under $F$. 
Suppose that a realized price vector is $t\in \mR^{n}$. 
Then, Pandora visits sellers in the order of \textit{price-adjusted reservation values}, $z_{i}-t_{i}$, where $z_{i}$ is the reservation value under $F_{i}$. 
Pandora chooses to stop search if a previously sampled value net of its price, $y_{i}-t_{i}$, exceeds the price-adjusted reservation value of every remaining box. 

% Remark
Recall that $F^{*}$ is a feasible CDF profile we construct in Lemma \ref{lem: Lemma 1 general} and Theorem \ref{thm: joint distribution}. 
Each distribution in the profile takes one of two possible values, $\underline{z}$ and some $h_{i}\geq \underline{z}$, and has the reservation value $\underline{z}$. 

% ----------------------------------
\begin{lemma}\label{lem: Bertrand competition}
    Suppose $\underline{z}_{i}=\underline{z}$ for all $i$. 
    Under the CDF profile $F^{*}$, every seller $i$ setting the lowest price $t_{i}=0$ is an equilibrium. 
\end{lemma}
\begin{proof}
    Suppose all sellers except seller $i$ set their prices to zero.
    Consider the case where seller $i$ posts a positive price $t_{i}>0$. 
    Since the price-adjusted reservation values of every seller $j\neq i$ equal $\underline{z}$ and that of seller $i$ equals $\underline{z}-t_{i}$, Pandora visits sellers $j\neq i$ earlier than seller $i$. 
    Here, recall that any realization from $F^{*}$ is weakly above $\underline{z}$. 
    Therefore, Pandora never visits seller $i$. 
    Hence, any deviation by any seller is not profitable. 
\end{proof}

% Remark
The next two lemmas demonstrate that every equilibrium pair lies in the set $S^{*}$ constructed in Theorem \ref{thm: joint distribution}. 
Lemma \ref{lem: price and search behavior} establishes the claim for equilibrium search behavior, and Lemma \ref{lem: price and pair} extends it to equilibrium pairs.

% ----------------------------------
\begin{lemma}\label{lem: price and search behavior}
    Suppose $\underline{z}_{i}=\underline{z}$ for all $i$. 
    Any equilibrium search behavior is in the set $V^{*}$ of all feasible search behaviors. 
\end{lemma}
\begin{proof}
    Consider any feasible CDF profile $F$ such that $F_{i}\in \mathcal{F}_{i}$ for all $i$. 
    Suppose it has an equilibrium and take any price realizations $t$. 
    Let the price-adjusted reservation values satisfy $z_{1}-t_{1}\geq z_{2}-t_{2}\geq \cdots\geq z_{n}-t_{n}$.
    As a preparation, we note that, for any $i$ and $j<i$, we have
    \begin{align*}
        F_{j}((z_{i}-t_{i})+t_{j}) 
        \leq F_{j}(z_{j})
        \leq F^{*}_{j}(\underline{z}),
    \end{align*}
    where the first inequality follows from $z_{j}-t_{j}\geq z_{i}-t_{i}$ and the second inequality follows from Lemma \ref{lem: Lemma 1 general}. 
    Note that if a sampled value from box $j$ is such that $y_{j}\leq (z_{i}-t_{i})+t_{j}$, we have $y_{j}-t_{j}\leq z_{i}-t_{i}$, and therefore, it is optimal for Pandora to continue search. 

    Take any non-zero weight $r\in \mR^{n}$. 
    Let $r^{+}_{i}=\max\{r_{i},0\}$.
    Note that a pure search strategy maximizes the weighted search length. 
    Therefore, assuming w.l.g that Pandora opens boxes in the order of indices under such pure strategy, the above inequality implies that the weighted search length at the price realization $t$ is bounded above by
    \begin{align*}
        \sum_{i=1}^{n} r^{+}_{i} \left\{\prod_{j<i} F_{j}((z_{i}-t_{i})+t_{j})\right\} 
        \leq \sum_{i=1}^{n} r^{+}_{i} \left\{\prod_{j<i} F^{*}_{j}(\underline{z})\right\}
        \leq \sum_{i\in I} r_{i} \left\{\prod_{j\in I: j<i} F^{*}_{j}(\underline{z})\right\},
    \end{align*}
    where we define $I$ to be the set of indices with non-negative weights.
    Therefore, the weighted search length is larger when the CDF profile is $F^{*}$ with no pricing, for any price realization. 
    Since Theorem \ref{thm: joint distribution} shows that $V^{*}$ is the marginal of $S^{*}$ on the space of search behaviors and is therefore a polytope, and $F^{*}$ induces all points in $V^{*}$, this inequality means that equilibrium search behavior under $F$ and any price realization must lie in $V^{*}$. 
\end{proof}

% Preparation
Recall that effective value of a CDF $F_{i}$ is defined as $\omega_{i}=\min\{y_{i},z_{i}\}$, where $y_{i}$ follows $F_{i}$ and $z_{i}$ is the reservation value under $F_{i}$. 
As we define price-adjusted reservation value, define the \textit{price-adjusted effective value} $\omega_{i}-t_{i}$.
Under a price realization $t$, Pandora chooses a box with the highest price-adjusted effective value. 
We use notation $H^{t}_{i}$ to represent the distribution of price-adjusted effective values. 
Then, we let $\overline{H}_{i}^{t}(\omega_{i}-t_{i})$ denote the probability that $\omega_{j}-t_{j} \leq \omega_{i}-t_{i}$ for all $j \neq i$. 

% ----------------------------------
\begin{lemma}\label{lem: price and pair}
    Suppose $\underline{z}_{i}=\underline{z}$ for all $i$. 
    Any equilibrium pair is in the set $S^{*}$ of all feasible pairs.
\end{lemma}
\begin{proof}
    The proof extends those in Lemmas \ref{lem: lower bound of choice}, \ref{lem: feasible joint distribution necessary conditions}, and \ref{lem: neccesary conditions imply implementability}.
    Consider a CDF profile $F$ and let $z$ be the induced profile of the reservation values.
    Take any profile of price realization $t$ in an equilibrium.
    Let $(v,d)$ be the induced pair.
    Assume w.l.g. $z_{1}-t_{1}\geq z_{2}-t_{2}\geq \cdots\geq z_{n}-t_{n}$.

    To simplify the exposition, we focus on the case in which Pandora adopts a pure strategy of opening boxes in the order of indices under price realization $t$.
    If Pandora adopts a mixed strategy, we can decompose it into pure strategies and later take the convex combination of all induced pairs, which is still in $S^{*}$ because $S^{*}$ is convex.

    First, we prove that $d_{i}\geq (1-p_{i}^{*})\cdot v_{i}$ for all $i$.
    The decision to open box $i$ is independent of its unobserved value $y_{i}$.
    When box $i$ is opened, every previously sampled value net of its price and every remaining price-adjusted reservation value are weakly below $z_{i}-t_{i}$.
    Thus, if $y_{i}>z_{i}$, we have $y_{i}-t_{i}>z_{i}-t_{i}$, and it is strictly optimal to stop search and choose box $i$.
    Hence,
    \begin{align*}
        d_{i}\geq (1-F_{i}(z_{i}))\cdot v_{i}\geq (1-p_{i}^{*})\cdot v_{i},
    \end{align*}
    where the last inequality follows from Lemma \ref{lem: Lemma 1 general}.

    Second, we show the other two inequalities that appear in Lemma \ref{lem: feasible joint distribution necessary conditions}.
    The same logic as in the proof of Lemma \ref{lem: feasible joint distribution necessary conditions} implies the second inequality in its statement, $\sum_{i\leq j}d_{i}\geq 1-v_{j+1}$ for all $j$.
    Then, for any box $i$ such that $v_{i}\neq 0$, we have 
    \begin{align*}
        \frac{v_{i+1}}{v_{i}}
        \leq F_{i}((z_{i+1}-t_{i+1})+t_{i})
        \leq F_{i}(z_{i}) \leq p^{*}_{i},
    \end{align*}
    where the second inequality uses $z_{i}-t_{i}\geq z_{i+1}-t_{i+1}$.
    Therefore, $p^{*}_{i}\cdot v_{i}\geq v_{i+1}$.
    Note that if $v_{i}=0$, then $v_{i+1}=0$, hence the same inequality.

    Finally, the above two observations and the proof of Lemma \ref{lem: neccesary conditions imply implementability} imply that the pair $(v,d)$ is feasible.
    Since this holds for any price realization and $S^{*}$ is convex, any equilibrium pair is in the set of all feasible pairs.
\end{proof}
\begin{proof}[Proof of Theorem \ref{thm: price competition}]
    Lemma \ref{lem: price and pair} implies that any equilibrium pair is a feasible pair.
    Hence, Lemma \ref{lem: Bertrand competition} and Theorem \ref{thm: joint distribution} imply that the set of equilibrium pairs coincides with the set of feasible pairs.
    Lemma \ref{lem: Bertrand competition} shows the latter statement of the theorem, completing the proof. 
\end{proof}

%%%%%%%%%%%%%%%%%%%%%%%%%%%%%%%%%%%%%%%%%%%%%%%%%%%%%%%%%%%%%%%%%%%%%%%%%%%%%%

%%%%%%%%%%%%%%%%%%%%%%%%%%%%%%%%%%%%%%%%%%%%%%%%%%%%%%%%%%%%%%%%%%%%%%%%%%%%%%
\singlespacing 
\bibliography{Reference}
\addcontentsline{toc}{section}{Reference}
%%%%%%%%%%%%%%%%%%%%%%%%%%%%%%%%%%%%%%%%%%%%%%%%%%%%%%%%%%%%%%%%%%%%%%%%%%%%%%%%%%%%%%%%%%%%
\end{document}